\pdfoutput=1
\documentclass[a4paper,11pt]{amsart}
\usepackage[utf8]{inputenc}
\usepackage[T1]{fontenc}
\usepackage[american]{babel}
\usepackage{xspace}
\usepackage{hyperref}
\usepackage{amssymb}
\usepackage{amsmath}
\usepackage{mathrsfs}  
\usepackage{tikz}
\usepackage{boxedminipage,verbatim,float,caption}
\usepackage{enumerate}
\usepackage{multicol}
\usepackage{fullpage}

\def\bN{\mathbb{N}}
\def\bQ{\mathbb{Q}}

\def\bI{\mathbb{I}}

\def\cA{\mathcal{A}}
\def\cB{\mathcal{B}}
\def\cC{\mathcal{C}}
\def\cD{\mathcal{D}}
\def\cF{\mathcal{F}}

\def\cL{\mathcal{L}}
\def\cM{\mathcal{M}}
\def\cN{\mathcal{N}}

\def\cR{\mathcal{R}}

\def\cT{\mathcal{T}}

\def\sL{\mathscr{L}}

\newcommand{\XP}{\textsf{XP}\xspace}
\newcommand{\W}{\textsf{W}}
\newcommand{\ETH}{\textsf{ETH}\xspace}
\newcommand{\NP}{\textsf{NP}\xspace}

\renewcommand{\mid}{\, : \, }

\def\wrtn{{w.r.t.}\xspace}

\def\ien{{i.e.}\xspace}

\def\eg{{e.g.}\xspace}

\def\Rep#1#2{\cR_{#2}^{#1}}
\def\rep#1#2{\mathsf{rep}_{#2}^{#1}}
\def\equi#1#2{\equiv_{#2}^{#1}}

\def\comp#1{\overline{#1}}
\newcommand{\opt}{\mathsf{opt}}
\newcommand{\best}{\mathsf{best}}

\newcommand{\reduce}{\mathsf{reduce}}
\newcommand{\cuts}{\mathsf{ccut}}
\newcommand{\cc}{\mathsf{cc}}
\newcommand{\w}{\mathsf{w}}

\newcommand{\acy}{\mathsf{acy}}

\newcommand{\nec}{\mathsf{nec}}
\newcommand{\snec}{\mathsf{s\text{-}nec}}
\newcommand{\mim}{\mathsf{mim}}
\newcommand{\rw}{\mathsf{rw}}
\newcommand{\Qrw}{\mathsf{rw}_\bQ}
\newcommand{\cw}{\mathsf{cw}}
\newcommand{\mw}{\mathsf{mw}}
\newcommand{\f}{\mathsf{f}}

\newcommand{\tw}{\mathsf{tw}}

\newtheorem{theorem}{Theorem}[section]
\newtheorem{lemma}[theorem]{Lemma}

\newtheorem{corollary}[theorem]{Corollary}
\newtheorem{claim}{Claim}[theorem]
\newtheorem{fact}[theorem]{Fact}
\newtheorem{observation}[claim]{Observation}
\newtheorem{definition}[theorem]{Definition}

\renewcommand{\leq}{\leqslant}

\begin{document}
	
	\hypersetup{
		pdftitle={More applications of the d-neighbor equivalence: acyclicity and connectivity constraints},
		pdfauthor={Benjamin Bergougnoux and Mamadou Moustapha Kant\'e}
	}
	\title{More applications of the \texorpdfstring{$D$}{D}-neighbor equivalence: acyclicity and connectivity constraints}
	
	\author{Benjamin Bergougnoux}
	\address{University of Warsaw, Poland}
	\email{benjamin.bergougnoux@gmail.com}	
	
	\author{Mamadou Moustapha Kant\'e}
	\address{Université Clermont Auvergne, LIMOS, CNRS, Aubière, France.}
	\email{mamadou.kante@uca.fr}	
	
	\thanks{This work is supported by French Agency for Research under the GraphEN project (ANR-15-CE40-0009). An extended abstract appeared in the proceedings of
		the $27^{\rm th}$ European Symposium on Algorithms 2019.}
	
	
	\begin{abstract} 
		In this paper, we design a framework to obtain efficient algorithms for several problems with a global constraint (acyclicity or connectivity) such as {\sc Connected Dominating Set},
		{\sc Node Weighted Steiner Tree}, {\sc Maximum Induced Tree}, {\sc Longest Induced Path}, and {\sc Feedback Vertex Set}.
		We design a meta-algorithm that solves all these problems and whose running time is upper bounded by  $2^{O(k)}\cdot n^{O(1)}$, $2^{O(k \log(k))}\cdot n^{O(1)}$, $2^{O(k^2)}\cdot n^{O(1)}$, and $n^{O(k)}$ where $k$ is respectively the clique-width, $\mathbb{Q}$-rank-width, rank-width, and maximum induced matching width of a given decomposition.
		Our approach simplifies and unifies the known algorithms for each of the parameters and its running time matches asymptotically also the running times of the best known algorithms for basic \textsf{NP}-hard problems such as {\sc Vertex Cover} and {\sc Dominating Set}.
		Our framework is based on the $d$-neighbor equivalence defined in [B.\ Bui-Xuan, J.\ A.\ Telle, and M.\ Vatshelle, {\it Theoret. Comput. Sci.}, (2013), pp.\ 66--76] and the rank-based approach introduced in [H.\ L.\ Bodlaender, M.\ Cygan, S.\ Kratsch, and J.\ Nederlof, {\it Inform. and Comput.}, 243 (2015), pp.\ 86--111].
		The results we obtain highlight the importance of the $d$-neighbor equivalence relation on the algorithmic applications of width measures.
		We also prove that our framework could be useful for ${\sf W}[1]$-hard problems parameterized by clique-width such as {\sc Max Cut} and {\sc Maximum Minimal Cut}.
		For these latter problems, we obtain $n^{O(k)}$, $n^{O(k)}$, and $n^{2^{O(k)}}$ time algorithms where $k$ is respectively the clique-width, the $\mathbb{Q}$-rank-width, and the rank-width of the input graph.
	\end{abstract}
	
	\maketitle
	
	\section{Introduction}

\emph{Tree-width} is one of the most well-studied graph parameters in the graph algorithm community, due to its numerous structural and algorithmic properties.
Nevertheless, despite the broad interest on tree-width, only sparse graphs can have bounded tree-width.
But many \NP-hard problems are tractable on dense graph classes.
For many graph classes, this tractability can be explained through other width measures.
The most remarkable ones are certainly clique-width \cite{CourcelleO00}, rank-width \cite{Oum05a}, and maximum induced matching width ({a.k.a.} mim-width) \cite{Vatshelle12}.

We obtain most of these parameters through the well-known notion of \emph{layout} (a.k.a. \emph{branch-decomposition}) introduced in \cite{RobertsonS91}.
A layout of a graph $G$ is a tree $T$ whose leaves are in bijection with the vertices of $G$.
Every edge $e$ of the layout is associated with a vertex bipartition of $G$ through the two connected components obtained by the removal of $e$.
Given a symmetric function $\f:2^{V(G)}\to \bN$, one can associate with each layout $T$ a measure, usually called $\f$-width, defined as the maximum  $\f(A)$ over all the vertex bipartitions $(A,\comp{A})$ of $V(G)$ associated with the edges of $T$.
For instance, rank-width is defined from the function $\f(A)$ which corresponds to the rank over $GF(2)$ of the adjacency matrix between the vertex sets $A$ and $\comp{A}$; if we take the rank over $\bQ$, we obtain a variant of rank-width introduced in \cite{OumSV13}, called $\bQ$-rank-with.
For mim-width, $\f(A)$ is the maximum size of an induced matching in the bipartite graph between $A$ and $\comp{A}$.

These other width measures have a modeling power strictly stronger than the modeling power of tree-width.
For example, if a graph class has bounded tree-width, then it has bounded clique-width \cite{CourcelleO00}, but the converse is false as cliques have clique-width at most $2$ and unbounded tree-width.
While ($\bQ$-)rank-width has the same modeling power as clique-width, mim-width has the strongest one among all these width measures and is even bounded on interval graphs~\cite{BelmonteV13}.
Despite their generality, a lot of \NP-hard problems admit polynomial time algorithms when one of these width measures is fixed.
But, designing efficient algorithms with these width measures is known to be harder than with tree-width.

Concerning their computations, it is not known whether the clique-width (resp., mim-width) of a graph can be approximated within a constant factor in time $\f(k)\cdot n^{O(1)}$ (resp., $n^{\f(k)}$) for some function $\f$.
However, for ($\bQ$-)rank-width, there is a $2^{3k}\cdot n^{4}$ time algorithm that, given a graph $G$ as input and $k\in \bN$, either outputs a decomposition for $G$ of ($\bQ$-)rank-width at most $3k+1$ or confirms that the ($\bQ$-)rank-width of $G$ is more than $k$ \cite{OumSV13,OumS06}.

Finding efficient algorithms parameterized by one of these width measures is by now a standard exercise for problems based on local constraints \cite{Bui-XuanTV13,TelleP97}.
In contrast, the task is quite complicated for problems involving a global constraint, \eg, connectivity or acyclicity.  For a long time, our knowledge on the
parameterized complexity of this latter kind of problem, with as parameters the common width measures, was quite limited even for tree-width.  For a while, the
parameterized complexity community used to think that for problems involving global constraints the naive $k^{O(k)}\cdot n^{O(1)}$ time algorithm, $k$ being the
tree-width of the input graph, could not be improved.
But, quite surprisingly, in 2011, Cygan et al. introduced in \cite{CyganNPPRW11} a technique called \textit{Cut \& Count} to design
Monte Carlo $2^{O(k)}\cdot n^{O(1)}$ time algorithms for a wide range of problems with global constraints, including \textsc{Hamiltonian Cycle}, \textsc{Feedback
	Vertex Set}, and \textsc{Connected Dominating Set}.  Later, Bodlaender et al. proposed in \cite{BodlaenderCKN15} a general toolkit, called the \textit{rank-based
	approach}, to design deterministic $2^{O(k)}\cdot n$ time algorithms to solve a wider range of problems.

In a recent paper \cite{BergougnouxK19}, the authors adapted the rank-based approach of \cite{BodlaenderCKN15} to obtain $2^{O(k)}\cdot n$ time algorithms, $k$
being the clique-width of a given decomposition, for many problems with a global constraint, \eg, \textsc{Connected Dominating Set} and \textsc{Feedback Vertex Set}.

Unlike tree-width and clique-width, algorithms parameterized by rank-width and mim-width for problems with a global constraint were not investigated, except for some special cases such as \textsc{Feedback Vertex Set} \cite{GanianH10,JaffkeKT20b} and \textsc{Longest Induced Path} \cite{JaffkeKT20a}.

One successful way to design efficient algorithms with these width measures is through the notion of \textit{$d$-neighbor equivalence}.
This concept was introduced by Bui-Xuan, Telle, and Vatshelle in \cite{Bui-XuanTV13}.
Formally, given $A\subseteq V(G)$ and an integer $d\geq 1$, two sets $X,Y\subseteq A$ are $d$-neighbor equivalent over $A$ if, for all $v\in V(G)\setminus A$, we have $\min(d,|N(v)\cap X|)=\min(d,|N(v)\cap Y|)$, where $N(v)$ is the set of neighbors of $v$ in $G$.
Notice that $X$ and $Y$ are $1$-neighbor equivalent over A if and only if both have the same neighborhood in $V(G)\setminus A$.

The $d$-neighbor equivalence gives rise to a width measure, called in this paper \emph{$d$-neighbor-width}.  This width measure, based also on layouts, is defined
from the function $\snec_d(A)$ which corresponds to the maximum number of equivalence classes of the $d$-neighbor equivalence over $A$ and $V(G)\setminus A$.
It is worth noticing that the boolean-width of a layout introduced in \cite{Bui-XuanTV11} corresponds to the binary logarithm of the $1$-neighbor-width.

These notions were used by Bui-Xuan, Telle, and Vatshelle in \cite{Bui-XuanTV13} to design efficient algorithms for the family of problems called \textsc{$(\sigma,\rho)$-Dominating Set} problems.
This family of problems was introduced by Telle and Proskurowski in \cite{TelleP97}.
Given a pair $(\sigma,\rho)$ of finite or co-finite subsets of $\bN$ and a graph $G$, a $(\sigma,\rho)$\emph{-dominating set} of $G$ is a subset $D$ of $V(G)$ such that, for each vertex $x\in V(G)$, the number of neighbors of $x$ in $D$ is in $\sigma$ if $x\in D$ and in $\rho$ otherwise.
A problem is a  \textsc{$(\sigma, \rho)$-Dominating  Set} problem if it consists in finding a minimum (or maximum) $(\sigma,\rho)$\emph{-dominating set}.
For instance, the \textsc{Dominating Set} problem asks for the computation of a minimum $(\bN,\bN\setminus \{0\})$-dominating set.
Many \NP-hard problems based on local constraints belong to this family; see \cite[Table 1]{Bui-XuanTV13}.

Bui-Xuan, Telle, and Vatshelle~\cite{Bui-XuanTV13} designed a meta-algorithm that, given a rooted layout $\cL$, solves any \textsc{$(\sigma,\rho)$-Dominating Set} problem in time $\snec_d(\cL)^{O(1)}\cdot n^{O(1)}$ where $d$ is a constant depending on the considered problem.
The known upper bounds on $\snec_d(\cL)$ (see Lemma \ref{lem:comparemim}) and the algorithm of \cite{Bui-XuanTV13} give efficient algorithms to solve any
\textsc{$(\sigma,\rho)$-Dominating Set} problem, with parameters tree-width, clique-width, ($\bQ$-)rank-width, and mim-width.
The running times of these algorithms are given in Table \ref{tab:result}.

\renewcommand{\arraystretch}{1.5}
\begin{table}[t]\vspace*{-.3pc}
	\footnotesize
	\centering
	\caption{Upper bounds on  $\snec_d(\cL)^{O(1)}\cdot n^{O(1)}$ with $\cL$ a layout and $d$ a constant.}     \label{tab:result}
	\begin{tabular}{c|c|c|c|c}
		\hline
		Tree-width &Clique-width & Rank-width & $\bQ$-rank-width  & Mim-width \\
		\hline
		$2^{O(k)}\cdot n^{O(1)}$ &$2^{O(k)}\cdot n^{O(1)}$ & $2^{O(k^2)}\cdot n^{O(1)}$ & $2^{O(k  \log( k))}\cdot n^{O(1)}$
		& $n^{O(k)}$\\
		\hline
	\end{tabular}\vspace*{-.3pc}
\end{table}

\subsection*{Our contributions}

In this paper, we design a framework based on the $1$-neighbor equivalence  relation and some ideas from the rank-based approach of \cite{BodlaenderCKN15} to design efficient algorithms for many problems involving a connectivity constraint.
This framework provides tools to reduce the size of the sets of partial solutions we compute at each step of a dynamic programming algorithm. We prove that many ad hoc
algorithms for these problems can be unified into a single meta-algorithm that is almost the same as the one from \cite{Bui-XuanTV13} computing a $(\sigma,\rho)$-dominating set.

We use our framework to design a meta-algorithm that, given a rooted layout $\cL$, solves any connectivity variant (a solution must induce a
connected graph) of a \textsc{$(\sigma, \rho)$-Dominating Set} problem.  This includes some well-known problems such as \textsc{Connected Dominating Set},
\textsc{Connected Vertex Cover,} or \textsc{Node Weighted Steiner Tree}.  The running time of our algorithm is polynomial in $n$ and $\snec_d(\cL)$, with $d$ a
constant that depends on $\sigma$ and $\rho$.  Consequently, each connectivity variant of a \textsc{$(\sigma, \rho)$-Dominating Set} problem admits algorithms
with the running times given in Table \ref{tab:result}.

We introduce some new concepts to deal with acyclicity.  We use these concepts to deal with the \emph{AC} variants\footnote{AC stands for ``acyclic and connected.''} (a
solution must induce a tree) of \textsc{$(\sigma, \rho)$-Dominating Set} problems.  Both \textsc{Maximum Induced Tree} and \textsc{Longest Induced Path} are \emph{AC}
variants of \textsc{$(\sigma, \rho)$-Dominating Set} problems.  We prove that we can modify slightly the meta-algorithm for connectivity constraints so that it can solve
these AC variants in the running times given in Table \ref{tab:result}.  To obtain these results, we rely heavily on the $d$-neighbor equivalence.  However, we were not
able to upper bound the running time of this modified meta-algorithm by a polynomial in $n$ and $\snec_d(\cL)$ for some constant~$d$.

We then reduce any acyclic variant (a solution must induce a forest) of a \textsc{$(\sigma, \rho)$-Dominating Set} problem to solving the AC-variant of the same problem.
In particular, this shows that we can use the algorithm for \textsc{Maximum Induced Tree} to solve the \textsc{Feedback Vertex Set} problem.

Up to a constant in the exponent, the running times of our meta-algorithms and their algorithmic consequences match those of the best known algorithms for basic
problems such as \textsc{Vertex Cover} and \textsc{Dominating Set} \cite{Bui-XuanTV13,OumSV13}.
Moreover, the $2^{O(k)}\cdot n^{O(1)}$ time algorithms we obtain for clique-width are optimal under the well-known exponential time hypothesis (\ETH)~\cite{ImpagliazzoP01}.
That is, unless \ETH fails, there are no $2^{o(k)}\cdot n^{O(1)}$ time algorithms, $k$ being the clique-width of a given decomposition, for the \NP-hard problems considered in this paper.
This follows from the facts that the clique-width of a graph is at most  its number of vertices and that (under well-known Karp reduction \cite{ImpagliazzoP01,ImpagliazzoPZ01}) those problems do not admit a $2^{o(n)}\cdot n^{O(1)}$ time algorithm unless \ETH fails.

Our results reveal that the $d$-neighbor equivalence relation can be used for problems which are not based only on local constraints.
This highlights the importance and the generalizing power of this concept on many width measures: for many problems and many width measures, one obtains the ``best''
algorithms by using the upper bounds on $\snec_{d}(\cL)$ (see Lemma \ref{lem:compare}).
We prove that the $d$-neighbor equivalence relation could be also useful for problems that are $\W[1]$-hard parameterized by clique-width. We provide some
evidence for this potentiality by showing  that, given an $n$-vertex graph and a rooted layout $\cL$, one can use the $n$-neighbor equivalence to solve, with almost the same
algorithm as for solving the \textsc{Independent Set} problem,

\begin{enumerate}
	\item \textsc{Max Cut} in time $\snec_n^{O(1)}(\cL)\cdot n^{O(1)}$.
	This algorithm gives the best known algorithms parameterized by clique-width, $\bQ$-rank-width, and rank-width for \textsc{Max Cut}.
	\item \textsc{Maximum Minimal Cut}, a variant of \textsc{Max Cut}, with two connectivity constraints discussed in \cite{DuarteLPSS19,EtoHKK19}. This is done by
	combining our algorithm for \textsc{Max Cut} and our framework.
\end{enumerate}

Finally, we also show the wide applicability of our framework by explaining how to use it for solving locally checkable partitioning problems with multiple global constraints.

\subsection*{Our approach}
To solve the considered problems, our algorithms do a bottom-up traversal of a given layout $\cL$ of the input graph $G$ and at each step we compute a set of partial solutions.
In our case, the steps of our algorithms are associated with the vertex bipartitions $(A,\comp{A})$ induced by the edges of a layout and the partial solutions are subsets of $A$.

Let us explain our approach for the variant of \textsc{$(\sigma,\rho)$-Dominating Set} with the connected and AC variants of the \textsc{Dominating Set} problem.
At each step, our algorithms compute, for each pair $(R,R')$ where $R$ (resp., $R'$) is a $1$-neighbor equivalence class of $A$ (resp., $\comp{A}$), a set of partial solutions $\cA_{R,R'}\subseteq R$.
The way we compute these sets guarantees that the partial solutions in $\cA_{R,R'}$ will be completed with sets in $R'$.
Consequently, we have information about how we will complete our partial solutions since every $Y\in R'$ has the same neighborhood in $A$.

To deal with the local constraint of these problems, namely the domination constraint, we use the ideas of  Bui-Xuan, Telle, and Vatshelle \cite{Bui-XuanTV13}.
For each pair $(R,R')$, let us say that $X\subseteq A$ is \emph{coherent} with $(R,R')$ if (1) $X \in R$ and (2) $X\cup Y$ dominates $A$ in the graph $G$ for every $Y\in R'$.
To compute a minimum dominating set, Bui-Xuan, Telle, and Vatshelle proved that it is enough to keep, for each pair $(R,R')$, a partial solution $X$ of minimum weight that is coherent with $(R,R')$.
Intuitively, if a partial solution $X$ that is coherent with $(R,R')$ could be completed into a dominating set of $G$ with a set $Y\in R'$, then it is the case for every partial solution coherent with $(R,R')$.
This is due to the fact that any pair of sets in $R$ (resp., $R'$) dominates the same vertices in $\comp{A}$ (resp., $A$).

To solve the connectivity variant, we compute, for each $(R,R')$, a set $\cA_{R,R'}$ of partial solutions coherent with $(R,R')$.  Informally, $\cA_{R,R'}$ has to be as
small as possible, but if a partial solution coherent with $(R,R')$ leads to a minimum connected dominating set, then $\cA_{R,R'}$ must contain a similar partial solution.
To deal with this intuition, we introduce the relation of \emph{$R'$-representativity} between sets of partial solutions.  We say that $\cA^\star$ $R'$-\emph{represents}
a set $\cA$ if, for all sets $Y\in R'$, we have $\best(\cA,Y)=\best(\cA^\star, Y)$ where $\best(\cB,Y)$ is the minimum weight of a set $X\in \cB$ such that $G[X\cup Y]$
is connected.

The main tool of our framework is a function $\reduce$ that, given a set of partial solutions $\cA$ and a 1-neighbor equivalence class $R'$ of $\comp{A}$, outputs a subset of $\cA$ that $R'$-represents $\cA$ and whose size is upper bounded by $\snec_1(\cL)^2$.  To design this function, we use ideas from the rank-based approach of \cite{BodlaenderCKN15}.
That is, we define a small matrix $\cC$ with $|\cA|$ rows and $\snec_1(\cL)^2$ columns.
Then, we show that a basis of minimum weight of the row space of $\cC$ corresponds to an $R'$-representative set of $\cA$.
Since $\cC$ has $\snec_1(\cL)^2$ columns, the size of a basis of $\cC$ is smaller than $\snec_1(\cL)^2$.
By calling $\reduce$ after each computational step, we keep the sizes of the sets of partial solutions polynomial in $\snec_1(\cL)$.
Besides, the definition of $R'$-representativity guarantees that the set of partial solutions computed for the root of $\cL$ contains a minimum connected dominating set.

For the AC variant of dominating set, we need more information in order to deal with the acyclicity.
We obtain this extra information by considering that $R$ (resp., $R'$) is a $2$-neighbor equivalence class over $A$ (resp., $\comp{A}$).
This way, for all sets $X\subseteq A$, the set $X^{+2}$ of vertices in $X$ that have at least two neighbors in some $W\in R'$ corresponds exactly to the set of vertices in $X$ that have at least two neighbors in $Y$, for all $Y\in R'$.
The vertices in $X^{+2}$ play a major role in the acyclicity constraint because if we control them, then we control the number of edges between $X$ and the sets $Y\in R'$.
To obtain this control, we prove in particular that if there exists $Y\in R'$ such that $X\cup Y$ is a forest, then $|X^{+2}|\leq 2\mim(A)$ where $\mim(A)$ is the size of an induced matching in the bipartite graph between $A$ and $\comp{A}$. Consequently, the size of  $X^{+2}$ must be small, otherwise the partial solution $X$ can be discarded.

We need also a new notion of representativity.
We say that $\cA^\star$ $R'$-\emph{ac-represents} a set $\cA$ if, for all sets $Y\in R'$, we have $\best^\acy(\cA,Y)=\best^\acy(\cA^\star, Y)$ where $\best^\acy(\cB,Y)$ is the minimum weight of a set $X\in \cB$ such that $G[X\cup Y]$ is a tree.

As for the $R'$-representativity, we provide a function that, given a set of partial solutions $\cA$ and a $2$-neighbor
equivalence class $R'$ of $\comp{A}$, outputs a small subset $\cA^\star$ of $\cA$ that $R'$-ac-represents $\cA$.
Unfortunately, we were not able to upper bound the size of $\cA^\star$ by a polynomial in $n$ and $\snec_d(\cL)$ (for some constant $d$).
Instead, we prove that, for clique-width, rank-width, $\bQ$-rank-width, and mim-width, the size of $\cA^\star$ can be upper bounded by, respectively, $2^{O(k)}\cdot n$, $2^{O(k^2)}\cdot n$, $2^{O(k\log(k))}\cdot n$, and $n^{O(k)}$.
The key to compute $\cA^\star$ is to decompose $\cA$ into a small number of sets $\cA_1\dots,\cA_\ell$, called \emph{$R'$-consistent}, where the notion of $R'$-ac-representativity matches the notion of $R'$-representativity.
More precisely, any $R'$-representative set of an $R'$-consistent set $\cA$ is also an $R'$-ac-representative set of $\cA$.

To compute an $R'$-ac-representative set of $\cA$ it is then enough to compute an $R'$-representative set for each $R'$-consistent set in the decomposition of $\cA$.
The union of these $R'$-representative sets is an $R'$-ac-representative set of $\cA$.
Besides the notion of representativity, the algorithm for the AC variant of \textsc{Dominating Set} is very similar to the one for finding a minimum connected dominating set.

It is worth noticing that we cannot use the same trick as in \cite{BodlaenderCKN15} to ensure the acyclicity, that is, counting the number of edges induced by the partial solutions.
Indeed, we would need to differentiate at least $n^{k}$ partial solutions (for any parameter $k$ considered in Table \ref{tab:result}) in order to update this information.
We give more explanation on this statement at the beginning of section \ref{sec:maxinducedtree}.

For \textsc{Max Cut}, we use the $n$-neighbor equivalence, with $n$ the number of vertices of $G$.
For every $X,W\subseteq A$, we prove that if $X$ and $W$ are $n$-neighbor equivalent over $A$, then, for every $Y\subseteq \comp{A}$, the number of edges between $X$ and $Y$ are the same as the number of edges between $W$ and $Y$.
It follows that if the number of edges between $X$ and $A\setminus X$ is greater than the number of edges between $W$ and $A\setminus W$, then $X\cup Y$ is a better solution than $W\cup Y$  for every $Y\subseteq \comp{A}$.
Consequently, it is sufficient to keep a partial solution for each $n$-neighbor equivalence class over $A$.
This observation leads to an $\snec_{n}(\cL)^{O(1)}\cdot n^{O(1)}$ time algorithm for \textsc{Max Cut}.

\subsection*{Relation to previous works} One general question in parameterized complexity is, given a function $f:\bN\to \bN$ and a parameter $k$, to identify
problems which admit algorithms running in time $f(k)\cdot n^{O(1)}$  for any input of size $n$ (or $n^{f(k)}$ if the problem is known to be $W[i]$-hard for
some $i\geq 1$), and the literature on parameterized complexity is full of
exemplified functions and problems (see, for instance, the book \cite{Cyganetal15} with a bunch of upper and lower bounds on running times of algorithms, under
(S)ETH). Among the exemplified functions, the one $[k\mapsto 2^{O(k)}]$ is probably the most considered one, maybe  because it allows us to solve efficiently the same problems in
instances whose parameters are logarithmic in their sizes. Indeed, there are a lot of
problems, parameterized by tree-width or clique-width, that were identified to admit algorithms running in time $2^{O(k)}\cdot n^{O(1)}$ and that are essentially tight under ETH (see, for instance, \cite{Bui-XuanTV13,Cyganetal15}). The authors of \cite{BergougnouxK19} and \cite{BodlaenderCKN15} pushed further this set of problems by showing that many problems with connectivity or acyclicity constraints admit also $2^{O(k)}\cdot n^{O(1)}$ time algorithms, parameterized by tree-width or clique-width.
We can transform in polynomial time a decomposition of clique-width $k$ into a layout of $d$-neighbor-width at most $2^{kd}$ (Theorem~\ref{thm:rao} and Lemma~\ref{lem:compare}), and given a vertex separator $S$ of size $k$, the number of $d$-neighbor equivalence classes over $S$ (resp., $V(G)\setminus S$) is upper bounded by $2^k$ (resp., $(d+1)^{k}$).
For this reason, we can consider our framework as a generalization of the rank-based approach considered in \cite{BergougnouxK19} and \cite{BodlaenderCKN15} as it allows us to consider more parameters and push further the set of problems that can be solved efficiently, under ETH, when
considering these parameters, for instance, mim-width and rank-width. Notice, however, that the constants in the running times of the algorithms in  \cite{BergougnouxK19,BodlaenderCKN15} are better than those of our algorithms.
For instance, the authors in \cite{BergougnouxK19} obtained a $15^k\cdot 2^{(\omega+1)\cdot k}\cdot k^{O(1)}\cdot n$ time algorithm for \textsc{Feedback Vertex Set}, while in this paper, we design a  $54^{k}\cdot 2^{2(\omega+1)\cdot k}\cdot n^{4}$ time algorithm for this latter problem.
Indeed, our approach is based on a more general parameter and is optimized for neither tree-width nor clique-width.
Even though the use of partitions can be considered as a formalism in \cite{BergougnouxK19} and \cite{BodlaenderCKN15}, our framework avoids naturally the use
of partitions. We prove that for connected and/or acyclic constraints, classifying the partial solutions with respect to their $d$-neighborhoods, for some
$d\geq 1$ depending on the considered problem, is enough to show that, if $M$ is a binary matrix whose rows are the partial solutions computed so far, and its
columns are the possible extensions of the partial solutions into complete solutions, any minimum (or maximum) weighted basis of $M$ is a representative set. We
moreover prove that one can indeed compute a minimum (or maximum) weighted basis of size bounded by a polynomial on the number of $1$-neighborhoods for
connectivity constraints (see Theorem \ref{thm:reduce1}). For acyclicity constraints, we prove a bound polynomial on the number of $2$-neighborhoods and linear
in a parameter $\mathcal{N}_{\sf f}$ that depends on the considered graph parameter ${\sf f}$ (see Theorem \ref{thm:reduceacy}). Consequently, the definitions
of the dynamic programming tables and the computational steps of our algorithms are simpler than those in \cite{BergougnouxK19,BodlaenderCKN15} and are almost
the same as the ones given in \cite{Bui-XuanTV13}, except for the fact we call the {\sf reduce} function in addition at each step.  This is particularly true for
\textsc{Feedback Vertex Set} where the use of weighted partitions to encode the partial solutions in \cite{BergougnouxK19} takes care of many technical details
concerning the acyclicity.

The results we obtain simplify the $2^{O(k^2)}\cdot n^{O(1)}$ time algorithm parameterized by rank-width for \textsc{Feedback Vertex Set} from \cite{GanianH10} and the
$n^{O(k)}$ time algorithms parameterized by mim-width for \textsc{Feedback Vertex Set} and \textsc{Longest Induced Path} from \cite{JaffkeKT20a,JaffkeKT20b}. We also notice
that contrary to the algorithm for \textsc{Max Cut} (and its variants) given in \cite{DuarteLPSS19,EtoHKK19,FominGLS14}, there is no need to assume that the graph is given with a clique-width expression as our algorithm can be parameterized by $\bQ$-rank-width, which is always smaller than clique-width and for which also a fast fixed-parameter tractable (FPT) $(3k+1)$-approximation algorithm exists \cite{OumSV13,OumS06}.

Concerning mim-width, we provide unified polynomial time algorithms for the considered problems for all well-known graph classes having bounded mim-width and for which a layout of bounded mim-width can be computed in polynomial time \cite{BelmonteV13} (\eg, interval graphs, \text{circular arc} graphs, \text{permutation} graphs, Dilworth-$k$ graphs, and $k$-polygon graphs for all fixed $k$).
Notice that we also generalize one of the results from \cite{MontealegreT16} proving that the \textsc{Connected Vertex Cover} problem is solvable in polynomial time for circular arc graphs.

It is worth noticing that the approach used in \cite{CyganNPPRW11} called Cut \& Count can also be generalized to the $d$-neighbor-width for any
\textsc{Connected $(\sigma,\rho)$-dominating set} problem with more or less the same arguments used in this paper (see the Ph.D. thesis \cite[Theorem 4.66]{Bergougnoux19}).
However, it is not clear how to generalize the Cut \& Count approach to solve the acyclic variants of the \textsc{Connected $(\sigma,\rho)$-dominating set} problems with
the width measures considered in this paper.

\subsection*{Organization of this paper} We give some general definitions and notations in section \ref{sec:prelim} and present the \emph{$d$-neighbor equivalence} relation in section
\ref{sec:dneighbor}. The framework based on the $1$-neighbor equivalence relation is given in section \ref{sec:represents}. The applications to connectivity and
acyclicity constraints on $(\sigma,\rho)$-\textsc{Dominating Set} problems are given in sections \ref{sec:dom} and \ref{sec:maxinducedtree}. The algorithms concerning \textsc{Max
	Cut} and variants are given in section \ref{sec:maxcut}. In section \ref{sec:problem}, we explain how to use our framework for solving locally checkable partitioning problems with multiple global constraints. Finally, we conclude with some open questions and by giving some examples of
problems which might be interesting to tackle with the help of the $d$-neighbor equivalence relation in section \ref{sec:conclusion}.

\section{Preliminaries}\label{sec:prelim}

We denote by $\bN$ the set of nonnegative integers and by $\bN^+$ the set $\bN\setminus \{0\}$.
Let $V$ be a finite set. The size of a set $V$ is denoted by $|V|$ and its power set is denoted by $2^V$.
A set function $\f: 2^V \to \bN$ is \emph{symmetric} if, for all $S\subseteq V$, we have $\f(S)=\f(V\setminus S)$.
We write $A\setminus B$ for the set difference of $A$ from $B$.
We let $\min (\emptyset):= +\infty$ and $\max(\emptyset):=-\infty$.

\subsection*{Graphs}
Our graph terminology is standard and we refer to \cite{Diestel12}.  The vertex set of a graph $G$ is denoted by $V(G)$ and its edge set by $E(G)$.  All the
graphs considered in this paper are node-weighted, and for all of them, we denote by $\w : V(G)\rightarrow \bQ$ the weight function.
For every $X\subseteq V(G)$, we denote by $\w(X)$ the sum $\sum_{x\in X} \w(x)$ and we consider that $\w(\emptyset)=0$.

For every vertex set $X\subseteq V(G)$, when the underlying graph is clear from context, we denote by $\comp{X}$ the set $V(G)\setminus X$.
An edge between two vertices $x$ and $y$ is denoted by $xy$ or $yx$.
The set of vertices that are adjacent to $x$ is denoted by $N_G(x)$. For a set $U\subseteq V(G)$, we define  $N_G(U):=(\cup_{x\in U}N_G(x))\setminus U$. If the underlying graph is clear, then we may remove $G$ from the subscript.

The subgraph of $G$ induced by a subset $X$ of its vertex set is denoted by $G[X]$.
For two disjoint subsets $X$ and $Y$ of $V(G)$, we denote by $G[X,Y]$ the bipartite graph with vertex set $X\cup Y$ and edge set $\{xy \in E(G)\mid x\in X \text{ and } \ y\in Y \}$.
Moreover, we write $E(X,Y)$ to denote the set $E(G[X,Y])$ and we denote by $M_{X,Y}$ the adjacency matrix between $X$ and $Y$, \ien, the $(X,Y)$-matrix such that $M_{X,Y}[x,y]=1$ if $y\in N(x)$ and 0 otherwise.
When $X$ or $Y$ is empty, $M_{X,Y}$ is the empty matrix and $E(X,Y)=\emptyset$.

A \emph{matching} is a set of edges having no common endpoint
and an \emph{induced matching} is a matching $M$ where every pair of edges of $M$ do not have a common adjacent edge in $G$. 
The \emph{size of an induced matching} $M$ refers to the number of edges in $M$.

For two subsets $\cA$ and $\cB$ of $2^{V(G)}$, we define the \emph{merging} of $\cA$ and $\cB$, denoted by $\cA\otimes \cB$, as
$\cA \otimes \cB :=\{ X\cup Y \mid X\in \cA \,\text{ and }\, Y\in \cB\}$.
Observe that $\cA\otimes \cB=\emptyset$ whenever $\cA=\emptyset$ or $\cB=\emptyset$.

For a graph $G$, we denote by $\cc(G)$ the partition $\{V(C)\mid C$ is a connected component of $G\}$.
Let $X\subseteq V(G)$.
A \emph{consistent cut} of $X$ is an ordered bipartition $(X_1,X_2)$ of $X$ such that $N(X_1)\cap X_2 =\emptyset$.
We denote by $\cuts(X)$ the set of all consistent cuts of $X$.
In our proofs, we use the following facts.
\begin{fact}\label{fact:cutsCC}
	Let $X\subseteq V(G)$. For every $C\in \cc(G[X])$ and every $(X_1,X_2)\in \cuts(X)$, we have either $C\subseteq X_1$ or $C\subseteq X_2$.
\end{fact}

We deduce from the above fact that $|\cuts(X)|=2^{|\cc(G[X])|}$.
\begin{fact}\label{fact:cutsXcupY}
	Let $X$ and $Y$ be two disjoint subsets of $V(G)$. We have $(W_1,W_2)\in \cuts(X\cup Y)$ if and only if the following conditions are satisfied:
	\begin{enumerate}
		\item  $(W_1\cap X, W_2\cap X)\in \cuts(X)$,
		\item  $(W_1\cap Y, W_2\cap Y)\in \cuts(Y)$, and
		\item  $N(W_1\cap X)\cap (W_2 \cap Y)=\emptyset$ and $N(W_2\cap X)\cap (W_1 \cap Y)=\emptyset$.
	\end{enumerate}
\end{fact}

\subsection*{Rooted layout}
A \emph{rooted binary tree} is a binary tree with a distinguished vertex called the \emph{root}.
Since we manipulate at the same time graphs and trees representing them, the vertices of trees will be called \emph{nodes}.

A \emph{rooted layout} of $G$ is a pair $\cL=(T,\delta)$ of a rooted binary tree $T$ and a bijective function $\delta$ between $V(G)$ and the leaves\footnote{When $|V(G)|=1$ the only rooted layout of $G$ is an isolated vertex that is considered as a leaf.} of $T$.

For each node $x$ of $T$, let $L_x$ be the set of all the leaves $l$ of $T$ such that the path from the root of $T$ to $l$ contains $x$.
We denote by $V_x^\cL$ the set of vertices that are in bijection with $L_x$, \ien, $V_x^\cL:=\{v \in V(G)\mid \delta(v)\in L_x\}$.
When $\cL$ is clear from the context, we may remove $\cL$ from the superscript.
Observe that every rooted layout of an $n$-vertex graph has $2n-1$ nodes.

All the width measures dealt with in this paper are special cases of the following one, the difference being in each case the used set function.
Given a set function $\f: 2^{V(G)} \to\bN$ and a rooted layout $\cL=(T,\delta)$, the $\f$-width of a node $x$ of $T$ is $\f(V_x^\cL)$ and the $\f$-width of $(T,\delta)$, denoted by $\f(T,\delta)$ (or $\f(\cL)$), is $\max\{\f(V_x^\cL ) \mid x \in V(T)\}$.
Finally, the $\f$-width of $G$ is the minimum $\f$-width over all rooted layouts of $G$.

\subsection*{Clique-width/Module-width}
We will not define \emph{clique-width} but define its equivalent measure \emph{module-width} \cite{Rao06}.
The module-width of a graph $G$ is the $\mw$-width where $\mw(A)$ is the cardinality of $\{ N(v)\cap \comp{A} \mid  v\in A \}$ for all $A\subseteq V(G)$.
One also observes that $\mw(A)$ is the number of different rows in $M_{A,\comp{A}}$.
The following theorem shows the link between module-width and clique-width.

\begin{theorem}[{see \cite[Theorem 6.6]{Rao06}}]\label{thm:rao}
	For every $n$-vertex graph $G$, $\mw(G)\leq \cw(G) \leq 2\mw(G)$, where $\cw(G)$ denotes the clique-width of $G$. One can moreover translate, in time at most $O(n^2)$, a given decomposition into the other one with width at most the given bounds.
\end{theorem}

\subsection*{(\boldmath$\bQ$-)rank-width}
The rank-width and $\bQ$-rank-width are, respectively, the $\rw$-width and $\Qrw$-width where $\rw(A)$ (resp., $\Qrw(A)$) is the rank over $GF(2)$ (resp., $\bQ$) of the matrix $M_{A,\comp{A}}$ for all $A\subseteq V(G)$.

\subsection*{mim-width}
The mim-width of a graph $G$ is the $\mim$-width of $G$ where $\mim(A)$ is the size of a maximum induced matching of the graph $G[A,\comp{A}]$ for all $A\subseteq V(G)$.

It is worth noticing that module-width is the only parameter associated with a set function that is not symmetric.

The following lemma provides some upper bounds between mim-width and the other parameters that we use in section \ref{sec:maxinducedtree}. All of these upper bounds are proved in \cite{Vatshelle12}.

\begin{lemma}[see \cite{Vatshelle12}]\label{lem:comparemim}
	Let $G$ be a graph. For every $A\subseteq V(G)$,  $\mim(A)$ is upper bounded by $\mw(A)$, $\rw(A)$, and $\Qrw(A)$.
\end{lemma}
\begin{proof}
	Let $A\subseteq V(G)$.
	It is clear that $\mim(A)$ is upper bounded by the number of different rows in $M_{A,\comp{A}}$, so $\mim(A)\leq \mw(A)$.
	Let $S$ be the vertex set of a maximum induced matching of the graph $G[A,\comp{A}]$.
	Observe that the restriction of the matrix $M_{A,\comp{A}}$ to rows and columns in $S$ is the identity matrix.
	Hence, $\mim(A)$ is upper bounded by both $\rw(A)$ and $\Qrw(A)$.
\end{proof}

\section{The \boldmath$d$-neighbor equivalence}
\label{sec:dneighbor}

Let $G$ be a graph. The following definition is from \cite{Bui-XuanTV13}.
Let $A\subseteq V(G)$ and $d\in \bN^+$. Two subsets $X$ and $Y$ of $A$ are \emph{$d$-neighbor equivalent over $A$}, denoted by $X\equi{A}{d} Y$, if $\min(d,|X\cap N(u)|) = \min(d,|Y\cap N(u)|)$ for all $u\in \comp{A}$.
It is not hard to check that $\equi{A}{d}$ is an equivalence relation.
See Figure \ref{fig:nec} for an example of $2$-neighbor equivalent sets.

\begin{figure}[h!t]
	\centering
	\includegraphics[scale=0.95]{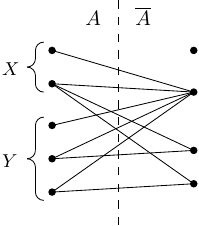}
	\caption{We have $X\equi{A}{2} Y$, but it is not the case that $X\equi{A}{3} Y$.}
	\label{fig:nec}
\end{figure}

For all $d\in \bN^+$, we let $\nec_d : 2^{V(G)}\to \bN$ where, for all $A\subseteq V(G)$, $\nec_d(A)$ is the number of equivalence classes of $\equi{A}{d}$.
Notice that while $\nec_1$ is a symmetric function \cite[Theorem 1.2.3]{Kim82}, $\nec_d$ is not necessarily symmetric for $d\geq 2$.
For example, if a vertex $x$ of $G$ has $c$ neighbors, then,    for every $d\in\bN^+$, we have $\nec_d(\{x\})= 2$ and $\nec_d(\comp{\{x\}})=1+\min(d,c)$.
It is worth noticing that, for all $d\in \bN^+$ and $A\subseteq V(G)$, $\nec_d(A)$ and $\nec_d(\comp{A})$ are at most $\nec_1(A)^{d \log_2(\nec_1(A))}$ \cite{Bui-XuanTV13}.

The following fact follows directly from the definition of the $d$-neighbor equivalence relation. We use it several times in our proofs.
\begin{fact}\label{fact:equivbiggerset}
	Let $A,B\subseteq V(G)$ such that $A\subseteq B$ and let $d\in\bN^+$.
	For all $X,Y\subseteq A$, if $X\equi{A}{d}  Y$, then $X\equi{B}{d} Y$.
\end{fact}

In order to manipulate the equivalence classes of $\equi{A}{d}$, one needs to compute a representative  for each equivalence class in polynomial time. This is achieved with the following notion of a representative.
Let $G$ be a graph with an arbitrary ordering of $V(G)$ and let $A\subseteq V(G)$.
For each $X\subseteq A$, let us denote by $\rep{A}{d}(X)$ the lexicographically smallest set $R\subseteq A$ among all $R\equi{A}{d} X$ of minimum size.
Moreover, we denote by $\Rep{A}{d}$ the set $\{\rep{A}{d}(X)\mid X\subseteq A\}$.
It is worth noticing that the empty set always belongs to $\Rep{A}{d}$ for all $A\subseteq V(G)$ and $d\in\bN^+$.
Moreover, we have $\Rep{V(G)}{d}=\Rep{\emptyset}{d}=\{\emptyset\}$ for all $d\in \bN^+$.
In order to compute $\Rep{A}{d}$, we use the following lemma.\vspace*{-.2pc}

\begin{lemma}[see \cite{Bui-XuanTV13}]\label{lem:computenecd}
	Let $G$ be an $n$-vertex graph. For every $A\subseteq V(G)$ and $d\in \bN^+$, one can compute in time $O(\nec_d(A) \cdot n^2 \cdot \log(\nec_d(A)))$ the sets $\Rep{A}{d}$ and a data structure that, given
	a set $X\subseteq A$, computes $\rep{A}{d}(X)$ in time $O(|A|\cdot n\cdot \log(\nec_d(A)))$.
\end{lemma}

\subsection*{\boldmath$d$-neighbor-width}
For every graph $G$ and $d\in\bN^+$, the \emph{$d$-neighbor-width} is the parameter obtained through the symmetric function $\snec_d: 2^{V(G)} \to \bN$ such that \[ \snec_d(A)=\max(\nec_d(A),\nec_d(\comp{A})). \]

The following lemma shows how the $d$-neighbor-width is upper bounded by the other parameters; most of the upper bounds were already proved in \cite{BelmonteV13,OumSV13}.\vspace*{-.2pc}

\begin{lemma}[see \cite{BelmonteV13,OumSV13,Vatshelle12}]\label{lem:compare}
	Let $G$ be a graph. For every $A\subseteq V(G)$ and $d\in \bN^+$, we have the following upper bounds on $\nec_d(A)$ and $\nec_d(\comp{A})$:
	\begin{multicols}{2}
		\begin{enumerate}[\rm(a)]
			\item $(d+1)^{\mw(A)}$,
			\item $2^{d\cdot \rw(A)^2}$,
			\par\begingroup
			\item $(d\cdot \Qrw(A) + 1 )^{\Qrw(A)}$,
			\item $n^{d\cdot \mim(A)}$.
			\par\endgroup
		\end{enumerate}
	\end{multicols}
\end{lemma}
\begin{proof}
	The first upper bound was proved in \cite[Lemma 5.2.2]{Vatshelle12}.
	The second upper bound was implicitly proved in \cite{Vatshelle12} and is due to the fact that $ \nec_d(A)\leq \mw(A)^{d\cdot \mim(A)}$ \cite[Lemma 5.2.3]{Vatshelle12}.
	Since $\mim(A)\leq \rw(A)$ by Lemma \ref{lem:comparemim} and $\mw(A)\leq 2^{\rw(A)}$ \cite{OumS06}, we deduce that $\nec_d(A)\leq 2^{d\cdot\rw(A)^2}$.
	The third upper bound was proved in \cite[Theorem 4.2]{OumSV13}.
	The fourth was proved in \cite[Lemma 2]{BelmonteV13}.
\end{proof}\vspace*{-.2pc}

Lemma \ref{lem:compare} implies the following.\vspace*{-.2pc}

\begin{corollary}\label{cor:nec_d}
	If there exists an algorithm that, given an $n$-vertex graph $G$ and a rooted layout $\cL$ of $G$, solves a problem $\Pi$ on $G$ in time  $\snec_c(\cL)^{O(1)}\cdot n^{O(1)}$ for some constant $c$, then $\Pi$ is decidable on $G$ within the following running times:
	\begin{multicols}{2}
		\begin{itemize}
			\item $2^{O(\mw(\cL))}\cdot n^{O(1)}$,
			\item $\Qrw(G)^{O(\Qrw(G))}\cdot n^{O(1)}$,
			\par\begingroup
			\item $2^{O(\rw(G)^2)}\cdot n^{O(1)}$,
			\item $n^{O(\mim(\cL))}$.
			\par\endgroup
		\end{itemize}
	\end{multicols}
\end{corollary}\vspace*{-.5pc}

Observe that the running times given in Corollary \ref{cor:nec_d} for rank-width and $\bQ$-rank-width use the width of the input graph and not the width of the rooted layout.
This follows from the fact that, given a graph $G$, we can compute a rooted layout of rank-width $3\rw(G)+1$ (resp., $\bQ$-rank width $3\Qrw(G)+1$) in time $2^{O(\rw(G))}\cdot n^3$ (resp., $2^{\Qrw(G)}\cdot n^{O(1)}$) \cite{OumSV13,OumS06}.

To deal with the \textsc{Max Cut} problem and the acyclic variants of \textsc{$(\sigma,\rho)$-Dominating Set} problems, we use the $n$-neighbor equivalence with $n$ being the number of vertices of the input graph.
For both problems, we use the following property of the $n$-neighbor equivalence.

\begin{lemma}\label{lem:property_nec_n}
	Let $G$ be a graph and $A\subseteq V(G)$.
	For all $X,W\subseteq A$ such that $X\equi{A}{n} W$ and for every $Y\subseteq \comp{A}$, we have $|E(X,Y) | = |E(W,Y) |$.
\end{lemma}
\begin{proof}
	Let $X,W\subseteq A$ such that $X\equi{A}{n} W$.
	Observe that, for every $v\in\comp{A}$, we have $\min(n,|N(v)\cap X|)=|N(v)\cap X|$.
	We deduce that, for every $v\in \comp{A}$, we have $|N(v)\cap X|=|N(v)\cap W|$.
	Thus, for every $Y\subseteq \comp{A}$, we have \[ |E(X,Y)|=\sum_{v\in Y}|N(v)\cap X| = \sum_{v\in Y}|N(v)\cap W| = |E(W,Y)|. \]
\end{proof}

For the acyclic variants of \textsc{$(\sigma,\rho)$-Dominating Set} problems, we use the $n$-neighbor equivalence over vertex sets of small size.
Given an integer $t\in\bN$, a graph $G$, and $A\subseteq V(G)$, we denote by $\nec_n^{\leq t}(A)$ the number of equivalence classes generated by the $n$-neighbor equivalence over the set $\{X\subseteq A \mid |X|\leq t\}$.
To deduce the algorithmic consequences on rank-width and $\bQ$-rank-width of our algorithm, we need the following upper bounds on $\nec_n^{\leq t}(A)$.

\begin{lemma}\label{lem:comparen_sizet}
	Let $G$ be a graph. For every $A\subseteq V(G)$, $\nec_n^{\leq t}(A)$ is upper bounded by $(t+1)^{\Qrw(A) }$, $(t+1)^{\mw(A)}$, and $2^{t\cdot\rw(A)}$.
\end{lemma}
\begin{proof}
	We start by proving that  $|\nec^{\leq t}_n(A)|\leq (t+1)^{\Qrw(A)}$.
	Observe that this inequality was implicitly proved in \cite[Theorem 4.2]{OumSV13} and that we use the same arguments here.
	For $X\subseteq A$, let $\sigma(X)$ be the vector corresponding to the sum over $\bQ$ of the row vectors of $M_{A,\comp{A}}$ corresponding to $X$.
	Observe that, for every $X,W\subseteq A$, we have $X\equi{A}{n} W$ if and only if $\sigma(X)=\sigma(W)$.
	Hence, we have $\nec^{\leq t}_n(A)=|\{\sigma(X)\mid X\subseteq A \wedge |X|\leq t\}|$.
	
	Let $C$ be a set of $\Qrw(A)$ linearly independent columns of $M_{A,\comp{A}}$.
	Since the rank over $\bQ$ of $M_{A,\comp{A}}$ is $\Qrw(A)$, every linear combination of row vectors of $M_{A,\comp{A}}$ is completely determined by its entries in $C$.
	For every $X\subseteq A$, the values in $\sigma(X)$ are between 0 and $t$.
	Hence,  we conclude that $\nec^{\leq t}_n(A)=|\{\sigma(X)\mid X\subseteq A\ \wedge\ |X|\leq t\}| \leq (t+1)^{\Qrw(A)}$.
	Since $\mw(A)$ corresponds to the number of different rows in $M_{A,\comp{A}}$, we have $\Qrw(A)\leq \mw(A)$.
	Thus, we have $\nec_n^{\leq t}(A)\leq (t+1)^{\mw(A)}$.
	
	Concerning rank-width, since there are at most $2^{\rw(A)}$ different rows in the matrix $M_{A,\comp{A}}$, we conclude that  $\nec_n^{\leq t}(A)\leq 2^{t\cdot \rw(A)}$.
\end{proof}

For the algorithms in section \ref{sec:maxinducedtree}, we use the $n$-neighbor equivalence over sets of size at most $2\mim(A)$, so we do not need any upper bounds on $\nec_n^{\leq t}(A)$ in function of mim-width as it is trivially upper bounded by $n^{2\mim(A)}$.
Moreover, for module-width, we will use another fact to get the desired running time.

It is worth noticing that upper bounds (b) and (c) of Lemma \ref{lem:compare} are implied by Lemma \ref{lem:comparen_sizet} and the fact that, for every $d\in \bN$ and $A\subseteq V(G)$, a representative of minimum size of a $d$-neighbor equivalence class over $A$ has size at most $d\cdot\Qrw(A)$ and $d\cdot\rw(A)$ \cite{Bui-XuanTV13,OumSV13}.

For \textsc{Max Cut}, we use the $n$-neighbor equivalence over sets of arbitrary size.
In particular, we use the following property of this equivalence relation.

\begin{fact}\label{fact:necncomplement}
	Let $A\subseteq V(G)$.
	For every $X,W\subseteq A$ such that $X\equi{A}{n} W$, we have $A\setminus X \equi{A}{n} A\setminus W$.
\end{fact}
\begin{proof}
	Let $X,W\subseteq A$ such that $X\equi{A}{n} W$ and let $v$ be a vertex of $\comp{A}$.
	We have \[ |N(v)\cap A|=|N(v)\cap X| + |N(v)\cap (A\setminus X)|. \]
	As $X\equi{A}{n} W$, by Lemma \ref{lem:property_nec_n}, we have $|N(v)\cap X|=|N(v)\cap W|$.
	We deduce that \[ |N(v)\cap A|-|N(v)\cap W|=|N(v)\cap (A\setminus X)|. \]
	Since $|N(v)\cap A|-|N(v)\cap W|=|N(v)\cap (A\setminus W)|$, we conclude that
	$|N(v)\cap (A\setminus X)|=|N(v)\cap (A\setminus W)|$.
	As this equality holds for every $v\in\comp{A}$, we can conclude that $A\setminus X\equi{A}{n} A\setminus W$.
\end{proof}

The following lemma provides some upper bound on $\nec_n(A)$ and $\nec_n(\comp{A})$.

\begin{lemma}\label{lem:comparen}
	Let $G$ be an $n$-vertex graph. For every $A\subseteq V(G)$, we have the following upper bounds on $\nec_n(A)$ and $\nec_n(\comp{A})$:
	\begin{multicols}{3}
		\begin{enumerate}[\rm(a)]
			\item $n^{\mw(A)}$,
			\par\begingroup
			\item $n^{\Qrw(A) }$,
			\item $n^{2^{\rw(A)}}$.
			\par\endgroup
		\end{enumerate}
	\end{multicols}
\end{lemma}
\begin{proof}
	Let $A\subseteq V(G)$.
	If $A=V(G)$, then obviously we have $\nec_n(A)=1$.
	Otherwise, we have $\nec_n(A)= \nec_n^{\leq n-1}(A)$ and by Lemma \ref{lem:comparen_sizet}, we conclude that $\nec_n(A)\leq n^{\Qrw(A)}$.
	
	Since $\Qrw(A)\leq \mw(A)$ (see the previous proof), we deduce that $\nec_n(A)\leq n^{\mw(A)}$.
	Moreover, as any binary matrix $M$ of rank $k$ over $GF(2)$ has at most $2^{k}$ different rows, we have  $\mw(A)\leq 2^{\rw(A)}$.
	Thus, we conclude that $\nec_n(A)\leq n^{2^{\rw(A)}}$.
	These upper bounds on $\nec_n(A)$ hold also on $\nec_n(\comp{A})$ because $\Qrw$ and $\rw$ are symmetric functions.
\end{proof}

Observe that the upper bounds of Lemma \ref{lem:comparen} are almost tight, that is, for every $k\in\bN$, we prove in the appendix that there exists a graph $G$ and $A\subseteq V(G)$ such that $\rw(A)=k+1$, $\mw(A)=\Qrw(A)=2^k$, and $\nec_n(A)\in (n/2^k-1)^{2^{k}}$.

Lemma \ref{lem:comparen} has the following consequences.

\begin{corollary}\label{cor:necn}
	If there exists an algorithm that, given a graph $G$ and a rooted layout $\cL$ of $G$, solves a problem $\Pi$ on $G$ in time  $\snec_n(\cL)^{O(1)}\cdot n^{O(1)}$ for some constant $c$, then $\Pi$ is decidable on $G$ within the following running times:
	\begin{multicols}{3}
		\begin{itemize}
			\item $n^{O(\mw(G))}$,
			\par\begingroup
			\item $n^{O(\Qrw(G))}$,
			\item $n^{2^{O(\rw(G))}}$.
			\par\endgroup
		\end{itemize}
	\end{multicols}
\end{corollary}

It is worth noticing that the running time given by Corollary \ref{cor:necn} for module-width depends on the width of the graph.
This is due to the facts that, for every graph $G$, we have $\Qrw(G)\leq \mw(G)$ \cite[Theorem 3.6]{OumSV13} and we can compute a rooted layout of $G$ of $\bQ$-rank-width at most $3\Qrw(G)+1$ in time $2^{O(\Qrw(G))}\cdot n^{O(1)}$ \cite[Theorem 3.1]{OumSV13}.

\section{Representative sets}\label{sec:represents}
In the following, we fix $G$ an $n$-vertex graph, $(T,\delta)$ a rooted layout of $G$, and $\w : V(G) \to \bQ$ a weight function over the vertices of $G$.

In this section, we define a notion of representativity between sets of partial solutions for the connectivity.
Our notion of representativity is defined \wrtn some node $x$ of $T$ and the 1-neighbor equivalence class of some set $R'\subseteq \comp{V_x}$.
In our algorithms, $R'$ will always belong to $\Rep{\comp{V_x}}{d}$ for some $d\in \bN^+$.
Our algorithms compute a set of partial solutions for each $R'\in \Rep{\comp{V_x}}{d}$.
The partial solutions computed for $R'$ will be completed with sets $d$-neighbor equivalent to $R'$ over $\comp{V_x}$.
Intuitively, the $R'$'s represent some expectation about how we will complete our sets of partial solutions.
For the connectivity and the domination, $d=1$ is enough but if we need more information for some reason (for example, the $(\sigma,\rho)$-domination or the acyclicity), we may take $d>1$.
This is not a problem as the $d$-neighbor equivalence class of $R'$ is included in the $1$-neighbor equivalence class of $R'$.
Hence, in this section, we fix a node $x$ of $T$ and a set $R'\subseteq \comp{V_x}$ to avoid to overloading the statements by the sentence ``let $x$ be a node of $T$ and $R'\subseteq \comp{V_x}$.''
We let $\opt \in \{\min, \max \}$; if we want to solve a maximization (or minimization) problem, we use $\opt=\max$ (or $\opt=\min$). We use it also, as here, in the next sections.

We recall that two subsets $Y,W$ of $\comp{V_x}$ are 1-neighbor equivalent over $\comp{V_x}$ if they have the same neighborhood in $V_x$, \ien, $N(Y)\cap V_x = N(W)\cap V_x$.

\begin{definition}[$(x,R')$-representativity]\label{def:represent}
	Given a weight function $\w : V(G) \to \bQ$, for every $\cA \subseteq 2^{V(G)}$ and $Y\subseteq V(G)$, we define
	\[ \best(\cA,Y):=\opt\{\w(X) \mid X\in \cA \text{ and } G[X\cup Y] \text{ is connected }\}. \]
	Let $\cA,\cB\subseteq 2^{V_x}$. We say that $\cB$ $(x,R')$-represents $\cA$ if, for every $Y\subseteq \comp{V_x}$ such that $Y\equi{\comp{V_x}}{1} R'$, we have $\best(\cA,Y)=\best(\cB,Y)$.
\end{definition}

When $\cA=\emptyset$ or there is no $X\in \cA$ such that $G[X\cup Y]$ is connected, we have $\best(\cA,Y)=\opt(\emptyset)$ and this equals $-\infty$ if $\opt=\max$ or $+\infty$ when $\opt=\min$.

Notice that the $(x,R')$-representativity is an equivalence relation.
The set $\cA$ is meant to represent a set of partial solutions of $G[V_x]$ which have been computed.
We expect to complete these partial solutions with partial solutions of $G[\comp{V_x}]$ which are 1-neighbor equivalent to $R'$ over $\comp{V_x}$.
If $\cB$ $(x,R')$-represents $\cA$, then we can safely substitute $\cA$ by $\cB$ because the quality of the output of the dynamic programming algorithm will remain the same.
Indeed, for every subset $Y$ of $\comp{V_x}$ such that  $Y\equi{\comp{V_x}}{1} R'$, the optimum solutions obtained by the union of a partial solution in $\cA$ and $Y$ will have the same weight as the optimum solution obtained from the union of a set in $\cB$ and $Y$.

The next theorem presents the main tool of our framework: a function $\reduce$ that, given a set of partial solutions $\cA$,
outputs a subset of $\cA$ that $(x,R')$-represents $\cA$ and whose size is upper bounded by $\snec_1(\cL)^2$.
To design this function, we use ideas from the rank-based approach of \cite{BodlaenderCKN15}.
That is, we define a small matrix $\cC$ with $|\cA|$ rows and $\snec_1(V_x)^2$ columns. Then, we show that a basis of maximum weight of the row space of $\cC$ corresponds to an $(x,R')$-representative set of $\cA$.
Since $\cC$ has $\snec_1(\cL)^2$ columns, the size of a basis of $\cC$ is smaller than $\snec_1(\cL)^2$.
To compute this basis, we use the following lemma.
The constant $\omega$ denotes the matrix multiplication exponent, which is known to be strictly less than $2.3727$ due to \cite{Williams12}.

\begin{lemma}[see \cite{BodlaenderCKN15}]\label{lem:optG}
	Let $M$ be a binary $n\times m$-matrix with $m\leq n$ and let $\w:\{1,\ldots,n\}\to \bQ$ be a weight function on the rows of $M$. Then, one can find a basis of maximum (or minimum) weight of the row space of $M$ in time $O(nm^{\omega-1})$.
\end{lemma}

In order to compute a small $(x,R')$-representative set of a set $\cA\subseteq 2^{V_x}$, the following theorem requires that the sets in $\cA$ are pairwise 1-neighbor equivalent over $V_x$.
This is useful since in our algorithm we classify our sets of partial solutions with respect to this property.
We need this to guarantee that the partial solutions computed for $R'$ will be completed with sets $d$-neighbor equivalent to $R'$ over $\comp{V_x}$.
However, if one wants to compute a small $(x,R')$-representative set of a set $\cA$ that does not respect this property, then it is enough to compute an $(x,R')$-representative set for each $1$-neighbor equivalence class of $\cA$.
The union of these $(x,R')$-representative sets is an $(x,R')$-representative set of $\cA$.

\begin{theorem}\label{thm:reduce1}
	Let $R\in \Rep{V_x}{1}$ and $R'\subseteq \comp{V_x}$. Then, there exists an algorithm $\reduce$ that, given $\cA\subseteq 2^{V_x}$ such that $X\equi{V_x}{1} R$ for all $X\in \cA$, outputs in time $O(|\cA|\cdot \nec_1(V_x)^{2 (\omega -1)} \cdot n^2)$ a subset $\cB\subseteq \cA$ such that $\cB$ $(x,R')$-represents $\cA$ and  $|\cB| \leq \nec_1(V_x)^2$.
\end{theorem}
\begin{proof}
	We assume without loss of generality (w.l.o.g.) that $\opt=\max$; the proof is symmetric for $\opt=\min$.
	First, we suppose that $R'\equi{\comp{V_x}}{1}\emptyset$.
	Observe that, for every $Y\equi{\comp{V_x}}{1} \emptyset$, we have $N(Y)\cap V_x=N(\emptyset)\cap V_x=\emptyset$.
	It follows that, for every $Y\subseteq \comp{V_x}$ such that $Y\equi{\comp{V_x}}{1} \emptyset $ and $Y\neq \emptyset$, $G[X\cup Y]$ cannot be connected
	for any $\emptyset \ne X\in \cA$, and thus we have $\best(\cA,Y)= 0$ if $\emptyset\in \cA$ and $G[Y]$ is connected, or $-\infty$ otherwise.
	Moreover, we have $\best(\cA,\emptyset)=\max\{ \w(X) \mid X\in \cA \text{ and } G[X] \textrm{ is connected}\}$.
	Note that we consider the empty graph $G[\emptyset]$ to be connected and we have $\w(\emptyset)=0$.
	Hence, if $R'\equi{\comp{V_x}}{1}\emptyset$, then it is sufficient to return a set $\cB$ constructed as follows:
	if $\cA$ contains a set inducing a connected graph, we add to $\cB$ a set in $\cA$ of maximum weight that induces a connected graph and if $\emptyset\in \cA$, we add the empty set to $\cB$.

	Assume from now on that $R'$ is not 1-neighbor equivalent to $\emptyset$ over $\comp{V_x}$.  Let $X\in \cA$. If there exists $C\in \cc(G[X])$ such that $N(C)\cap R'=\emptyset$, then, for all $Y\equi{\comp{V_x}}{1}
	R'$, we have $N(C)\cap Y = \emptyset$.
	Moreover, as $R'$ is not 1-neighbor equivalent to $\emptyset$ over $\comp{V_x}$, we have $Y\neq \emptyset$.  Consequently, for every $Y\equi{\comp{V_x}}{1} R'$, the
	graph $G[X\cup Y]$ is not connected.  We can conclude that $\cA\setminus \{X\}$ $(x,R')$ represents $\cA$.  Thus, we can safely remove from $\cA$ all such sets and this can be done in time $|\cA|\cdot n^2$.
	From now on, we may assume that, for all $X\in \cA$ and for all $C\in \cc(G[X])$, we have $N(C)\cap R'\neq \emptyset$.
	It is worth noticing that if $R=\emptyset$ or more generally $N(R)\cap R'=\emptyset$, then by assumption, $\cA=\emptyset$.
	
	Indeed, if $N(R)\cap R'=\emptyset$, then, for every $X\in\cA$, we have $N(X)\cap R' = N(R)\cap R'=\emptyset$ and in particular, for every $C\in\cc(G[X])$, we have $N(C)\cap R'=\emptyset$ (and we have assumed that no such set exists in $\cA$).
	
	Symmetrically, if for some $Y\subseteq \comp{V_x}$ there exists $C\in \cc(G[Y])$ such that $N(C)\cap R=\emptyset$, then, for every $X\in \cA$, the graph $G[X\cup Y]$ is not connected.
	Let $\cD$ be the set of all subsets $Y$ of $\comp{V_x}$ such that $Y\equi{\comp{V_x}}{1} R'$ and, for all $C\in \cc(G[Y])$, we have $N(C)\cap R\neq \emptyset$.
	Notice that the sets in $2^{\comp{V_x}} \setminus \cD$ do not matter for the $(x,R')$-representativity.
	
	For every $Y\in \cD$, we let $v_Y$ be one fixed vertex of $Y$.
	In the following, we denote by $\mathfrak{R}$ the set $\{ (R_1',R_2') \in \Rep{\comp{V_x}}{1}\times \Rep{\comp{V_x}}{1}  \}$.
	Let $\cM$, $\cC$, and $\comp{\cC}$ be, respectively, an $(\cA, \cD)$-matrix, an $(\cA,\mathfrak{R} )$-matrix, and an $(\mathfrak{R}, \cD)$-matrix such that
	{\small \begin{align*}
			\cM[X,Y] & :=
			\begin{cases} 1 &\!\! \textrm{if $G[X\cup Y]$ is connected},\\
				0 &\!\! \textrm{otherwise}. \end{cases}\\
			\cC[X,(R_1',R_2')]&:=
			\begin{cases} 1 &\!\! \textrm{if $\exists (X_1,X_2)\,{\in}\, \cuts(X)$ such that $N(X_1)\,{\cap}\, R_2'\,{=}\,\emptyset$ and $N(X_2)\,{\cap}\, R_1' \,{=}\,\emptyset$,}\\
				0 &\!\! \textrm{otherwise}. \end{cases}\\
			\comp{\cC}[(R_1',R_2'),Y]&:=
			\begin{cases} 1 &\!\! \textrm{if $\exists (Y_1,Y_2)\in \cuts(Y)$ such that $v_Y \in Y_1$, $Y_1\equi{\comp{V_x}}{1}R_1'$ and $Y_2 \equi{\comp{V_x}}{1}R_2'$},\\
				0 &\!\! \textrm{otherwise}. \end{cases}
	\end{align*}}
	
	Intuitively, $\cM$ contains all the information we need.
	In fact, a basis of maximum weight of the row space of $\cM$ in $GF(2)$ is an $(x,R')$-representative set of $\cA$.
	But, $\cM$ is too big to be computable efficiently.
	Instead, we prove that a basis of maximum weight of the row space of $\cC$ is an $(x,R')$-representative set of $\cA$.
	This follows from the fact that $(\cC\cdot \comp{\cC})[X,Y]$ equals the number of consistent cuts $(W_1,W_2)$ in $\cuts(X\cup Y)$ such that $v_Y\in W_1$.
	That is, $(\cC\cdot \comp{\cC})[X,Y]=2^{|\cc(G[X\cup Y])| -1}$. Consequently, $\cM=_2 \cC\cdot \comp{\cC}$, where $=_2$ denotes the equality in $GF(2)$, \ien, $(\cC\cdot \comp{\cC})[X,Y]$ is odd if and only if $G[X\cup Y]$ is connected.
	We deduce the running time of $\reduce$ and the size of $\reduce(\cA)$ from the size of $\cC$ (\ien, $|\cA|\cdot \nec_1(V_x)^2$).
	
	We start by proving that $\cM=_2 \cC\cdot \comp{\cC}$.
	Let $X\in \cA$ and $Y\in \cD$.
	We want to prove the following equality:
	\[ (\cC\cdot \comp{\cC})[X,Y]= \sum_{(R_1',R_2')\in \mathfrak{R}}\cC[X,(R_1',R_2')]\cdot \comp{\cC}[(R_1',R_2'),Y] = 2^{|\cc(G[X\cup Y])| -1}. \]
	We prove this equality with the following two claims.
	\begin{claim}\label{claim:secondwayreduce}
		We have $\cC[X,(R_1',R_2')]\cdot \comp{\cC}[(R_1',R_2'),Y]=1$ if and only if there exists $(W_1,W_2)\in \cuts(X\cup Y)$ such that $v_Y\in W_1$, $W_1\cap Y \equi{\comp{V_x}}{1} R_1'$, and $W_2\cap Y\equi{\comp{V_x}}{1} R_2'$.
	\end{claim}
	\begin{proof}
		By definition, we have  $\cC[X,(R_1',R_2')]\cdot \comp{\cC}[(R_1',R_2'),Y]=1$ if and only if
		\begin{enumerate}[(a)]
			\item $\exists(Y_1,Y_2)\in \cuts(Y)$ such that $v_Y\in Y_1$, $Y_1 \equi{\comp{V_x}}{1} R_1'$, $Y_2\equi{\comp{V_x}}{1} R_2'$ and
			\item $\exists(X_1,X_2)\in \cuts(X)$ such that $N(X_1)\cap R_2'=\emptyset$ and $N(X_2)\cap R_1' =\emptyset$.
		\end{enumerate}
		
		Let $(Y_1,Y_2)\in \cuts(Y)$ and $(X_1,X_2)\in \cuts(X)$ that satisfy, respectively, properties (a) and (b).
		By definition of $\equi{\comp{V_x}}{1}$, we have $N(X_1)\cap Y_2=\emptyset$ because $N(X_1)\cap R_2'=\emptyset$ and  $Y_2\equi{\comp{V_x}}{1} R_2'$. Symmetrically, we have $N(X_2)\cap Y_1=\emptyset$.
		By Fact \ref{fact:cutsXcupY}, we deduce that $(X_1\cup Y_1,X_2\cup Y_2)\in \cuts(X\cup Y)$.
		This proves the claim.
	\end{proof}

	\begin{claim}\label{claim:onewayreduce}
		Let $(W_1, W_2)$ and $(W_1',W_2')\in \cuts(X\cup Y)$.  We have $W_1\cap Y\equi{\comp{V_x}}{1} W_1'\cap Y$ and $W_2\cap Y\equi{\comp{V_x}}{1} W_2'\cap Y$ if and only if $W_1=W_1'$ and $W_2=W_2'$.
	\end{claim}
	\begin{proof}
		We start by an observation about the connected components of $X\cup Y$.
		As $Y\in \cD$, for all $C\in\cc(G[Y])$, we have $N(C)\cap R \neq \emptyset$.
		Moreover, by assumption, for all $C\in \cc(G[X])$, we have $N(C)\cap R'\neq \emptyset$.
		Since $X\equi{V_x}{1} R$ and $Y \equi{\comp{V_x}}{1} R'$, every connected component of $G[X\cup Y]$ contains at least one vertex of $X$ and one vertex of $Y$.
		
		Suppose that $W_1\cap Y\equi{\comp{V_x}}{1} W_1'\cap Y$ and $W_2\cap Y\equi{\comp{V_x}}{1} W_2'\cap Y$.
		Assume toward a contradiction that $(W_1,W_2)\neq (W_1',W_2')$.
		As these cuts are a bipartition of $X\cup Y$, we deduce that $W_1\neq W_1'$ and $W_2\ne W_2'$.
		Since $W_1\neq W_1'$, by Fact \ref{fact:cutsCC}, we deduce that there exists $C\in \cc(G[X\cup Y])$ such that either (1) $C\subseteq W_1$ and $C\subseteq W_2'$ or (2) $C\subseteq W_1'$ and $C\subseteq W_2$.
		We can assume w.l.o.g. that there exists $C\in\cc(G[X\cup Y])$ such that $C\subseteq W_1$ and $C\subseteq W_2'$.
		From the above observation, $C$ contains at least one vertex of $X$ and one of $Y$ and we have $N(C\cap X)\cap (W_1\cap Y)\neq \emptyset$ and $N(C\cap X)\cap (W_2'\cap Y)\neq\emptyset$.
		But, since $W_2\cap Y\equi{\comp{V_x}}{1} W_2'\cap Y$, we have $N(C\cap X)\cap( W_2\cap Y)\neq \emptyset$.
		This implies in particular that $N(W_1)\cap W_2\neq \emptyset$. It is a contradiction with the fact that $(W_1,W_2)\in\cuts(X\cup Y)$.
		
		The other direction being trivial we can conclude the claim.
	\end{proof}
	
	Notice that Claim \ref{claim:onewayreduce} implies that, for every $(R_1',R_2')\in \mathfrak{R}$, there exists at most one consistent cut $(W_1,W_2)\in \cuts(X\cup Y)$ such that $v_Y\in W_1$, $W_1\cap Y \equi{\comp{V_x}}{1} R_1'$ and $W_2\cap Y \equi{\comp{V_x}}{1} R_2'$.
	We can thus conclude from these two claims that
	\[      (\cC\cdot \comp{\cC})[X,Y]=|\{ (W_1,W_2)\in\cuts(X\cup Y) \mid v_Y\in W_1 \}|.  \]
	By Fact \ref{fact:cutsCC}, we deduce that $(\cC\cdot \comp{\cC})[X,Y] = 2^{|\cc(G[X\cup Y])|-1}$ since every connected component of $G[X\cup Y]$ that does not contain $v_Y$ can be in both sides of a consistent cut.
	Hence, $(\cC\cdot \comp{\cC})[X,Y]$ is odd if and only if $|\cc(G[X\cup Y])|=1$. We conclude that $\cM=_2 \cC\cdot \comp{\cC}$.
	
	Let $\cB\subseteq \cA$ be a basis of maximum weight of the row space of $\cC$ over $GF(2)$.
	We claim that $\cB$ $(x,R')$-represents $\cA$.
	Let $Y\subseteq \comp{V_x}$ such that $Y\equi{\comp{V_x}}{1} R'$.
	Observe that, by definition of $\cD$, if $Y\notin \cD$, then $\best(\cA,Y)=\best(\cB, Y)=-\infty$.
	Thus, it is sufficient to prove that, for every $Y\in \cD$, we have $\best(\cA,Y)=\best(\cB, Y)$.
	
	Let $X\in \cA$ and $Y\in\cD$. Recall that we have proved that $M[X,Y]=_2 (\cC\cdot \comp{\cC})[X,Y]$.
	Since $\cB$ is a basis of $\cC$, there exists\footnote{Notice that $\cB'$ is unique because $\cB$ is a row basis of $\cC$.} $\cB'\subseteq \cB$ such that, for each $(R_1',R_2')\in \mathfrak{R}$, we have $\cC[X,(R_1',R_2')]=_2 \sum_{W\in \cB'} \cC[W,(R_1',R_2')]$. Thus, we have the following equality:
	\begin{align*}
		\cM[X,Y]&=_2\sum_{(R_1',R_2')\in \mathfrak{R}} \cC[X,(R_1',R_2')]\cdot \comp{\cC}[(R_1',R_2'),Y]\\
		&=_2 \sum_{(R_1',R_2')\in \mathfrak{R}} \left(  \sum_{W\in \cB'} \cC[W,(R_1',R_2')]\right)  \cdot \comp{\cC}[(R_1',R_2'),Y]\\
		&=_2 \sum_{W\in \cB'} \left(\sum_{(R_1',R_2')\in \mathfrak{R}} \cC[W,(R_1',R_2')] \cdot \comp{\cC}[(R_1',R_2'),Y]\right)\\
		&=_2 \sum_{W\in \cB'} (\cC\cdot \comp{\cC})[W,Y]=_2 \sum_{W\in \cB'} \cM[W,Y].
	\end{align*}
	
	If $\cM[X,Y]=1$ (\ien, $G[X\cup Y]$ is connected), then there is an odd number of sets $W$ in $\cB'$ such that $\cM[W,Y]=1$ (\ien, $G[W\cup Y]$ is connected). Hence, there exists at least one $W\in \cB'$ such that $G[W\cup Y]$ is connected.
	Let $W\in \cB'$ such that $\cM[W,Y]=1$ and $\w(W)$ is maximum. Assume toward a contradiction that $\w(W)<\w(X)$.
	Notice that $(\cB \setminus \{W\} )\cup \{X\}$ is also a basis of $\cC$.
	Indeed, by definition of $\cB'$, for every $(R_1',R_2')$, we have $\cC[W,(R_1',R_2')]=_2 \cC[X,(R_1',R_2')] + \sum_{Z\in \cB'\setminus\{W\}} \cC[Z,(R_1',R_2')]$. That is, we can generate the row of $W$ in $\cC$ with the rows of $(\cB \setminus \{W\} )\cup \{X\}$.
	As $\cB$ is a basis of $\cC$, we deduce that $(\cB \setminus \{W\} )\cup \{X\}$ is also a basis of $\cC$.
	Since $\w(W)<\w(X)$, the weight of the basis $(\cB \setminus \{W\} )\cup \{X\}$ is strictly greater than the weight of the basis $\cB$, yielding a contradiction. Thus, we have $\w(X)\leq \w(W)$.
	Hence, for all $Y\in \cD$ and all $X\in \cA$, if $G[X\cup Y]$ is connected, then there exists $W\in \cB$ such that $G[W\cup Y]$ is connected and $\w(X)\leq \w(W)$.
	This is sufficient to prove that $\cB$ $(x,R')$-represents $\cA$.
	Since $\cB$ is a basis, the size of $\cB$ is at most the number of columns of $\cC$, thus,  $|\cB|\leq \nec_1(V_x)^2$.
	
	It remains to prove the running time.
	We claim that $\cC$ is computable in time $O(|\cA|\cdot \nec_1 (V_x)^2 \cdot n^2)$
	By Fact \ref{fact:cutsCC}, $\cC[X,(R_1',R_2')]=1$ if and only if, for each $C\in\cc(G[X])$, we have either $N(C)\cap R_1'=\emptyset$ or $N(C)\cap R_2'=\emptyset$.
	Thus, each entry of $\cC$ is computable in time $O(n^2)$.
	Since $\cC$ has $|\cA|\cdot |\Rep{\comp{V_x}}{1}|^2=|\cA|\cdot \nec_1(V_x)^2$ entries, we can compute $\cC$ in time $O(|\cA|\cdot \nec_1 (V_x)^2 \cdot n^2)$.
	Furthermore, by Lemma \ref{lem:optG}, a basis of maximum weight of $\cC$ can be computed in time $O(|\cA|\cdot \nec_1(V_x)^{2 (\omega -1)})$.
	We conclude that $\cB$ can be computed in time $O(|\cA|\cdot \nec_1(V_x)^{2 (\omega -1)} \cdot n^2)$.
\end{proof}

Now to boost up a dynamic programming algorithm $P$ on some rooted layout $(T,\delta)$ of $G$, we can use the function $\reduce$ to keep the size of the sets of partial solutions bounded by $\snec_1(T,\delta)^2$.
We call $P'$ the algorithm obtained from $P$ by calling the function $\reduce$ at every step of computation.
We can assume that the set of partial solutions $\cA_r$ computed by $P$ and associated with the root $r$ of $(T,\delta)$ contains an optimal solution (this will be the cases in our algorithms).
To prove the correctness of $P'$, we need to prove that $\cA_r'$ $(r,\emptyset)$-represents $\cA_r$ where $\cA_r'$ is the set of partial solutions computed by $P'$ and associated with $r$.
For doing so, we need to prove that at each step of the algorithm the operations we use preserve the $(x,R')$-representativity.
The following fact states that we can use the union without restriction; it follows directly from Definition \ref{def:represent} of $(x,R')$-representativity.

\begin{fact}\label{fact:unionpreserve}
	If $\cB$ and $\cD$ $(x,R')$-represents,  respectively, $\cA$ and $\cC$, then $\cB\cup \cD$ $(x,R')$-represents $\cA\cup \cC$.
\end{fact}

The second operation we use in our dynamic programming algorithms is the merging operator $\otimes$.
In order to safely use it, we need the following notion of compatibility that just tells which partial solutions from $V_a$ and $V_b$ can be joined to possibly form a
partial solution in $V_x$. (It was already used in \cite{Bui-XuanTV13} without naming it.)
\begin{definition}[$d$-$(R,R')$-compatibility]\label{def:compatibility}
	Suppose that $x$ is an internal node of $T$ with $a$ and $b$ as children.
	Let $d\in \bN^+$ and $R\in \Rep{V_x}{d}$.
	We say that $(A,A')\in \Rep{V_a}{d}\times \Rep{\comp{V_a}}{d} $ and $(B,B')\in \Rep{V_b}{d} \times \Rep{\comp{V_b}}{d}$ are $d$-$(R,R')$-compatible if we have
	\begin{itemize}
		\item $A\cup B \equi{V_x}{d} R$,
		\item $A' \equi{\comp{V_a}}{d} B \cup R'$, and
		\item $B' \equi{\comp{V_b}}{d} A\cup R'$.
	\end{itemize}
\end{definition}

The $d$-$(R,R')$-compatibility just tells which partial solutions from $V_a$ and $V_b$ can be joined to possibly form a partial solution in $V_x$.

\begin{lemma}\label{lem:bigotimespreserve}
	Suppose that $x$ is an internal node of $T$ with $a$ and $b$ as children. Let $d\in \bN^+$ and $R\in \Rep{V_x}{d}$. Let $(A,A')\in \Rep{V_a}{d}\times \Rep{\comp{V_a}}{d} $ and $(B,B')\in \Rep{V_b}{d} \times \Rep{\comp{V_b}}{d}$ that are $d$-$(R,R')$-compatible.
	Let $ \cA\subseteq 2^{V_a}$ such that, for all $X\in \cA$, we have $X\equi{V_a}{d} A$, and let
	$\cB\subseteq 2^{V_b}$ such that, for all $W\in \cB$, we have $W\equi{V_b}{d} B$.
	If $\cA'\subseteq \cA$ $(a,A')$-represents $\cA$ and $\cB'\subseteq \cB$ $(b,B')$-represents $\cB$, then
	$ \cA' \otimes \cB' \ (x,R')\text{-represents } \cA \otimes \cB. $
\end{lemma}
\begin{proof}
	We assume w.l.o.g. that $\opt=\max$; the proof is symmetric for $\opt=\min$.
	Suppose that $\cA'\subseteq \cA$ $(a,A')$-represents $\cA$ and $\cB'\subseteq \cB$ $(b,B')$-represents $\cB$.
	To prove the lemma, it is sufficient to prove that $\best(\cA'\otimes \cB', Y)=\best(\cA\otimes \cB ,Y)$ for every $Y\equi{\comp{V_x}}{1} R'$.
	
	Let $Y\subseteq \comp{V_x}$ such that $Y\equi{\comp{V_x}}{1} R'$.
	We start by proving the following facts: (a)~for every $W\in\cB$, we have $W\cup Y \equi{\comp{V_a}}{1} A'$ and~(b)~for every $X\in \cA$, we have $X\cup Y\equi{\comp{V_b}}{1} B'$.
	
	Let $W\in \cB$. Owing to the $d$-$(R,R')$-compatibility, we have $B\cup R' \equi{\comp{V_a}}{d} A'$.  Since $W\equi{V_b}{d} B$ and $V_b \subseteq \comp{V_a}$, by Fact \ref{fact:equivbiggerset}, we deduce that $W\equi{\comp{V_a}}{d} B$ and thus $W\cup R' \equi{\comp{V_a}}{d} A'$.
	In particular, we have $W\cup R' \equi{\comp{V_a}}{1} A'$ because $d\geq 1$.
	Similarly, we have from Fact \ref{fact:equivbiggerset} that $W\cup Y \equi{\comp{V_a}}{1} A'$ because $Y\equi{\comp{V_x}}{1} R'$ and $\comp{V_x}\subseteq \comp{V_a}$.
	This proves fact (a).
	The proof for fact      (b) is symmetric.
	
	Now observe that, by the definitions of $\best$ and the merging operator $\otimes$, we have (even if $\cA=\emptyset$ or $\cB=\emptyset$)
	{\small
		\begin{align*}
			\best\big( \cA \otimes \cB, Y\big) &= \max\{ \w(X)+ \w(W) \mid X\in \cA \wedge W\in \cB \wedge G[X\cup W\cup Y] \text{ is connected} \}.
	\end{align*}}%
	Since $\best(\cA, W\cup Y)=\max\{ \w(X) \mid X\in \cA \wedge G[X\cup W\cup Y] \text{ is connected} \}$, we deduce that
	\begin{align*}
		\best\big( \cA \otimes \cB, Y\big) &= \max\{ \best(\cA, W\cup Y) + \w(W) \mid W\in \cB \}.
	\end{align*}
	\noindent Since $\cA'$ $(a,A')$-represents $\cA$, by fact (a), we have
	\begin{align*}
		\best\big( \cA \otimes \cB, Y\big) &= \max\{ \best(\cA', W\cup Y) + \w(W) \mid W\in \cB \}\\
		&=\best\big( \cA' \otimes \cB, Y\big).
	\end{align*}
	Symmetrically, we deduce from fact (b) that $\best\big( \cA' \otimes \cB , Y\big) = \best\big( \cA' \otimes \cB', Y\big)$.
	This stands for every $Y\subseteq \comp{V_x}$ such that $Y\equi{\comp{V_x}}{1}  R'$.
	Thus, we conclude that $\cA' \otimes \cB'$ $(x,R')$-represents $\cA \otimes \cB$.
\end{proof}

\section{Connected (Co-)\boldmath$(\sigma,\rho)$-dominating sets}\label{sec:dom}

Let $\sigma$ and $\rho$ be two (nonempty) finite or co-finite subsets of $\bN$.
We say that a subset $D$ of $V(G)$ \emph{$(\sigma, \rho)$-dominates} a subset $U\subseteq V(G)$ if
\begin{align*}
	|N(u)\cap D|\in & \begin{cases} \sigma & \textrm{if $u\in D$},\\
		\rho & \textrm{otherwise (if $u\in U\setminus D$).}
	\end{cases}
\end{align*}

A subset $D$ of $V(G)$ is a \emph{$(\sigma, \rho)$-dominating set} (resp., Co-\emph{$(\sigma, \rho)$-dominating set}) if $D$ (resp., $V(G)\setminus D$) \emph{$(\sigma, \rho)$-dominates} $V(G)$.

The \textsc{Connected $(\sigma,\rho)$-Dominating Set} problem asks, given a weighted graph $G$, a maximum or minimum $(\sigma, \rho)$-dominating set which induces a connected graph.
Similarly, one can define \textsc{Connected Co-$(\sigma,\rho)$-Dominating Set}.
Examples of some \textsc{Connected (Co-)$(\sigma,\rho)$-Dominating Set} problems are shown in Table \ref{tab:dom}.

Let $d(\bN):=0$ and for a finite or co-finite subset $\mu$ of $\bN$, let
\[ d(\mu):=1 + \min(\max(\mu),\max(\bN\setminus \mu)).  \]

The definition of $d$ is motivated by the following observation, which is due to the fact that, for all $\mu\subseteq \bN$, if $d(\mu)\in \mu$, then $\mu$ is co-finite and contains $\bN\setminus\{1,\dots,d(\mu)-1\}$.
\begin{fact}\label{fact:d}
	Let $A\subseteq V(G)$ and let $(\sigma,\rho)$ be a pair of finite or co-finite subsets of $\bN$.
	Let $d:=\max(1,d(\sigma),d(\rho))$.
	For all $X\subseteq A$ and $Y\subseteq \comp{A}$, $X\cup Y$ $(\sigma,\rho)$-dominates $A$ if and only if
	$\min(d,|N(v)\cap X| + |N(v)\cap Y|)$ belongs to $\sigma$ (resp., $\rho$) if $v\in X\cup Y$ (resp., $v\notin X\cup Y$).
\end{fact}

As in \cite{Bui-XuanTV13}, we use the \emph{$d$-neighbor equivalence} relation to characterize the $(\sigma,\rho)$-domination of the partial solutions.

\begin{table}[!pb]
	\footnotesize
	\centering
	\renewcommand\arraystretch{1.5}
	\caption{ Examples of \textsc{(Co-)$(\sigma,\rho)$-Dominating Set} problems.
		To solve these problems, we use the $d$-neighbor equivalence with $d=\max\{1,d(\sigma),d(\rho)\}$.
		Column $d$ shows the value of $d$ for each problem.}\label{tab:dom}
	{\small
		\begin{tabular}{c|c|c|c|c}
			\hline
			$\sigma$ & $\rho$ & $d$ & Version & Standard name  \\\hline
			$\bN$ & $\bN^+$ & 1 & \textsc{Connected} & \textsc{Connected Dominating Set}  \\\hline
			$\{q\}$ & $\bN$ & $q$+1  &\textsc{Connected}  & \textsc{Connected Induced $q$-Regular Subgraph}  \\ \hline
			$\bN$ & $\{1\}$ & 2 & \textsc{Connected}  & \textsc{Connected Perfect Dominating Set}  \\ \hline
			$\{0\}$ & $\bN$ & 1 & \textsc{Connected Co}  & \textsc{Connected Vertex Cover}\\ \hline
	\end{tabular} }
	
\end{table}

In the rest of this section, we fix  $\sigma,\rho$ two (nonempty) finite or co-finite subsets of $\bN$, the integer $d:=\max\{1,d(\sigma),d(\rho)\}$, a graph $G$, and a rooted layout $(T,\delta)$ of $G$.
We present an algorithm that \pagebreak computes a maximum (or minimum) connected $(\sigma,\rho)$-dominating set of $G$ by a bottom-up traversal of $(T,\delta)$.
Its running time is $O(\snec_d(T,\delta)^{O(1)}\cdot n^3)$.
The same algorithm, with some small modifications, will be able to find a minimum Steiner tree or a maximum (or minimum) connected Co-$(\sigma,\rho)$-dominating set as well.
We will need the following lemma in our proof.

\begin{lemma}[{see \cite[Lemma~2]{Bui-XuanTV13}}]\label{lem:equivR'}
	Let $A\subseteq V(G)$.  Let $X\subseteq A$ and $Y,Y'\subseteq \comp{A}$ such that $Y\equi{\comp{A}}{d} Y'$. Then $(X\cup Y)$ $(\sigma,\rho)$-dominates $A$ if
	and only if $(X\cup Y')$ $(\sigma,\rho)$-dominates $A$.
\end{lemma}

For each node $x$ of $T$ and for each pair $(R,R')\in \Rep{V_x}{d}\times \Rep{\comp{V_x}}{d}$, we will compute a set of partial solutions $\cD_x[R,R']$ \emph{coherent} with $(R,R')$ that $(x,R')$-represents the set of all partial solutions coherent with $(R,R')$.
We say that a set $X\subseteq V_x$ is coherent with $(R,R')$ if $X\equi{V_x}{d} R$ and $X\cup R'$ $(\sigma,\rho)$ dominates $V_x$.
Observe that by Lemma \ref{lem:equivR'}, we have that $X\cup Y$ $(\sigma,\rho)$-dominates $V_x$ for all $Y\equi{\comp{V_x}}{d} R'$ and for all $X\subseteq V_x$ coherent with $(R,R')$.
We compute these sets by a bottom-up dynamic programming algorithm, starting at the leaves of $T$.
The computational steps are trivial for the leaves. For the internal nodes of $T$, we simply use the notion of $d$-$(R,R')$-compatibility and the merging operator.

By calling the function $\reduce$ defined in section \ref{sec:represents}, each set $\cD_x[R,R']$ contains at most $\snec_1(T,\delta)^2$ partial solutions.
If we want to compute a maximum (resp., minimum) connected $(\sigma,\rho)$-dominating set, we use the framework of section \ref{sec:represents} with $\opt=\max$ (resp., $\opt=\min$).
If $G$ admits a connected $(\sigma,\rho)$-dominating set, then a maximum (or minimum) connected $(\sigma,\rho)$-dominating set can be found by looking at the entry $\cD_r[\emptyset,\emptyset]$ with $r$ the root of $T$.

We begin by defining the sets of partial solutions for which we will compute representative sets.

\begin{definition}\label{def:tabdom}
	Let $x\in V(T)$. For all pairs $(R,R')\in \Rep{V_x}{d}\times \Rep{\comp{V_x}}{d}$, we let $\cA_x[R,R']:=\{ X\subseteq V_x \mid X\equi{V_x}{d} R \textrm{ and } X\cup R'\  (\sigma,\rho)\textrm{-dominates } V_x\}.$
\end{definition}

For each node $x$ of $V(T)$, our algorithm will compute a table $\cD_x$ that satisfies the following invariant.

\subsection*{Invariant}
For every $(R, R')\in \Rep{V_x}{d}\times \Rep{\comp{V_x}}{d}$, the set $\cD_x[R,R']$ is a subset of $\cA_x[R,R']$ of size at most $\snec_1(T,\delta)^2$ that $(x,R')$-represents $\cA_x[R,R']$.

Notice that, by the definition of $\cA_r[\emptyset,\emptyset]$ ($r$ being the root of $T$) and the definition of $(x,R')$-representativity, if $G$ admits a connected $(\sigma,\rho)$-dominating set, then $\cD_r[\emptyset,\emptyset]$ must contain a maximum (or minimum) connected $(\sigma,\rho)$-dominating set.

The following lemma provides an equality between the entries of the table $\cA_x$ and the entries of the tables $\cA_a$ and $\cA_b$ for each internal node $x\in V(T)$ with $a$ and $b$ as children.
We use this lemma to prove, by induction, that the entry $\cD_x[R,R']$ $(x,R')$-represents $\cA_x[R,R']$ for every $(R,R')\in \Rep{V_x}{d}\times \Rep{\comp{V_x}}{d}$.
Note that this lemma can be deduced from \cite{Bui-XuanTV13}.

\begin{lemma}\label{lem:basicDymProg}
	For all $(R,R')\in \Rep{V_x}{d} \times \Rep{\comp{V_x}}{d}$, we have
	\[ \cA_x[R,R']= \bigcup_{(A,A') \text{ and } (B,B') \text{ are } d\textrm{-}(R,R')\textrm{-compatible}}  \cA_a[A,A'] \otimes \cA_b[B,B'], \]
	where the union is taken over all $(A,A')\in \Rep{V_a}{d}\times\Rep{\comp{V_a}}{d}$ and $(B,B')\in \Rep{V_b}{d}\times \Rep{\comp{V_b}}{d}$.
\end{lemma}
\begin{proof}
	The lemma is implied by the two following claims.
	\begin{claim}\label{claim:upequi}
		For all $X\in\cA_x[R,R']$, there exist $d$-$(R,R')$-compatible pairs $(A,A')\in \Rep{V_a}{d}\times \Rep{\comp{V_a}}{d}$ and $(B,B')\in \Rep{V_b}{d}\times \Rep{\comp{V_b}}{d}$ such that $X\cap V_a \in \cA_a[A,A']$ and $X\cap V_b \in \cA_b[B,B']$.
	\end{claim}
	\begin{proof}
		Let $X\in \cA_x[R,R']$, $X_a:=X\cap V_a$, and $X_b:=X\cap V_b$.  Let $A:=\rep{V_a}{d}(X_a)$ and $A':=\rep{\comp{V_a}}{d}(X_b \cup R')$. Symmetrically, we define
		$B:=\rep{V_b}{d}(X_b)$ and $B':=\rep{\comp{V_b}}{d}(X_a\cup R')$.
		
		We claim that $X_a\in \cA_a[A,A']$.
		As $X\in\cA_x[R,R']$, we know, by Definition \ref{def:tabdom}, that $X \cup R'=X_a\cup X_b\cup R'$ is a  $(\sigma,\rho)$-dominating set of $V_x$.
		In particular, $X_a\cup (X_b\cup R')$ $(\sigma,\rho)$-dominates $V_a$.
		Since $A'\equi{\comp{V_a}}{d} X_b \cup R'$, by Lemma \ref{lem:equivR'}, we conclude that $X_a\cup A'$ $(\sigma,\rho)$-dominates $V_a$.
		As $A\equi{V_a}{d} X_a$, we have $X_a\in\cA_a[A,A']$.
		By symmetry, we deduce $B\in \cA_b[B,B']$.
		
		It remains to prove that $(A,A')$ and $(B,B')$ are $d$-$(R,R')$-compatible.
		\begin{itemize}
			\item By construction, we have $X_a \cup X_b = X  \equi{V_x}{d} R$. As $A\equi{V_a}{d} X_a$ and from Fact \ref{fact:equivbiggerset}, we have $A \cup X_b \equi{V_x}{d} R$. Since $B\equi{V_b}{d} X_b$, we deduce that $A\cup B \equi{V_x}{d} R$.
			\item By definition, we have $A'\equi{\comp{V_a}}{d} X_b\cup R'$. As $B \equi{V_b}{d} X_b$ and by Fact \ref{fact:equivbiggerset}, we have $A'\equi{\comp{V_a}}{d} B \cup  R'$. Symmetrically, we deduce that $B'\equi{\comp{V_b}}{d} R' \cup A$.
		\end{itemize}
		Thus, $(A,A')$ and $(B,B')$ are $d$-$(R,R')$-compatible.
	\end{proof}
	\begin{claim}\label{claim:downequi}
		For every $X_a\in\cA_a[A,A']$ and $X_b\in\cA_b[B,B']$ such that $(A,A')$ and $(B,B')$ are $d$-$(R,R')$-compatible, we have $X_a\cup X_b \in \cA_x[R,R']$.
	\end{claim}
	\begin{proof}
		Since $X_a\equi{V_a}{d} A$ and $X_b \equi{V_b}{d} B$, by Fact \ref{fact:equivbiggerset}, we deduce that $X_a \cup X_b \equi{V_x}{d} A\cup B$.
		Thus, by the definition of $d$-$(R,R')$-compatibility, we have $X_a\cup X_b \equi{V_x}{d} R$.
		
		It remains to prove that $X_a\cup X_b \cup R'$ $(\sigma,\rho)$-dominates $V_x$.
		Since $X_b\equi{V_b}{d} B$ and $A'\equi{\comp{V_a}}{d} B\cup R'$, we deduce from Fact~\ref{fact:equivbiggerset} that $X_b\cup R' \equi{\comp{V_a}}{d} A'$.
		As  $X_a\in\cA_a[A,A']$, we know that $X_a\cup A'$ $(\sigma,\rho)$-dominates $V_a$.
		Thus, from Lemma~\ref{lem:equivR'}, we conclude that $X_a\cup X_b \cup R'$ $(\sigma,\rho)$-dominates $V_a$.
		Symmetrically, we prove that $X_a\cup X_b \cup R'$ $(\sigma,\rho)$-dominates $V_b$. As $V_x =V_a \cup V_b$, we deduce that $X_a\cup X_b \cup R'$ $(\sigma,\rho)$-dominates $V_x$.
		Hence, we have $X_a\cup X_b \in \cA_x[R,R']$.
	\end{proof}
\end{proof}

We are now ready to prove the main theorem of this section.

\begin{theorem}\label{thm:dom}
	There exists an algorithm that, given an $n$-vertex graph $G$  and a rooted layout $(T,\delta)$ of $G$, computes a maximum (or minimum)  connected $(\sigma,\rho)$-dominating set in time $O(\snec_d(T,\delta)^{3}\cdot  \snec_1(T,\delta)^{ 2(\omega + 1)} \cdot \log(\snec_d(T,\delta)) \cdot n^3)$ with $d:=\max\{1, d(\sigma),d(\rho)\}$.
\end{theorem}
\begin{proof}
	The algorithm is a usual bottom-up dynamic programming algorithm and computes for each node $x$ of $T$ the table $\cD_x$.
	
	The first step of our algorithm is to compute, for each $x\in V(T)$, the sets $\Rep{V_x}{d}$, $\Rep{\comp{V_x}}{d}$ and a data structure to compute $\rep{V_x}{d}(X)$ and $\rep{\comp{V_x}}{d}(Y)$, for any $X\subseteq V_x$ and any $Y\subseteq \comp{V_x}$, in time $O(\log(\snec_d(T,\delta)) \cdot n^2 )$.  As $T$ has $2n-1$ nodes, by Lemma \ref{lem:computenecd}, we can compute these sets and data structures in time $O(\snec_d(T,\delta)\cdot \log(\snec_d(T,\delta)) \cdot n^3)$.
	
	Let $x$ be a leaf of $T$ with $V_x=\{v\}$.
	Observe that, for all $(R,R')\in \Rep{d}{V_x}\times \Rep{d}{\comp{V_x}}$, we have $\cA_x[R,R'] \subseteq 2^{V_x} = \{\emptyset,  \{ v \}\}$.
	Thus, our algorithm can directly compute $\cA_x[R,R']$ and set $\cD_x[R,R']:=\cA_x[R,R']$.
	In this case, the invariant trivially holds.
	
	Now let $x$ be an internal node with $a$ and $b$ as children such that the invariant holds for $a$ and $b$.
	For each $(R,R')\in\Rep{V_x}{d}\times \Rep{\comp{V_x}}{d}$, the algorithm computes $\cD_x[R,R']:=\reduce(\cB_x[R,R'])$, where the set $\cB_x[R,R']$ is defined as follows:
	\[ \cB_x[R,R']:= \bigcup_{(A,A') \text{ and } (B,B') \text{ are } d\textrm{-}(R,R')\textrm{-compatible}}
	\cD_a[A,A'] \otimes \cD_b[B,B'],\]
	where the union is taken over all $(A,A')\in \Rep{V_a}{d}\times\Rep{\comp{V_a}}{d}$, and $(B,B')\in \Rep{V_b}{d}\times \Rep{\comp{V_b}}{d} $.
	
	We claim that the invariant holds for $x$. Let $(R,R')\in\Rep{V_x}{d}\times \Rep{\comp{V_x}}{d}$.
	
	We start by proving that the set $\cB_x[R,R']$ is an $(x,R')$-representative set of $\cA_x[R,R']$.
	By Lemma \ref{lem:bigotimespreserve}, for all $d$-$(R,R')$-compatible pairs $(A,A')$ and $(B,B')$, we have
	\[ \cD_a[A,A'] \otimes \cD_b[B,B']\ (x,R')\textrm{-represents } \cA_a[A,A'] \otimes \cA_b[B,B']. \]
	By Lemma \ref{lem:basicDymProg} and by construction of $\cD_x[R,R']$ and from Fact \ref{fact:unionpreserve}, we conclude that $\cB_x[R,R']$ $(x,R')$-represents $\cA_x[R,R']$.
	
	From the invariant, we have $\cD_a[A,A']\subseteq \cA_a[A,A']$ and $\cD_b[B,B']\subseteq \cA_b[B,B']$ for all $d$-$(R,R')$-compatible pairs $(A,A')$ and $(B,B')$.
	Thus, from Lemma \ref{lem:basicDymProg}, it is clear that by construction, we have $\cB_x[R,R']\subseteq \cA_x[R,R']$. Hence, $\cB_x[R,R']$ is a subset and an $(x,R')$-representative set of $\cA_x[R,R']$.
	
	Notice that, for each $X\in \cB_x[R,R']$, we have $X\equi{V_x}{d} R$. Thus, we can apply Theorem \ref{thm:reduce1} and the function $\reduce$ on $\cB_x[R,R']$.
	By  Theorem \ref{thm:reduce1}, $\cD_x[R,R']$ is a subset and an $(x,R')$-representative set of $\cB_x[R,R']$.
	Thus $\cD_x[R,R']$ is a subset of $\cA_x[R,R']$.
	Notice that the $(x,R')$-representativity is an equivalence relation and in particular it is transitive.
	Consequently, $\cD_x[R,R']$ $(x,R')$-represents $\cA_x[R,R']$.
	From Theorem \ref{thm:reduce1}, the size of $\cD_x[R,R']$ is at most $\nec_1(V_x)^2$ and we have that $\cD_x[R,R']\subseteq \cB_x[R,R']$.
	As $\nec_1(V_x)\leq \snec_1(T,\delta)$ and $\cB_x[R,R']\subseteq \cA_x[R,R']$, we conclude that the invariant holds for $x$.
	
	By induction, the invariant holds for all nodes of $T$.
	The correctness of the algorithm follows from the fact that $\cD_r[\emptyset,\emptyset]$ $(r,\emptyset)$-represents $\cA_r[\emptyset,\emptyset]$.
	
	\subsection*{Running time}
	Let $x$ be a node of $T$.
	Suppose first that $x$ is a leaf of $T$. Then $|\Rep{V_x}{d}|\leq 2$ and $|\Rep{\comp{V_x}}{d}|\leq d$. Thus, $\cD_x$ can be computed in time $O(d\cdot n)$.
	
	Assume now that $x$ is an internal node of $T$ with $a$ and $b$ as children.
	
	Notice that, by Definition \ref{def:compatibility}, for every $(A,B,R')\in\Rep{V_a}{d}\times\Rep{V_b}{d}\times\Rep{\comp{V_x}}{d}$, there exists only one tuple $(A',B',R)\in\Rep{\comp{V_a}}{d}\times\Rep{\comp{V_b}}{d}\times\Rep{V_x}{d}$ such that $(A,A')$ and $(B,B')$ are $d$-$(R,R')$-compatible. More precisely, you have to take $R=\rep{V_x}{d}(A\cup B)$, $A'=\rep{\comp{V_a}}{d}(R'\cup B)$, and $B'=\rep{\comp{V_b}}{d}(R'\cup A)$.
	Thus, there are at most $\snec_d(T,\delta)^{3}$ tuples $(A,A',B,B',R,R')$ such that $(A,A')$ and $(B,B')$ are $d$-$(R,R')$-compatible.
	It follows that we can compute the intermediary table $\cB_x$ by doing the following;
	\begin{itemize}
		\item Initialize each entry of $\cB_x$ to $\emptyset$.
		\item For each $(A,B,R')\in\Rep{V_a}{d}\times\Rep{V_b}{d}\times\Rep{\comp{V_x}}{d}$, compute $R':=\rep{V_x}{d}(A\cup B)$, $A'=\rep{\comp{V_a}}{d}(R'\cup B)$, and $B'=\rep{\comp{V_b}}{d}(R'\cup A)$.
		Then, update $\cB_x[R,R']:=\cB_x[R,R']\cup (\cD_a[A,A']\otimes \cD_b[B,B'])$.
	\end{itemize}
	Each call to the functions $\rep{V_x}{d}$, $\rep{\comp{V_a}}{d}$ and $\rep{\comp{V_b}}{d}$ takes $O(\log(\snec_d(T,\delta))\cdot n^2)$ time.
	We deduce that the running time to compute the entries of $\cB_x$ is
	\[              O\left(\snec_d(T,\delta)^3 \cdot \log(\snec_d(T,\delta))\cdot n^2 + \sum_{(R,R')\in\Rep{V_x}{d}\times \Rep{\comp{V_x}}{d} } |\cB_x[R,R']|\cdot n^2 \right). \]
	Observe that, for each $(R,R')\in \Rep{V_x}{d}\times \Rep{\comp{V_x}}{d}$, by Theorem \ref{thm:reduce1}, the running time to compute $\reduce(\cB_x[R,R'])$ from $\cB_x[R,R']$ is  $O(|\cB_x[R,R']|\cdot \snec_1(T,\delta)^{2(\omega -1)}\cdot n^2)$.
	Thus, the total running time to compute the table $\cD_x$ from the table $\cB_x$ is
	\begin{align}
		O\left(\sum_{(R,R')\in\Rep{V_x}{d}\times \Rep{\comp{V_x}}{d} }|\cB_x[R,R']| \cdot  \log(\snec_d(T,\delta)) \cdot \snec_1(T,\delta)^{2(\omega -1)}\cdot n^2\right).
	\end{align}
	For each $(A,A')$ and $(B,B')$, the size of $\cD_a[A,A'] \otimes \cD_b[B,B']$ is at most $|\cD_a[A,A']|\cdot |\cD_b[B,B']| \leq \snec_1(T,\delta)^4$.
	Since there are at most $\snec_d(T,\delta)^{3}$ pairs that are $d$-$(R,R')$-compatible, we can conclude that
	\[ \sum_{(R,R')\in\Rep{V_x}{d}\times \Rep{\comp{V_x}}{d} }|\cB_x[R,R']|\leq \snec_d(T,\delta)^3\cdot \snec_1(T,\delta)^4. \]
	From (5.1), we deduce that the entries of $\cD_x$ are computable in time,
	\[ O(\snec_d(T,\delta)^{3}\cdot \snec_1(T,\delta)^{ 2(\omega + 1)}\cdot  \log(\snec_d(T,\delta)) \cdot n^2). \]
	Since $T$ has $2 n-1$ nodes, the running time of our algorithm is $O(\snec_d(T,\delta)^{3}\cdot\snec_1(T,\delta)^{ 2(\omega + 1)}\cdot \log(\snec_d(T,\delta)) \cdot n^3)$.
\end{proof}

As a corollary, we can solve in time $\snec_1(T,\delta)^{(2\omega + 5)}\cdot  \log(\snec_1(T,\delta)) \cdot n^{3}$ the \textsc{Node-Weighted Steiner Tree} problem that asks, given a subset of vertices $K\subseteq V(G)$ called \emph{terminals}, a subset $T$ of minimal weight such that $K\subseteq T \subseteq V(G)$ and $G[T]$ is connected.

\begin{corollary}\label{col:steiner}
	There exists an algorithm that, given an $n$-vertex graph $G$, a subset $K\subseteq V(G)$, and  a rooted layout $(T,\delta)$ of $G$,  computes a minimum node-weighted Steiner tree for $(G,K)$ in time $O(\snec_1(T,\delta)^{(2\omega + 5)}\cdot\log(\snec_1(T,\delta)) \cdot n^{3})$.
\end{corollary}
\begin{proof}
	Observe that a Steiner tree is a minimum connected $(\bN,\bN)$-dominating set of $G$ that contains $K$.
	Thus, it is sufficient to change the definition of the table $\cA_x$ as follows.
	Let $x\in V(T)$. For all $(R,R')\in \Rep{V_x}{1}\times \Rep{\comp{V_x}}{1}$, we define $\cA_x[R,R']\subseteq V_x$    as follows:
	\[ \cA_x[R,R']:=\{ X\subseteq V_x \mid X\equi{V_x}{d} R \textrm{, } K\cap V_x \subseteq X \text{ and } X\cup R'\ (\bN,\bN)\text{-dominates } V_x \}.\]
	Notice that this modification will just modify the way we compute the table $\cD_x$ when $x$ is a leaf of $T$ associated with a vertex in $K$.
	With this definition of $\cA_x$ and by Definition \ref{def:represent} of $(x,R')$-representativity, if $G$ contains an optimal solution, then $\cD_r[\emptyset,\emptyset]$ contains an optimal solution of $G$.
	The running time comes from the running time of Theorem \ref{thm:dom} with $d=1$.
\end{proof}

Observe that Corollary \ref{col:steiner} simplifies and generalizes the algorithm from \cite{BodlaenderCKN15} for the \textsc{Edge-Weighted Steiner Tree} problem.
Indeed, the incidence graph of a graph of tree-width $k$ has tree-width at most $k+1$ and one can reduce the computation of an edge-weighted Steiner tree on a graph to the
computation of a node-weighted Steiner tree on its incidence graph \cite{GuhaK98}.

With few modifications, we can easily deduce an algorithm to compute a maximum (or minimum) connected Co-$(\sigma,\rho)$-dominating set.

\begin{corollary}\label{col:codom}
	There exists an algorithm that, given an $n$-vertex graph $G$ and a rooted layout $(T,\delta)$ of $G$, computes a maximum (or minimum) connected Co-$(\sigma,\rho)$-dominating set  in time $O(\snec_d(T,\delta)^{3}\cdot \snec_1(T,\delta)^{(2\omega + 5)}\cdot \log(\snec_d(T,\delta)) \cdot n^{3})$ with $d:=\max\{1, d(\sigma),d(\rho)\}$.
\end{corollary}
\begin{proof}
	To find a maximum (or minimum) Co-$(\sigma,\rho)$-dominating set, we need to modify the definition of the table $\cA_x$, the invariant, and the computational steps of the algorithm from Theorem \ref{thm:dom}.
	For each vertex $x\in V(T)$, we define the set of indices of our table $\cD_x$ as $\bI_x:=\Rep{V_x}{d}\times  \Rep{\comp{V_x}}{d} \times \Rep{V_x}{1}\times \Rep{\comp{V_x}}{1}$.
	
	For all $(R,R',\comp{R},\comp{R}')\in\bI_x$, we define $\cA_x[R,R',\comp{R},\comp{R}']\subseteq 2^{V_x}$ as the following set:
	{  \[\{ X\subseteq V_x \mid X\equi{V_x}{1} \comp{R}, \ (V_x\setminus X) \equi{V_x}{d} R \textrm{ and } (V_x\setminus X)\cup R' \ (\sigma,\rho) \textrm{-dominates } V_x  \}.\]}%
	It is worth noticing that the definition of $\cA_x$ does not depend on $\comp{R}'$; it is just more convenient to write the proof this way in order to obtain an algorithm similar to the one from Theorem \ref{thm:dom}.
	
	Similarly to Theorem \ref{thm:dom}, for each node $x$ of $V(T)$, our algorithm will compute a table $\cD_x$ that satisfies the following invariant.
	
	\subsection*{Invariant}
	For every $(R,R',\comp{R},\comp{R}')\in \bI_x$, the set $\cD_x[R,R',\comp{R},\comp{R}']$ is a subset of $\cA_x[R,R',\comp{R},\comp{R}']$ of size at most $\snec_1(T,\delta)^2$ that $(x,\comp{R}')$-represents $\cA_x[R,R',\comp{R},\comp{R}']$.

	Intuitively, we use $\comp{R}$ and $\comp{R}'$ to deal with the connectivity constraint of the Co-$(\sigma,\rho)$-dominating set and $R$ and $R'$ for the $(\sigma,\rho)$-domination.
	
	The following claim adapts Lemma \ref{lem:basicDymProg} to the Co-$(\sigma,\rho)$-dominating set case.
	\begin{claim}\label{claim:co-connected}
		Let $x$ be an internal node of $T$ with $a$ and $b$ as children.
		For all $(R,R',\comp{R},\comp{R}')\in \bI_x$, we have
		\[ \cA_x[R,R',\comp{R},\comp{R}']:=\!\!\! \bigcup_{\substack{(A,A') \text{ and } (B,B') \text{ are } d\textrm{-}(R,R')\textrm{-compatible}\\
				(\comp{A},\comp{A'}) \text{ and } (\comp{B},\comp{B'}) \text{ are }  1\textrm{-}(\comp{R},\comp{R}')\textrm{-compatible}}}
		\!\!\!\!\!\!\!  \cA_a[A,A',\comp{A},\comp{A'}] \otimes \cA_b[B,B',\comp{A},\comp{A'}],\]
		where the union is taken over all $(A,A',\comp{A},\comp{A}')\in\bI_a$ and $(B,B',\comp{B},\comp{B}')\in\bI_b$.
	\end{claim}
	
	The proof of this claim follows from the proof of Lemma \ref{lem:basicDymProg}.  With these modifications, it is straightforward to check that the algorithm of Theorem \ref{thm:dom} can be adapted to compute a minimum or maximum connected Co-$(\sigma,\rho)$-dominating set of $V(G)$.
	With the same analysis as in Theorem \ref{thm:dom}, one easily deduces that the running time of this modified algorithm is
	$O(\snec_d(T,\delta)^{3}\cdot \snec_1(T,\delta)^{(2\omega + 5)}\cdot \log(\snec_d(T,\delta)) \cdot n^{3})$.
\end{proof}

\section{Acyclic variants of {\sc(Connected)} \boldmath$(\sigma,\rho)$-{\sc Dominating Set}}\label{sec:maxinducedtree}

We call \textsc{AC-$(\sigma,\rho)$-Dominating Set} (resp., \textsc{Acyclic $(\sigma,\rho)$-Dominating Set}) the family of problems which consists in finding a subset
$X\subseteq V(G)$ of maximum (or minimum) weight such that $X$ is a $(\sigma,\rho)$-dominating set of $G$ and  $G[X]$ is a tree (resp., a forest).
Some examples of famous problems which belong to these families of problems are presented in Table \ref{tab:exampleACPb}.

\begin{table}[ht]
	\footnotesize
	\centering
	\renewcommand\arraystretch{1.5}
	\caption{Examples of \textsc{AC-$(\sigma,\rho)$-Dominating Set} problems and \textsc{Acyclic $(\sigma,\rho)$-Dominating Set} problems.
		To solve these problems, we use the $d$-neighbor equivalence with $d:=\max \{2,d(\sigma),d(\rho)\}$.
		Column $d$ shows the value of $d$ for each problem.}\label{tab:exampleACPb}
	{\small
		\begin{tabular}{c|c|c|c|c}
			\hline
			$\sigma$ & $\rho$ & $d$ & Version & Standard name  \\\hline
			$\bN$ & $\bN$ & 2 & \sc{AC} & \textsc{Maximum Induced Tree}  \\\hline
			$\bN$ & $\bN$ & 2  & \sc{Acyclic} & \textsc{Maximum Induced Forest}  \\ \hline
			$\{1,2\}$ & $\bN$ & 3 & \sc{AC} & \textsc{Longest Induced Path}  \\ \hline
			$\{1,2\}$ & $\bN$ & 3 & \sc{Acyclic} & \textsc{Maximum Induced Linear Forest}\\
			\hline
	\end{tabular} }
\end{table}

In this section, we present an algorithm that solves any \textsc{AC-$(\sigma,\rho)$-Dominating Set} problem.
Unfortunately, we were not able to obtain an algorithm whose running time is polynomial in $n$ and the $d$-neighbor-width of the given layout (for some constant $d$).
But, for the other parameters, by using their respective properties, we get the running times presented in Table \ref{tab:mit} which are roughly the same as those in the previous section.\
Moreover, we show, via a polynomial reduction, that we can use our algorithm for \textsc{AC-$(\sigma,\rho)$-Dominating Set} problems (with some modifications)  to solve any \textsc{Acyclic $(\sigma,\rho)$-Dominating Set} problem.

\begin{table}[ht]
	\footnotesize
	\centering
	\caption{Upper bounds on the running times of our algorithm for \textsc{AC-$(\sigma,\rho)$-Dominating Set }problems with $\cL=(T,\delta)$ and $d:=\max\{2, d(\sigma), d(\rho)\}$.}
	\label{tab:mit}
	{\small         \begin{tabular}{c|c}
			\hline
			Parameter & Running time  \\
			\hline
			Module-width & $O((d+1)^{3\mw(\cL)}\cdot 2^{ (2\omega + 3) \mw(\cL)} \cdot \mw(\cL)\cdot n^4)$ \\
			\hline
			Rank-width  & $O( 2^{(2\omega +3d + 4) \rw(\cL)^2}\cdot \rw(\cL)\cdot n^4 ) $ \\
			\hline
			$\mathbb{Q}$-rank-width  & $O( (d\cdot\Qrw(\cL))^{(2\omega + 5)\Qrw(\cL)}\cdot \Qrw(\cL)\cdot n^4 ) $  \\
			\hline
			Mim-width& $O(n^{(2\omega + 3d + 4  )  \mim(\cL)  +4}\cdot \mim(\cL))$ \\
			\hline
	\end{tabular}}
\end{table}

Let us first explain why we cannot use the same trick as in \cite{BodlaenderCKN15} on the algorithms of section \ref{sec:dom} to ensure the acyclicity, that is classifying the partial solutions $X$---associated with a node $x\in V(T)$---with respect to $|X|$ and $|E(G[X])|$.
Indeed, for two sets $X,W\subseteq V_x$ with $|X|=|W|$ and $|E(G[X])|=|E(G[W])|$, we have $|E(G[X\cup Y])| =|E(G[W\cup Y])|$, for all $Y\subseteq \comp{V_x}$, if and only if $X\equi{V_x}{n} W$.
Hence, the trick used in \cite{BodlaenderCKN15} would imply to classify the partial solutions with respect to their $n$-neighbor equivalence class.
But, the upper bounds we have on $\nec_n(V_x)$ with respect to module-width, ($\bQ$-)rank-width would lead to an \XP algorithm.
In fact, for every $k\in\bN$ and every $n\geq 2k$, one can construct an $n$-vertex bipartite graph $H_k[A,\comp{A}]$ where $\mw(A)=k$ and $\nec_n(A)=(n/\mw(A))^{\mw(A)}$
(see Figure \ref{fig:H_k} in Appendix \ref{appendix}).
Since both $\rw(A)$ and $\Qrw(A)$ are upper bounded by $\mw(A)$, we deduce that using the trick of \cite{BodlaenderCKN15} would give, for each $\f\in \{\mw,\rw,\Qrw\}$, an $n^{\Omega(\f(T,\delta))}$ time algorithm.

In the following, we introduce some new concepts that extend the framework designed in section \ref{sec:represents} in order to manage acyclicity.
All along, we give intuitions on these concepts through a concrete example: \textsc{Maximum Induced Tree}.
Finally, we present the algorithms for the \textsc{AC-$(\sigma,\rho)$-Dominating Set} problems and the algorithms for \textsc{Acyclic $(\sigma,\rho)$-Dominating Set} problems.

We start by defining a new notion of representativity to deal with the acyclicity constraint.
This new notion of representativity is defined \wrtn the 2-neighbor equivalence class of a set $R'\subseteq \comp{V_x}$.
We consider 2-neighbor equivalence classes instead of 1-neighbor equivalence classes in order to manage the acyclicity (see the following explanations).
Similarly to section \ref{sec:represents}, every concept introduced in this section is defined with respect to a node $x$ of $T$ and a set $R'\subseteq \comp{V_x}$.
To simplify this section, we fix a node $x$ of $T$ and $R'\subseteq \comp{V_x}$.
In our algorithm, $R'$ will always belong to $\Rep{\comp{V_x}}{d}$ for some $d\in\bN^+$ with $d\geq 2$.
For \textsc{Maximum Induced Tree} $d=2$ is enough and, in general, we use $d:=\max\{2,d(\sigma),d(\rho)\}$.

The following definition extends Definition \ref{def:represent} of section \ref{sec:represents} to deal with  the acyclicity.
We let $\opt \in \{\min, \max \}$; if we want to solve a maximization (or minimization) problem, we use $\opt=\max$ (or $\opt=\min$).

\begin{definition}[$(x,R')^\acy$-representativity]\label{def:acy-represent}
	For every $\cA \subseteq 2^{V(G)}$ and $Y\subseteq V(G)$, we define
	\[ \best(\cA,Y)^\acy:=\opt\{\w(X) \mid X\in \cA \text{ and } G[X\cup Y] \text{ is a tree}\}. \]
	
	Let $\cA,\cB\subseteq 2^{V_x}$. We say that $\cB$ $(x,R')^\acy$-represents $\cA$ if, for every $Y\subseteq \comp{V_x}$ such that $Y\equi{\comp{V_x}}{2} R'$, we have $\best^\acy(\cA,Y)=\best^\acy(\cB,Y)$.
\end{definition}

When $\cA=\emptyset$ or there is no $X\in \cA$ such that $G[X\cup Y]$ is a tree, we have $\best(\cA,Y)=\opt(\emptyset)$ and this equals $-\infty$ when $\opt=\max$ and $+\infty$ when $\opt=\min$.

As for the $(x,R')$-representativity, we need to prove that the operations we use in our algorithm preserve the $(x,R')^\acy$-representativity.
The following fact follows from Definition \ref{def:acy-represent} of $(x,R')^\acy$-representativity.

\begin{fact}\label{fact:acyunionpreserve}
	If $\cB$ and $\cD$ $(x,R')^\acy$-represents, respectively, $\cA$ and $\cC$, then $\cB\cup \cD$ $(x,R')^\acy$-represents $\cA\cup \cC$.
\end{fact}

The following lemma is an adaptation of Lemma \ref{lem:bigotimespreserve} to the notion of $(x,R')^\acy$-representativity.
The proof is almost the same as the one of Lemma \ref{lem:bigotimespreserve}.
We refer to Definition \ref{def:compatibility} for the notion of $d$-$(R,R')$-compatibility.

\begin{lemma}\label{lem:acybigotimespreserve} Let $d\in\bN^+$ such that $d\geq 2$.
	Suppose that $x$ is an internal node of $T$ with $a$ and $b$ as children.
	Let $R\in \Rep{V_x}{d}$. Let $(A,A')\in\Rep{V_a}{d}\times\Rep{\comp{V_a}}{d}$ and $(B,B')\in\Rep{V_b}{d}\times\Rep{\comp{V_b}}{d}$ that are $d$-$(R,R')$-compatible.
	Let $\cA\subseteq 2^{V_a}$ such that, for all $X\in \cA$, we have $X\equi{V_a}{d} A$ and let
	$\cB\subseteq 2^{V_b}$ such that, for all $W\in \cB$, we have $W\equi{V_b}{d} B$.
	
	If $\cA'\subseteq \cA$ $(a,A')^\acy$-represents $\cA$ and $\cB'\subseteq\cB$ $(b,B')^\acy$-represents $\cB$, then  \[ \cA' \otimes \cB' \ (x,R')^\acy\text{-represents } \cA \otimes \cB.  \]
\end{lemma}
\begin{proof}
	Suppose that $\cA'\subseteq \cA$ $(a,A')^\acy$-represents $\cA$ and $\cB'\subseteq\cB$ $(b,B')^\acy$-\break represents $\cB$.
	
	In order to prove this lemma, it is sufficient to prove that, for each $Y\equi{\comp{V_x}}{2} R'$, we have $\best^\acy(\cA'\otimes \cB',Y)=\best^\acy(\cA\otimes \cB,Y)$.
	
	Let $Y\subseteq \comp{V_x}$ such that $Y\equi{\comp{V_x}}{2} R'$.
	We claim the following facts: (a)~for every $W\in\cB$, we have $W\cup Y \equi{\comp{V_a}}{2} A'$, and~(b)~for every $X\in \cA$, we have $X\cup Y\equi{\comp{V_b}}{2} B'$.
	
	Let $W\in \cB$.
	By Fact \ref{fact:equivbiggerset}, we have that $W\equi{\comp{V_a}}{d} B$ because $V_b\subseteq \comp{V_a}$ and $W\equi{V_b}{d} B$.
	Since $d\geq 2$, we have $W\equi{\comp{V_a}}{2} B$.
	By Fact \ref{fact:equivbiggerset}, we deduce also that $Y\equi{\comp{V_a}}{2} R'$.
	Since $(A,A')$ and $(B,B')$ are $d$-$(R,R')$-compatible, we have $A'\equi{\comp{V_a}}{d} R'\cup B$.
	In particular, we have  $A'\equi{\comp{V_a}}{2} R'\cup B$ because $d\geq 2$.
	We can conclude that $W\cup Y\equi{\comp{V_a}}{2} A'$ for every $W\in\cB$.
	The proof for fact (b) is symmetric.
	
	Now observe that, by the definitions of $\best^\acy$ and of the merging operator $\otimes$, we have
	{\small
		\begin{align*}
			\best^\acy\big( \cA \otimes \cB, Y\big) &= \opt\{ \w(X)+ \w(W) \mid X\in \cA \wedge W\in \cB \wedge G[X\cup W\cup Y] \text{ is a tree} \}.
	\end{align*}}%
	Since $\best^\acy(\cA, W\cup Y)=\opt\{ \w(X) \mid X\in \cA \wedge G[X\cup W\cup Y] \text{ is a tree} \}$, we deduce that
	\begin{align*}
		\best^\acy\big( \cA \otimes \cB, Y\big) &= \opt\{ \best^\acy(\cA, W\cup Y) + \w(W) \mid W\in \cB \}.
	\end{align*}
	Since $\cA'$ $(a,A')$-represents $\cA$ and by fact (a), we have
	\begin{align*}
		\best^\acy\big( \cA \otimes \cB, Y\big) &= \opt\{ \best^\acy(\cA', W\cup Y) + \w(W) \mid W\in \cB \}\\
		&=\best^\acy\big( \cA' \otimes \cB, Y\big).
	\end{align*}
	Symmetrically, we deduce from fact (b) that $\best^\acy\big( \cA' \otimes \cB , Y\big)=\best^\acy\big( \cA' \otimes \cB', Y\big)$.
	This stands for every $Y\subseteq \comp{V_x}$ such that $Y\equi{\comp{V_x}}{2}  R'$.
	Thus, we conclude that $\cA' \otimes \cB'$ $(x,R')^\acy$-represents $\cA \otimes \cB$.
\end{proof}

In order to compute a maximum induced tree, we design an algorithm similar to those of section \ref{sec:dom}.  That is, for each
$(R,R')\in \Rep{V_x}{2}\times\Rep{\comp{V_x}}{2}$, our algorithm will compute a set $\cD_x[R,R']\subseteq 2^{V_x}$ that is an $(x,R')^\acy$-representative set of small
size of the set $\cA_x[R]:=\{X\subseteq V_x$ such that $X\equi{V_x}{2} R\}$.  This is sufficient to compute a maximum induced tree of $G$ since we have
$\cA_r[\emptyset]=2^{V(G)}$ with $r$ the root of $T$.  Thus, by Definition \ref{def:acy-represent}, any $(r,\emptyset)^\acy$-representative set of $\cA_r[\emptyset]$
contains a maximum induced tree.

The key to compute the tables of our algorithm is a function that, given $\cA\subseteq 2^{V_x}$, computes a small subset of $\cA$ that $(x,R')^\acy$-represents $\cA$.
This function starts by removing from $\cA$ some sets that will never give a tree with a set $Y\equi{\comp{V_x}}{2} R'$.
For doing so, we characterize the sets $X\in\cA$ such that $G[X\cup Y]$ is a tree for some $Y\equi{\comp{V_x}}{2} R'$.
We call these sets \emph{$R'$-important}.
The following gives a formal definition of these important and unimportant partial solutions.

\begin{definition}[$R'$-important]
	We say that $X\subseteq V_x$ is $R'$-important if there exists $Y\subseteq \comp{V_x}$ such that $Y\equi{\comp{V_x}}{2} R'$ and $G[X\cup Y]$ is a tree; otherwise, we say that $X$ is $R'$-unimportant.
\end{definition}

By definition, any set obtained from a set $\cA$ by removing $R'$-unimportant sets is an $(x,R')^\acy$-representative set of $\cA$.
The following lemma gives some necessary conditions on $R'$-important sets.
It follows that any set which does not respect one of these conditions can safely be removed from $\cA$.
These conditions are the key to obtain the running times of Table \ref{tab:mit}.

At this point, we need to introduce the following notations.
For every $X\subseteq V_x$, we define $X^0:=\{v\in X \mid N(v)\cap R' =\emptyset \}$,  $X^1:=\{v\in X \mid |N(v)\cap R'| = 1 \}$ and $X^{2+}:=\{v\in X \mid |N(v)\cap R' | \geq 2 \}$.
From the definition of the 2-neighbor equivalence, we have the following property, which is the key to manage acyclicity.
\begin{observation}\label{obs:X2+}
	For every  $X\subseteq V_x$ and $Y\subseteq \comp{V_x}$ such that $Y\equi{\comp{V_x}}{2} R'$, the vertices in $X^0$ have no neighbor in $Y$, those in $X^1$ have exactly one neighbor in $Y$, and those in $X^{2+}$ have at least $2$ neighbors in $Y$.
\end{observation}
\begin{lemma}\label{lem:important}
	If $X\subseteq V_x$ is $R'$-important, then $G[X]$ is a forest and the following properties are satisfied:
	\begin{enumerate}
		\item for every pair of distinct vertices $a$ and $b$ in $X^{2+}$, we have $N(a)\cap \comp{V_x}\neq N(b)\cap \comp{V_x}$,
		
		\item $|X^{2+}|$ is upper bounded by $2\mim(V_x)$, $2\rw(V_x)$, $2\Qrw(V_x)$, and $2\log_2(\nec_1(V_x))$.
	\end{enumerate}
\end{lemma}
\begin{proof}
	Obviously, any $R'$-important set must induce a forest.
	Let $X\subseteq V_x$ be an $R'$-important set.
	Since $X$ is $R'$-important, there exists $Y\subseteq \comp{V_x}$ such that $Y\equi{\comp{V_x}}{2} R'$ and $G[X\cup Y]$ is a tree.
	
	Assume toward a contradiction that there exist  two distinct vertices $a$ and $b$ in $X^{2+}$ such that $N(a)\cap \comp{V_x}= N(b)\cap \comp{V_x}$.
	Since $a$ and $b$ belong to $X^{2+}$ and $Y \equi{\comp{V_x}}{2} R'$, both $a$ and $b$ have at least two neighbors in $Y$ according to Observation~\ref{obs:X2+}.
	Thus, $a$ and $b$ have at least two common neighbors in $Y$.
	We conclude that $G[X\cup Y]$ admits a cycle of length four, yielding a contradiction.
	We conclude that property 1 holds for every $R'$-important set.
	
	Now, we prove that property 2 holds for $X$.
	Observe that, by Lemma \ref{lem:comparemim}, $\mim(V_x)$ is upper bounded by $\rw(V_x)$, $\Qrw(V_x)$, and $\log_2(\nec_1(V_x))$.
	Thus, in order to prove property 2, it is sufficient to prove that $|X^{2+}|\leq 2\mim(V_x)$.
	
	We claim that $|X^{2+}|\leq 2k$  where $k$ is the size of a maximum induced matching of $F:=G[X^{2+}, Y]$.
	Since $F$ is an induced subgraph of $G[V_x,\comp{V_x}]$, we have $k\leq \mim(V_x)$ and this is enough to prove property 2.
	Notice that $F$ is a forest because $F$ is a subgraph of $G[X\cup Y]$, which is a tree.

	In the following, we prove that $F$ admits a \emph{good bipartition}, that is, a bipartition $\{X_1,X_2\}$ of $X^{2+}\cap V(F)$ such that, for each $i\in\{1,2\}$ and for each $v\in X_i$, there exists $y_v\in Y\cap V(F)$ such that $N_{F}(y_v)\cap X_i = \{v\}$.
	Observe that this is enough to prove property 2 since if $F$ admits a good bipartition $\{X_1,X_2\}$, then $|X_1|\leq k$ and $|X_2|\leq k$.
	Indeed, if $F$ admits a good bipartition $\{X_1,X_2\}$, then, for each $i\in\{1,2\}$, the set of edges $M_i=\{ vy_v \mid v\in X_i \}$ is an induced matching of $F$.
	In order to prove that $F$ admits a good bipartition it is sufficient to prove that each connected component of $F$ admits a good bipartition.
	
	Let $C\in\cc(F)$ and $u\in C\cap X^{2+}$. As $F$ is a forest, $F[C]$ is a tree.
	Observe that the distance in $F$ between each vertex $v\in C\cap X^{2+}$ and $u$ is even because $F:=G[X^{2+},Y]$.
	Let $C_1$ (resp., $C_2$) be the set of all vertices $v\in C\cap X^{2+}$ such that there exists an odd (resp., even) integer $\ell\in \bN$ so that the distance between $v$ and $u$ in $F$ is $2\ell$.
	We claim that $\{C_1,C_2\}$ is a good bipartition of $F[C]$.
	
	Let $i \in \{1,2\}$, $v\in C_i$, and $\ell\in \bN$ such that the distance between $v$ and $u$ in $F$ is $2\ell$.
	Let $P$ be the set of vertices in $V(F)\setminus \{v\}$ that share a common neighbor with $v$ in $F$.
	We want to prove that there exists $y\in Y$ such that $N_{F}(y)\cap C_i = \{v\}$.
	For doing so, it is sufficient to prove that $N_F(v)\setminus N_F(C_i \setminus \{v\})=N_F(v)\setminus N_F(P\cap C_i)\neq \emptyset$.
	Observe that, for every $v'\in P$, the distance between $v'$ and $u$ in $F$ is either $2\ell-2$, $2\ell$ or $2\ell+2$ because $F[C]$ is a tree and the distance between $v$ and $u$ is $2\ell$.
	By construction of $\{C_1,C_2\}$, every vertex at distance $2\ell-2$ and $2\ell+2$ from $u$ must belong to $C_{3-i}$.
	Thus, every vertex in $P\cap C_i$ is at distance $2\ell$ from $u$.
	If $\ell=0$, then we are done because $v=u$ and $P\cap C_i=\emptyset$.
	Assume that $\ell\neq 0$.
	As $F[C]$ is a tree, $v$ has only one neighbor $w$ at distance $2\ell-1$ from $u$ in $F$.
	Because $F[C]$ is a tree, we deduce that $N_F(v)\cap N_F(P\cap C_i)=\{w\}$.
	Since $v\in X^{2+}$, $v$ has at least two neighbors in $F=G[X^{2+},Y]$ (because $Y\equi{\comp{V_x}}{2} R'$), we conclude that $N_F(v)\setminus N_F(P\cap C_i) \neq \emptyset$.
	Hence, we deduce that $\{C_1,C_2\}$ is a good bipartition of $F[C]$.
	
	We deduce that every connected component of $F$ admits a good bipartition and thus $F$ admits a good bipartition.
	Thus, $|X^{2+}|\leq 2\mim(V_x)$.
\end{proof}

These vertices in $X^{2+}$ play a major role in the acyclicity and the computation of representatives in the following sense.
By removing from $\cA$ the sets that do not respect the two above properties, we are able to decompose $\cA$ into a small number of sets $\cA_1,\dots,\cA_t$ such that an $(x,R')$-representative set of $\cA_i$ is an $(x,R')^\acy$-representative set of $\cA_i$ for each $i\in\{1,\dots,t\}$.
We find an $(x,R')^\acy$-representative set of $\cA$ by computing an $(x,R')$-representative set $\cB_i$ for each $\cA_i$ with the function $\reduce$.
This is sufficient because $\cB_1\cup \dots\cup \cB_t$ is an $(x,R')^\acy$-representative set of $\cA$ thanks to Fact \ref{fact:acyunionpreserve}.

The following definition characterizes the sets $\cA\subseteq 2^{V_x}$ for which an $(x,R')$-representative set is also an $(x,R')^\acy$-representative set.

\begin{definition}\label{def:consistent}
	We say that $\cA\subseteq 2^{V_x}$ is $R'$-consistent if, for each $Y\subseteq \comp{V_x}$ such that  $Y \equi{\comp{V_x}}{2} R'$ and there is $W\in
	\cA$ with $G[W\cup Y]$ a tree, we have
	\[ \{X\in \cA\mid G[X\cup Y]\ \textrm{is connected}\}=\{X\in \cA\mid G[X\cup Y]\ \textrm{is a tree}\}.\]
\end{definition}

The following lemma proves that an $(x,R')$-representative set of an $R'$-consistent  set $\cA$ is also an $(x,R')^\acy$-representative set of $\cA$.

\begin{lemma}\label{lem:consistent}
	Let $ \cA\subseteq 2^{V_x}$.
	For all $\cD\subseteq \cA$, if  $\cA$ is $R'$-consistent and $\cD$ $(x,R')$-represents $\cA$, then $\cD$ $(x,R')^\acy$-represents $\cA$.
\end{lemma}
\begin{proof}
	We assume that $\opt=\max$, the proof for $\opt=\min$ is similar. Let $Y\equi{\comp{V_x}}{2} R'$.
	If $\best^\acy(\cA,Y)=-\infty$, then we also have $\best^\acy(\cD,Y)=-\infty$ because $\cD\subseteq \cA$.
	
	Assume now that $\best^\acy(\cA,Y)\neq-\infty$. Thus, there exists $W\in \cA$ such that $G[W\cup Y]$ is a tree.
	Since $\cA$ is $R'$-consistent, for all $X\in\cA$, either the graph $G[X\cup Y]$ is a tree or it is not connected.
	Thus, by Definition \ref{def:represent} of $\best$, we have $\best^\acy(\cA,Y) = \best(\cA,Y)$.
	As $\cD\subseteq \cA$, we have also $\best^\acy(\cD,Y) = \best(\cD,Y)$.
	We conclude by observing that if $\cD$ $(x,R')$-represents $\cA$, then $\best^\acy(\cD,Y)=\best^\acy(\cA,Y)$.
\end{proof}

The next lemma proves that, for each $\f\in\{\mw,\rw,\Qrw,\mim\}$, we can decompose a set $\cA\subseteq 2^{V_x}$ into a small collection $\{\cA_1,\dots,\cA_t\}$ of  pairwise disjoint subsets of $\cA$ such that each $\cA_i$ is $R'$-consistent.
Even though some parts of the proof are specific to each parameter, the ideas are roughly the same.
First, we remove the sets $X$ in $\cA$ that do not induce a forest. If $\f=\mw$, we remove the sets in $\cA$ that do not respect condition 1 of Lemma \ref{lem:important}; otherwise, we remove the sets that do not respect the upper bound associated with $\f$ from condition 2 of Lemma \ref{lem:important}.
These sets can be safely removed as, by Lemma \ref{lem:important}, they are $R'$-unimportant.
After removing these sets, we obtain the decomposition of $\cA$ by taking the equivalence classes of some equivalence relation that is roughly the $n$-neighbor
equivalence relation.
Owing to the set of $R'$-unimportant sets we have removed from $\cA$, we prove that the number of equivalence classes of this latter equivalence relation respects the upper bound associated with $\f$ that is described in Table \ref{tab:cN}.

\begin{lemma}\label{lem:consistent-partition}
	Let $\cA\subseteq 2^{V_x}$. For each $\f\in \{\mw,\rw,\Qrw,\mim\}$, there exist pairwise disjoint subsets $\cA_1,\dots,\cA_t$ of $\cA$ computable in time $O(|\cA|\cdot \cN_\f(T,\delta)\cdot n^2)$  such that
	\begin{itemize}
		\item $\cA_1\cup\dots\cup \cA_t$ $(x,R')^\acy$-represents $\cA$,
		\item $\cA_i$ is $R'$-consistent for each $i\in\{1,\dots,t\}$, and
		\item $t\leq \cN_\f(T,\delta)$,
	\end{itemize}
	where $\cN_\f(T,\delta)$ is the term defined in Table \emph{\ref{tab:cN}}.
\end{lemma}

\begin{table}[tbh]
	\footnotesize
	\caption{Upper bounds $\cN_\f(T,\delta)$ on the size of the decomposition of Lemma {\rm \ref{lem:consistent-partition}} for each $\f\in \{\mw,\rw,\Qrw,\mim\}$.}
	\label{tab:cN}
	\centering
	{       \begin{tabular}{c|c|c|c|c}\hline
			$\f$  & $\mw$  & $\Qrw$ & $\rw$ & $\mim$  \\
			\hline
			$\cN_\f(T,\delta)$ &  $2^{\mw(T,\delta)}\cdot 2n$ &  $(2\Qrw(T,\delta) +1)^{\Qrw(T,\delta)}\cdot 2n$ &
			$2^{2 \rw(T,\delta)^2}\cdot 2n$ & $2 n^{2 \mim(T,\delta) +1}$\\ \hline
	\end{tabular}}
	
\end{table}

\begin{proof}
	We define the equivalence relation $\sim$ on $2^{V_x}$ such that $X\sim W$ if we have $X^{2+}\equi{V_x}{n} W^{2+}$ and
	$ |E(G[X])| - |X\setminus X^1|= |E(G[W])| - |W\setminus W^1|$.
	
	The following claim proves that an equivalence class of $\sim$ is an $R'$-consistent set.
	Intuitively, we use the property of the $n$-neighbor equivalence to prove that, for every $X,W\subseteq V_x$ such that $X\sim W$ and for all $Y\equi{\comp{V_x}}{2} R'$, we have $|E(G[X\cup Y])|= |X\cup Y| -1$ if and only if $|E(G[W\cup Y])|= |W\cup Y| -1$.
	Consequently, if $X\sim W$ and both sets induce with $Y$ a connected graph, then both sets induce with $Y$ a tree (because a graph $F$ is a tree if and only if $F$ is connected and $|E(F)|=|V(F)|-1$). We conclude from these observations the following claim.

	\begin{claim}\label{claim:equiconsistent}
		Let $\cB\subseteq\cA$. If, for all $X,W\in\cB$, we have $X\sim W$, then $\cB$ is $R'$-consistent.
	\end{claim}
	\begin{proof}
		Suppose that $X\sim W$ for all $X,W\in\cB$.
		In order to prove that $\cB$ is $R'$-consistent, it is enough to prove that, for each $X,W\in\cB$ and $Y\equi{\comp{V_x}}{2} R'$, if $G[X\cup Y]$ is a tree and $G[W\cup Y]$ is connected, then $G[W\cup Y]$ is a tree.
		
		Let $Y\equi{\comp{V_x}}{2} R'$ and $X,W\in\cB$.
		Assume that $G[X\cup Y]$ is a tree and that $G[W\cup Y]$ is connected.
		We want to prove that $G[W\cup Y]$ is a tree.
		
		Since $G[X\cup Y]$ is a tree, we have $|E(G[X\cup Y])| = |X\cup Y| -1$.
		Since the vertices in $X^0$ have no neighbors in $Y$, we can decompose $|E(G[X\cup Y])| = |X\cup Y| -1$ to obtain the following equality:
		\begin{align}\label{eq6.1}
			|E(G[Y])| + |E(X^{2+}, Y)| + |E(X^1, Y) | + |E(G[X])| =  |X\cup Y| -1.
		\end{align}
		Since every vertex in $X^1$ has exactly one neighbor in $Y$ (because $Y\equi{\comp{V_x}}{2} R'$), we deduce that $|E(X^1, Y) |= |X^1|$.
		Thus, (\ref{eq6.1}) is equivalent to
		\begin{align}\label{eq6.2}
			|E(G[Y])| + |E(X^{2+}, Y)| + |E(G[X])|  =  |X\setminus X^1| + |Y|  -1 .
		\end{align}
		Since $X\sim W$, we have $ |E(G[X])| - |X\setminus X^1|= |E(G[W])| - |W\setminus W^1|$.
		Moreover, due to $X^{^2+}\equi{V_x}{n} W^{2+}$ and Lemma \ref{lem:property_nec_n}, we have $|E(G(X^{2+}, Y))|=|E(G(W^{2+}, Y))|$.
		We conclude that (\ref{eq6.2}) is equivalent to
		\begin{align}\label{eq6.3}
			|E(G[Y])| + |E(W^{2+}, Y)| + |E(G[W])|  =  |W\setminus W^1 | +|Y|   -1 .
		\end{align}
		With the same arguments to prove that (\ref{eq6.3}) is equivalent to $|E(G[X\cup Y])| = |X\cup Y| -1$, we can show that (\ref{eq6.3}) is equivalent to $|E(G(W\cup Y))|=|W\cup Y | -1$.
		By assumption, $G[W\cup Y]$ is connected and thus we conclude that $G[W\cup Y]$ is a tree.
	\end{proof}
	
	We are now ready to decompose $\cA$. We start by removing from $\cA$ all the sets that do not induce a forest.
	Trivially, this can be done in time $O(|\cA|\cdot n)$.
	Moreover, these sets are $R'$-unimportant and thus we keep an $(x,R')^\acy$-representative set of $\cA$.
	Before explaining how we proceed separately for each parameter, we need the following observation, which follows from the removal of all the sets in $\cA$ that do not induce a forest.
	
	\begin{observation}\label{obs:nbedges}
		For all $X\in \cA$, we have  $|E(G[X])| - |X\setminus X^1| \in \{-n,\dots,n\}$.
	\end{observation}
	
	\subsection*{Concerning module-width}
	We remove all the sets $X$ in $\cA$ that do not respect condition 1 of Lemma \ref{lem:important}.
	By Lemma \ref{lem:important}, these sets are $R'$-unimportant and thus we keep an $(x,R')^\acy$-representative set of $\cA$.
	After removing these sets, for each $X\in\cA$, every pair $(a,b)$ of distinct vertices in $X^{2+}$ have a different neighborhood in $\comp{V_x}$.
	By definition of module-width, we have $\mw(V_x)=|\{ N(v)\cap \comp{V_x}\mid v\in V_x\}|$.
	Moreover, observe that, for every $X,W\in \cA$, if $\{N(v)\cap \comp{V_x} \mid v\in X^{2+}\}=\{N(v)\cap \comp{V_x} \mid v\in  W^{2+}\}$, then $X^{2+}\equi{V_x}{n} W^{2+}$.
	
	We deduce that the number of $n$-neighbor equivalence classes over the set $\{X^{2+}\mid X\in \cA \}$ is at most $2^{\mw(V_x)}$.
	Thus, the number of equivalence classes of $\sim$ over $\cA$  is at most $2^{\mw(V_x)} \cdot 2n \leq \cN_{\mw}(T,\delta)$.
	The factor $2n$ comes from Observation \ref{obs:nbedges} and appears also in all subsequent upper bounds.

	\subsection*{Concerning mim-width}
	We remove from $\cA$ all the sets $X$ such that $|X^{2+}|>2\mim(V_x)$.
	By Lemma \ref{lem:important}, these sets are $R'$-unimportant and thus we keep an $(x,R')^\acy$-representative set of $\cA$.
	Observe that this can be done in time $O(n^{\mim(V_x)+1} + |\cA|\cdot n^2)$ because $\mim(V_x)$ can be computed in time $O(n^{\mim(V_x)}+1)$.
	Since $|X^{2+}|\leq 2  \mim(V_x)$ for every $X\in\cA$, we deduce that the number of equivalence classes of $\sim$ over $\cA$ is at most $2 n^{2 \mim(V_x) +1}\leq\cN_{\mim}(T,\delta)$.

	\subsection*{Concerning rank-width}
	We remove from $\cA$ all the sets $X$ such that $|X^{2+}|> 2\rw(V_x)$ because they are $R'$-unimportant by Lemma \ref{lem:important}.
	By Lemma \ref{lem:comparen_sizet}, we know that $\nec_n^{\leq 2\rw(V_x)}(V_x)$ is upper bounded by $2^{2\rw(V_x)^2}$.
	We can therefore conclude that the number of equivalence classes of $\sim$ over $\cA$ is at most $2^{2 \rw(V_x)^2}\cdot 2n \leq\cN_\rw(T,\delta)$.

	\subsection*{Concerning \boldmath$\mathbb{Q}$-rank-width}
	We remove all the sets $X\in\cA$ such that $|X^{2+}|> 2\Qrw(V_x)$.
	By Lemma \ref{lem:important}, we keep an $(x,R')^\acy$-representative set of $\cA$.
	By Lemma \ref{lem:comparen_sizet}, we know that $\nec_n^{\leq 2\Qrw(V_x)}(V_x)$ is upper bounded by $(2\Qrw(V_x)+1)^{\Qrw(V_x)}$.
	We conclude that the number of equivalence classes of $\sim$ over $\cA$ is at most $(2\Qrw(V_x)+1)^{\Qrw(V_x)}\cdot 2n\leq\cN_{\Qrw}(T,\delta)$.
	
	It remains to prove the running time.
	Observe that, for module-width and\break ($\bQ$-)rank-width, the removal of $R'$-unimportant sets can be done in time $O(|\cA|\cdot n^2)$. Indeed, $\mw(V_x),\rw(V_x)$, and $\Qrw(V_x)$ can be computed in time $O(n^2)$.
	For 1-neighbor-width, we can assume that the size of $\nec_1(V_x)$ is given because the first step of our algorithm for \textsc{AC-$(\sigma,\rho)$-Dominating Set} problems is to compute $\Rep{V_x}{d}$ for some $d\in \bN^+$ and one can easily compute $\nec_1(V_x)$ while computing $\Rep{V_x}{d}$.
	Notice that we can decide whether $X\sim W$ in time $O(n^2)$.
	Therefore, for each $\f\in\{\mw,\rw,\Qrw,\mim\}$, we can compute the equivalence classes of $\cA$ in time $O(|\cA|\cdot \cN_\f(T,\delta) \cdot n^2)$.
\end{proof}

We are now ready to give an analogue of Theorem \ref{thm:reduce1} for the $(x,R')^\acy$-\break representativity.

\begin{theorem}\label{thm:reduceacy}
	Let $R\in\Rep{V_x}{2}$.
	For each $\f\in\{\mw,\rw,\Qrw,\mim\}$, there exists an algorithm $\reduce_\f^\acy$ that, given a set $\cA\subseteq 2^{V_x}$ such that $X\equi{V_x}{2} R$ for every $X\in\cA$, outputs in time $O(|\cA| \cdot  (\nec_1(V_x)^{2 (\omega -1)}+ \cN_\f(T,\delta)) \cdot n^2)$, a subset $\cB\subseteq \cA$ such that $\cB$ $(x,R')^\acy$-represents $\cA$ and  $|\cB| \leq \cN_\f(T,\delta) \cdot\nec_1(V_x)^{2}$.
\end{theorem}
\begin{proof}
	Let $\f\in\{\mw,\rw,\Qrw,\mim\}$.
	By Lemma \ref{lem:consistent-partition}, we can compute in time $O(|\cA|\cdot \cN_\f(T,\delta)\cdot n^2)$ a collection $\{\cA_1,\dots,\cA_t\}$ of pairwise disjoint subsets of $\cA$ such that
	\begin{itemize}
		\item $\cA_1\cup\dots\cup \cA_t$ $(x,R')^\acy$-represents $\cA$,
		\item $\cA_i$ is $R'$-consistent for each $i\in\{1,\dots,t\}$,
		\item $t\leq \cN_\f(T,\delta)$.
	\end{itemize}
	
	For each $X\in \cA$, we have $X\equi{V_x}{1} R$ because $X\equi{V_x}{2} R$.
	Since $\cA_1,\dots,\cA_t\subseteq \cA$, we can apply Theorem \ref{thm:reduce1} to compute, for each $i\in\{1,\ldots, t\}$, the set $\cB_i:=\reduce(\cA_i)$.
	By Theorem \ref{thm:reduce1}, for each $i\in \{1,\ldots, t\}$, the set $\cB_i$ is a subset and an $(x,R')$-representative set of $\cA_i$ whose size is bounded by $\nec_1(V_x)^2$.
	Moreover, as $\cA_i$ is $R'$-consistent, we have $\cB_i$ $(x,R')^\acy$-represents $\cA_i$ by Lemma \ref{lem:consistent}.
	
	Let $\cB:=\cB_1\cup \dots\cup \cB_t$.
	Since $\cA_1\cup \dots\cup \cA_t$ $(x,R')^\acy$-represents $\cA$, we deduce from Fact \ref{fact:acyunionpreserve} that $\cB$ $(x,R')^\acy$-represents $\cA$.
	Furthermore, we have $|\cB|\leq \cN_\f(T,\delta)\cdot \nec_1(V_x)^2$ owing to $t\leq \cN_\f(T,\delta)$ and $|\cB_i|\leq \nec_1(V_x)^2$ for all $i\in\{1,\dots,t\}$.
	
	It remains to prove the running time.  By Theorem \ref{thm:reduce1}, we can compute $\cB_1,\dots,\cB_t$ in time
	$O(|\cA_1\cup \dots\cup \cA_t|\cdot \nec_1(V_x)^{2 (\omega -1)} \cdot n^2)$.
	Since the sets $\cA_1,\dots,\cA_t$ are subsets of $\cA$ and pairwise disjoint, we have $|\cA_1\cup\dots\cup \cA_t| \leq |\cA|$.
	That proves the running time and concludes the theorem.
\end{proof}

We are now ready to present an algorithm that solves any \textsc{AC-$(\sigma,\rho)$-Dominating Set} problem.
This algorithm follows the same ideas as the algorithm from Theorem \ref{thm:dom}, except that we use $\reduce^\acy_\f$ instead
of $\reduce$.

\begin{theorem}\label{thm:maxinducedtree}
	For each $\f\in\{\mw,\rw,\Qrw,\mim\}$, there exists an algorithm that, given an $n$-vertex graph $G$ and a rooted layout $(T,\delta)$  of $G$, solves any \textsc{AC-$(\sigma,\rho)$-Dominating Set} problem, in time
	\[ O(\snec_d(T,\delta)^{3}\cdot\snec_1(T,\delta)^{2(\omega + 1 )}\cdot \cN_\f(T,\delta)^2 \cdot \log(\snec_d(T,\delta)) \cdot n^{3}), \]
	with $d:=\max\{2,d(\sigma),d(\rho)\}$.
\end{theorem}
\begin{proof}
	Let  $\f\in\{\mw,\rw,\Qrw,\mim\}$.
	If we want to compute a solution of maximum (resp., minimum) weight, then we use the framework of section \ref{sec:represents} with $\opt=\max$ (resp., $\opt=\min$).
	
	The first step of our algorithm is to compute, for each $x\in V(T)$, the sets $\Rep{V_x}{d}$, $\Rep{\comp{V_x}}{d}$ and a data structure to compute $\rep{V_x}{d}(X)$ and $\rep{\comp{V_x}}{d}(Y)$, for any $X\subseteq V_x$ and any $Y\subseteq \comp{V_x}$, in time $O(\log(\snec_d(T,\delta)) \cdot n^2 )$.
	As $T$ has $2n-1$ nodes, by Lemma \ref{lem:computenecd}, we can compute these sets and data structures in time $O(\snec_d(T,\delta)\cdot \log(\snec_d(T,\delta)) \cdot n^3)$.
	
	For each node $x\in T$ and, for each $(R,R')\in \Rep{V_x}{d}\times \Rep{\comp{V_x}}{d}$, we define $\cA_x[R,R']\subseteq 2^{V_x}$ as follows:
	\[ \cA_x[R,R']:=\{ X\subseteq V_x \mid X\equi{V_x}{d} R \text{ and } X\cup R' \text{ $(\sigma,\rho)$-dominates } V_x\}.\]

	We deduce the following claim from the proof of Claim \ref{lem:basicDymProg}.
	
	\begin{claim}\label{claim:basicprogdym}
		For every internal node $x\in V(T)$ with $a$ and $b$ as children and $(R,R')\in \Rep{V_x}{d}\times \Rep{\comp{V_x}}{d}$, we have
		\[ \cA_x[R,R']= \bigcup_{(A,A') \text{ and } (B,B') \text{ are } d\textrm{-}(R,R')\textrm{-compatible}}    \cA_a[A,A']\otimes \cA_b[B,B'],\]
		where the union is taken over all $(A,A')\in \Rep{V_a}{d}\times\Rep{\comp{V_a}}{d}$ and $(B,B')\in \Rep{V_b}{d}\times \Rep{\comp{V_b}}{d}$.
	\end{claim}
	
	For each node $x$ of $V(T)$, our algorithm will compute a table $\cD_x$ that satisfies the following invariant.
	
	\subsection*{Invariant}
	For every $(R, R')\in \Rep{V_x}{d}\times \Rep{\comp{V_x}}{d}$, the set $\cD_x[R,R']$ is a subset of $\cA_x[R,R']$ of size at most $\cN_\f(T,\delta)\cdot \nec_1(V_x)^{2}$ that $(x,R')^\acy$-represents $\cA_x[R,R']$.

	Notice that by Definition of $(x,R')^\acy$-representativity, if the invariant holds for $r$, then $\cD_r[\emptyset,\emptyset]$ contains a set $X$ of maximum (or minimum) weight such that $X$ is a $(\sigma,\rho)$-dominating set of $G$ and $G[X]$ is a tree.
	
	The algorithm is a usual bottom-up dynamic programming algorithm and computes for each node $x$ of $T$ the table $\cD_x$.
	
	Let $x$ be a leaf of $T$ with $V_x=\{v\}$.
	Observe that $\cA_x[R,R'] \subseteq 2^{V_x} = \{\emptyset,  \{ v \}\}$.
	Thus, our algorithm can directly compute $\cA_x[R,R']$ and set $\cD_x[R,R']:=\cA_x[R,R']$.
	In this case, the invariant trivially holds.
	
	Now, take $x$ an internal node of $T$ with $a$ and $b$ as children such that the invariant holds for $a$ and $b$.
	For each $(R,R')\in\Rep{V_x}{d}\times \Rep{\comp{V_x}}{d}$, the algorithm computes $\cD_x[R,R']:=\reduce^\acy_\f(\cB_x[R,R'])$, where the set $\cB_x[R,R']$ is defined as follows:
	\[ \cB_x[R,R']:= \bigcup_{(A,A') \text{ and } (B,B') \text{ are }  d\textrm{-}(R,R')\textrm{-compatible}} \cD_a[A,A']\otimes \cD_b[B,B'],\]
	where the union is taken over all $(A,A')\in \Rep{V_a}{d}\times\Rep{\comp{V_a}}{d}$ and $(B,B')\in \Rep{V_b}{d}\times \Rep{\comp{V_b}}{d}$.
	
	Similarly to the proof of Theorem \ref{thm:dom}, we deduce from Fact \ref{fact:acyunionpreserve}, Lemma \ref{lem:acybigotimespreserve}, Claim \ref{claim:basicprogdym}, and Theorem \ref{thm:reduceacy} that $\cD_x[R,R']$ is a subset and an $(x,R')^\acy$-representative set of $\cA_x[R,R']$.
	By Theorem \ref{thm:reduceacy}, we have $|\cD_x[R,R']|\leq \cN_\f(T,\delta)\cdot \snec_1(T,\delta)^2$.
	
	Consequently, the invariant holds for $x$, and by induction, it holds for all the nodes of $T$.
	The correctness of the algorithm follows.
	
	\subsection*{Running time}
	The running time of our algorithm is almost the same as the running time given in Theorem \ref{thm:maxinducedtree}.
	The only difference is the factor $\cN_\f(T,\delta)^2$ which is due to the following fact: by the invariant condition, for each $(A,A')$ and $(B,B')$, the size of $\cD_a[A,A']\otimes\cD_b[B,B']$ is at most $\cN_\f(T,\delta)^2\cdot\snec_1(T,\delta)^4$.
\end{proof}

By constructing for any graph $G$ a graph $G^\star$ such that the width measure of $G^\star$ is linear in the width measure of $G$ and any optimum acyclic
$(\sigma,\rho)$-dominating set of $G$ corresponds to an optimum \textsc{AC-$(\sigma,\rho)$}-dominating set of $G^\star$ and vice versa, we obtain the following which allows us,
for instance, to compute a feedback vertex set in time $n^{O(c)}$, $c$ the mim-width.

\begin{theorem}\label{thm:fvs}
	For each $\f\in\{\mw,\rw,\Qrw,\mim\}$, there exists an algorithm that, given an $n$-vertex graph  $G$ and  a rooted layout $(T,\delta)$ of $G$, solves any \textsc{Acyclic $(\sigma,\rho)$-Dominating Set} problem in time
	\[ O(\snec_d(T,\delta)^{3}\cdot\snec_1(T,\delta)^{2(\omega + 1 )}\cdot \cN_\f(T,\delta)^{O(1)} \cdot n^3) \]
	with $d:=\max\{2,d(\sigma),d(\rho)\}$.
\end{theorem}
\begin{proof}
	Let  $\f\in\{\mw,\rw,\Qrw,\mim\}$.
	Suppose that we want to compute a maximum acyclic $(\sigma,\rho)$-dominating set.
	The proof for computing a minimum  acyclic $(\sigma,\rho)$-dominating set is symmetric.

	The first step of this proof is to construct a $2n+1$-vertex graph $G^\star$ from $G$ and a layout $(T^\star,\delta^\star)$ of $G^\star$ from $(T,\delta)$ in time $O(n^2)$ such that $(T^\star,\delta^\star)$ respects the following inequalities:
	\begin{enumerate}
		\item for every $d\in \bN^+$, $\snec_d(T^\star,\delta^\star)\leq (d+1)\cdot\snec_d(T,\delta)$,
		\item for every $\f\in \{\mim, \mw, \rw,\Qrw\}$, $\f(T^\star,\delta^\star)\leq \f(T,\delta) + 1$.
	\end{enumerate}

	The second step of this proof consists in showing how the algorithm of Theorem \ref{thm:maxinducedtree} can be modified to find a maximum acyclic $(\sigma,\rho)$-dominating set of $G$ by running this modified algorithm on $G^\star$ and $(T^\star,\delta^\star)$.
	
	We construct $G^\star$ as follows. Let $\beta$ be a bijection from $V(G)$ to a set $V^+$ disjoint from $V(G)$.
	The vertex set of $G^\star$ is $V(G)\cup V^+ \cup \{v_0\} $ with $v_0$ a vertex distinct from those in $V(G)\cup V^+$.
	We extend the weight function $\w$ of $G$ to $G^\star$ such that the vertices of $V(G)$ have the same weight as in $G$ and the weight of the vertices in $V^+\cup \{v_0\}$ is 0.
	Finally, the edge set of $G^\star$ is defined as follows:
	\[ E(G^\star):= E(G)\cup \{ \{ v,\beta(v)\},\{ v_0, \beta(v)\} \mid v\in V(G) \} .\]
	
	We now construct $\sL=(T^\star,\delta^\star)$ from $\cL:=(T,\delta)$.
	We obtain $T^\star$ and $\delta^\star$ by doing the following transformations on $T$ and $\delta$:
	\begin{itemize}
		\item For each leaf $\ell$ of $T$ with $\delta(\ell)=\{v\}$, we transform $\ell$ into an internal node by adding two new nodes $a_\ell$ and $b_\ell$ as its children such that $\delta^\star(a_\ell)=v$ and $\delta^\star(b_\ell)=\beta(v)$.
		\item The root of $T^\star$ is a new node $r$ whose children are the root of $T$ and a new node $a_r$ with $\delta^\star(a_r)=v_0$.
	\end{itemize}
	In order to simplify the proof, we use the following notations.
	
	For each node $x\in V(T^\star)$, we let $\comp{V_x^\sL}:=V(G^\star)\setminus V_x^\sL$, and for each node $x\in V(T)$, we let $\comp{V_x^\cL}:=V(G)\setminus V_x^\cL$.

	Now, we prove that $(T^\star,\delta^\star)$ respects inequalities 1 and 2.
	Let $x$ be a node of $T^\star$.
	Observe that if $x\in V(T^\star)\setminus V(T)$, then the set $V_x^{\sL}$ either contains one vertex or equals $V(G^\star)$.
	Hence, in this case, the inequalities hold because $\nec_d(V_x^{\sL})\leq d$ for each $d\in\bN^+$ and $\f(V_x^{\sL})\leq 1$ for each $\f\in \{\mim, \mw, \rw,\Qrw\}$.
	
	Now, assume that $x$ is also a node of $T$.
	Hence, by construction, we have
	\begin{align*}
		V_{x}^{\sL} &=V_{x}^\cL \cup \{\beta(v) \mid v \in V_{x}^\cL \},\\
		\comp{V_x^\sL}&= \comp{V_x^\cL} \cup \{\beta(v) \mid v \in \comp{V_{x}^\cL} \} \cup \{v_0\}.
	\end{align*}
	Now, we prove inequality 1. Let $d\in\bN^+$.
	By construction of $G^\star$ and $\sL$, for each vertex $v\in V_{x}^{\cL}$, we have $\beta(v)\in V_{x}^{\sL}$ and
	\begin{align}
		N_{G^\star}(v)\cap \comp{V_{x}^{\sL}}&= N_{G}(v)\cap \comp{V_{x}^{\cL}}, \\
		N_{G^\star}(\beta(v))\cap \comp{V_{x}^{\sL}}&=\{v_0\}.
	\end{align}
	We deduce that, for every $X,Y\subseteq V_{x}^{\sL}$, we have $X\equi{V_{x}^{\sL}}{d} Y$ if and only
	\begin{itemize}
		\item $X\cap V(G) \equi{V_{x}^{\cL}}{d} Y\cap V(G)$ and
		\item $\min(d,|N(v_0)\cap X|)=\min(d,|N(v_0)\cap Y|)$.
	\end{itemize}
	Similarly, we deduce that, for every $X,Y\subseteq \comp{V_{x}^{\sL}}$, we have $X\equi{\comp{V_{x}^{\sL}}}{d} Y$ if and only if
	\begin{itemize}
		\item $X\cap V(G)\equi{\comp{V_{x}^{\cL}}}{d} Y\cap V(G)$ and
		\item $X\cap \{v_0\}= Y\cap \{v_0\}$.
	\end{itemize}
	Thus, we can conclude that $\snec_d(V_{x}^{\sL})\leq (d+1)\cdot \snec_d(V_{x}^{\cL})$.
	Consequently, inequality 1 holds.
	
	We deduce inequality 2 from Figure \ref{fig:matrix} describing the adjacency matrix between $V_{x}^{\sL}$ and $\comp{V_{x}^{\sL}}$.
	
	\begin{figure}[!pb]
		\centering
		\includegraphics[scale=0.95]{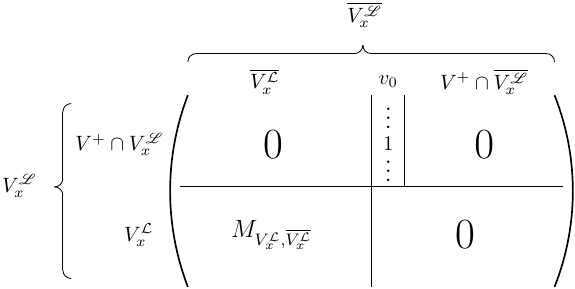}
		\caption{The adjacency matrix between $V_{x}^{\sL}$ and $\comp{V_{x}^{\sL}}$.}
		\label{fig:matrix}
	\end{figure}

	Now, we explain how we modify the algorithm of Theorem \ref{thm:maxinducedtree} in order to find a maximum acyclic $(\sigma,\rho)$-dominating set of $G$ by calling this algorithm on $G^\star$.
	For doing so, we modify the definition of the table $\cA_x$, the invariant, and the computational steps of the algorithm of Theorem \ref{thm:maxinducedtree}.
	The purpose of these modifications is to restrict the $(\sigma,\rho)$-domination to the vertices of $V(G)$.
	For doing so, we consider the set of nodes $S:= V(T)\cup \{r,a_r\}$.
	Observe that, for every node $x$ in $S$, there are no edges in $G[V_x^\sL, \comp{V_x^\sL}]$ between a vertex in $V(G)$ and a vertex in $V(G^\star)\setminus V(G)$.
	This is not true for the nodes of $V(T^\star)\setminus S$.
	For this reason, our algorithm ignores the nodes in $V(T^\star)\setminus S$ and computes a table only for the nodes in $S$.
	
	For every $x\in S$ and every $(R,R')\in \Rep{V_x^{\sL}}{d}\times \Rep{\comp{V_{x}^{\sL}}}{d}$ we define $\cA_x[R,R']\subseteq 2^{V_x^{\sL}}$ as follows:
	\begin{align*}
		\cA_x[R,R']&:= \{ X\subseteq V_x^{\sL} \mid X\equi{V_x^{\sL}}{d} R \text{ and } (X\cup R')\\ &\quad\cap V(G)\text{ $(\sigma,\rho)$-dominates } V_x^{\sL}\cap V(G)   \}.
	\end{align*}
	
	We claim that if $G$ admits an acyclic $(\sigma,\rho)$-dominating set $D$, then there exists $D'\in\cA_r[\emptyset,\emptyset]$ such that $D'\cap V(G)=D$ and $G^\star[D']$ is a tree.
	Let $D$ be an acyclic $(\sigma,\rho)$-dominating set of $G$ with $\cc(G[D])=\{C_1,\dots,C_t\}$.
	For each $i\in\{1,\dots,t\}$, let $v_i$ be a vertex in $C_i$.
	One easily checks that $G^\star[D\cup \{\beta(v_i) \mid 1\leq i\leq t\} \cup v_0]$ is a tree.
	Moreover, by definition of $\cA_r[\emptyset,\emptyset]$, for every $X\in \cA_r[\emptyset,\emptyset]$, if $G[X]$ is a tree, then $X\cap V(G)$ is an acyclic $(\sigma,\rho)$-dominating set of $G$.
	Hence, if $G$ admits an acyclic $(\sigma,\rho)$-dominating set, any $(r,\emptyset)^\acy$-representative set of $\cA_r[\emptyset,\emptyset]$ contains a set $X$ such that $X\cap V(G)$ is a maximum acyclic $(\sigma,\rho)$-dominating set of $G$.

	For every node $x\in S$, we compute a table $\cD_x$ satisfying the following invariant.
	
	\subsection*{Invariant} For each node $x\in S$ and each $(R,R')\in \Rep{V_x^{\sL}}{d}\times \Rep{\comp{V_x^{\sL}}}{d}$, the set $\cD_x[R,R']$ is a subset of $\cA_x[R,R']$ of size at most $\cN_\f(T^\star,\delta^\star)\cdot \nec_1(V_x^\sL)^{2}$ that $(x,R')^\acy$-represents $\cA_x[R,R']$.
	
	Before we explain how to compute the table $\cD_x$, for each $x\in S$, we need the following fact and claim.
	We deduce the following fact from Lemma \ref{lem:equivR'} and the fact that, for every node $x$ in $S$, there are no edges in $G[V_x^\sL, \comp{V_x^\sL}]$ between a vertex in $V(G)$ and a vertex in $V(G^\star)\setminus V(G)$.
	\begin{fact}\label{fact:FVSequivR'}
		Let $x\in S$.
		Let $X\subseteq V_x^\sL$ and $Y,R'\subseteq \comp{V_x^\sL}$ such that $Y\equi{\comp{V_x^\sL}}{d} R'$. Then $(X\cup R')\cap V(G)$
		$(\sigma,\rho)$-dominates $V_x^\sL\cap V(G)$ if and only if $( X\cup Y)\cap V(G)$ $(\sigma,\rho)$-dominates $V_x^\sL\cap V(G)$.
	\end{fact}
	
	We deduce the following claim from Fact \ref{fact:FVSequivR'} and Lemma \ref{lem:basicDymProg}.
	\begin{claim}\label{claim:basicDymProg}
		Let $x\in S\setminus \{a_r\}$ such that $x$ is not a leaf in $T$.
		Let $a$ and $b$ be the children of $x$ in $T^\star$.
		For every $(R,R')\in \Rep{V_x^{\sL}}{d}\times \Rep{\comp{V_x^{\sL}}}{d}$, we have
		\[ \cA_x[R,R'] = \bigcup_{(A,A') \text{ and } (B,B')  d\textrm{-}(R,R')\textrm{-compatible}} \cA_a[A,A']\otimes \cA_b[B,B'],\]
		where the union is taken over all $(A,A')\in \Rep{V_a^\sL}{d}\times\Rep{\comp{V_a^\sL}}{d}$ and $(B,B')\in \Rep{V_b^\sL}{d}\times \Rep{\comp{V_b^\sL}}{d}$.
	\end{claim}

	The algorithm starts by computing the table $\cD_x$ for each node $x\in S$ such that $x=a_r$ or $x$ is a leaf of $T$.
	Since $|V_x^\sL|\leq 2$, our algorithm directly computes $\cA_x[R,R']$ and sets $\cD_x[R,R']:=\cA_x[R,R']$ for every $(R,R')\in \Rep{V_x^{\sL}}{d}\times \Rep{\comp{V_x^{\sL}}}{d}$.
	
	For the other nodes our algorithm computes the table $\cD_x$ exactly as the algorithm of Theorem \ref{thm:maxinducedtree}.

	The correctness of this algorithm follows from Theorem \ref{thm:maxinducedtree} and Claim \ref{claim:basicDymProg}.
	By Theorem \ref{thm:maxinducedtree}, the running time of this algorithm is
	\[  O(\snec_2(\sL)^{3}\cdot\snec_1(\sL)^{2(\omega + 1 )}\cdot \cN_\f(\sL)^2 \cdot n^3).\]
	We deduce the running time in function of $\cL$ from inequalities 1 and 2.
\end{proof}

\section{{\sc Max Cut}}\label{sec:maxcut}

Prior to this work, the $d$-neighbor-equivalence relation was used only for problems with a locally checkable property like \textsc{Dominating Set} \cite{Bui-XuanTV13,GolovachHKKSV18,OumSV13}.
We prove in this paper that the $d$-neighbor-equivalence relation can also be useful for problems with a connectivity constraint and an acyclicity constraint.
Is this notion also useful for other kinds of problems?
Can we use it as a parameter to propose fast algorithms for problems which are $W[1]$-hard parameterized by clique-width, $\bQ$-rank-width, and rank-width such as \textsc{Hamiltonian Cycle}, \textsc{Edge Dominating Set}, and \textsc{Max Cut}?
The complexity of these problems parameterized by clique-width is well-known.
Indeed, for each of these problems, we have an ad hoc $n^{O(k)}$ time algorithm with $k$ the clique-width of a given $k$-expression \cite{BergougnouxKK20,FominGLS14}.
On the other hand, little is known concerning rank-width and $\bQ$-rank-width.
For mim-width, we know that \textsc{Hamiltonian Cycle} is para-\NP-hard parameterized by the mim-width of a given rooted layout \cite{JaffkeKT20a}.
As these problems are $\W[1]$-hard parameterized by clique-width, we cannot expect to rely only on the $d$-neighbor equivalence relation for $d$ a constant.

In this section, we prove that, given an $n$-vertex graph and a rooted layout $\cL$, we can use the $n$-neighbor equivalence to solve the \textsc{Max Cut} problem in time $\snec_n(\cL)^{O(1)}\cdot n^{O(1)}$.
The \textsc{Max Cut} problem asks, given a graph $G$, for the maximum $w\in\bN$ such that there exists a subset $X\subseteq V(G)$ with $w=|E(X,\comp{X  })|$.

We also prove that we can solve a connectivity variant of \textsc{Max Cut} called  \textsc{Maximum Minimal Cut} (a.k.a. \textsc{Largest Bond}) in time $\snec_n(\cL)^{O(1)}\cdot n^{O(1)}$.
A \emph{minimal cut} of a graph $G$ is a subset of vertices $X$ such that $G[X]$ and $G[\comp{X}]$ are connected.
The problem \textsc{Maximum Minimal Cut} asks, given a graph $G$, for the computation of a minimal cut $X\subseteq V(G)$ such that $|E(X,\comp{X})|$ is maximum.
The parameterized complexity of this problem was studied recently in \cite{DuarteLPSS19,EtoHKK19}, and in both papers, the authors proved that  \textsc{Maximum Minimal Cut} is solvable in time $n^{O(\cw)}$ with $\cw$ the clique-width of a given decomposition.

It is worth mentioning that to deal with the two connectivity constraints of \textsc{Maximum Minimal Cut}, we modify the framework developed in section \ref{sec:represents} and in particular the notion of representative set and the way we compute one.
For doing so, we use a nontrivial trick which could be easily generalized to deal with any fixed number of connectivity constraints.

By Corollary \ref{cor:necn}, our results imply that \textsc{Max Cut} and \textsc{Maximum Minimal Cut} are solvable in time $n^{O(\mw(G))}$, $n^{O(\Qrw(G))}$, and $n^{2^{O(\rw(G))}}$.

Let us give some explanations. Let $G$ be an $n$-vertex graph and $(T,\delta)$ a rooted layout of $G$.
Suppose that we want to solve \textsc{Max Cut}.
Let $x$ be a node of $T$ and $X,W\subseteq V_x$ such that $X\equi{V_x}{n} W$.
From Lemma \ref{lem:internal_node}, we can show that if the number of edges between $X$ and $V_x\setminus X$ is bigger than the number of edges between $W$ and $V_x \setminus W$, then $X$ is a better partial solution than $W$.
That is, for every $Y\subseteq \comp{V_x}$, the set $X\cup Y$ is a better solution than $W\cup Y$.
This follows from the fact that the number of edges of $E(V_x,\comp{V_x})$ between $X\cup Y$ and $\comp{X\cup Y}$ is the same as the number of edges of $E(V_x,\comp{V_x})$ between $W\cup Y$ and $\comp{W\cup Y}$.

It follows that it is enough to compute, for each node $x$ and each $R\in\Rep{V_x}{n}$, the maximum $k\in\bN$ such that $k=|E(X,V_x\setminus X)|$ for some $X\equi{V_x}{n} R$.
The integer $k$ computed for the root of $T$ and $R=\emptyset$ corresponds to the solution of \textsc{Max Cut}.

Before we present this algorithm, we need the following lemma, which we use to compute the tables associated with an internal node.

\begin{lemma}\label{lem:internal_node}
	Let $a$ and $b$ be the children of some internal node of $T$.
	Let $A\in \Rep{V_a}{n}$ and $B\in \Rep{V_b}{n}$.
	For every $X\subseteq V_a$ and $W\subseteq V_b$ such that $A\equi{V_a}{n} X$ and $B\equi{V_b}{n} W$, the number of edges between $X\cup W$ and $V_x\setminus(X\cup W)$ equals
	\[ |E(X,V_a\setminus X) | + |E(W,V_b\setminus W) | + |E(A,V_b\setminus B)| + |E(B,V_a\setminus A) |. \]
\end{lemma}
\begin{proof}
	Let $X\subseteq V_a$ and $W\subseteq V_b$ such that $A\equi{V_a}{n} X$ and $B\equi{V_b}{n} W$.
	First, observe that the number of edges between $X\cup W$ and $V_x\setminus(X\cup W)$ equals
	\[ |E(X,V_a\setminus X) | + |E(W,V_b\setminus W) | + |E(X,V_b\setminus W)| + |E(W,V_a\setminus X) | .\]
	So, we only need to prove that $|E(X,V_b\setminus W)|=|E(A,V_b\setminus B)|$ and $|E(W,V_a\setminus X) |= |E(B,V_a\setminus A) |$.
	As it is symmetrical, it is sufficient to prove the first equality.

	Since $B\equi{V_b}{n} W$, by Fact \ref{fact:necncomplement}, we have $V_b\setminus B \equi{V_b}{n} V_b \setminus W$.
	Thus, by Lemma \ref{lem:property_nec_n}, we have $|E(X,V_b\setminus W)|=|E(X,V_b\setminus B)|$.
	Moreover, as $A\equi{V_a}{n} X$ by applying Lemma \ref{lem:property_nec_n} again, we deduce that $|E(X,V_b\setminus W)|=|E(X,V_b\setminus B)|=|E(A,V_b\setminus B)|$.
\end{proof}

\begin{theorem}\label{thm:maxcut}
	There exists an algorithm that, given an $n$-vertex graph $G$ and a rooted layout $(T,\delta)$, solves \textsc{Max Cut} in time  $O(\snec_n(T,\delta)^{2}\cdot \log(\snec_n(T,\delta))\cdot n^3)$.
\end{theorem}
\begin{proof}
	The first step of our algorithm is to compute, for each $x\in V(T)$ and $X\subseteq V_x$, the sets $\Rep{V_x}{n}$ and a data structure to compute $\rep{V_x}{n}(X)$ in time $O(\log(\snec_n(T,\delta)) \cdot n^2 )$.
	As $T$ has $2n-1$ nodes, by Lemma \ref{lem:computenecd}, we can compute these sets and data structures in time $O(\snec_n(T,\delta)\cdot \log(\snec_n(T,\delta)) \cdot n^3)$.
	\smallskip
	
	For every node $x\in V(T)$ and every $R\in \Rep{V_x}{n}$, we define $\cT_x[R]$ as follows:
	\[ \cT_x[R]:=\max\{ |E(X,V_x\setminus X) | \mid X\subseteq V_x \text{ and } X\equi{V_x}{n} R \}. \]
	Observe that, for $r$ the root of $T$, the entry of $\cT_r[\emptyset]$ is the size of a maximum cut of $G$.
	
	The algorithm is a usual bottom-up dynamic programming algorithm and computes for each node $x$ of $T$ the table $\cT_x$.
	For the leaves $x$ of $T$, we simply set $\cT_x[R]:=0$ for every $R\in\Rep{V_x}{n}$.
	This is correct because the graph $G[V_x]$ contains only one vertex.
	
	Let $x$ be an internal node of $T$ with $a$ and $b$ as children.
	To compute the tables for the internal nodes we need the following claim.
	
	\begin{claim}\label{claim:internal_node}
		For every $R\in\Rep{V_x}{n}$, we have $\cT_x[R]$ equals
		\begin{align*}\max\{ \cT_a[A] \,{+}\, \cT_b[B] \,{+}\, |E(A,V_b\setminus B)| - |E(B,V_a\setminus A) | \,{ \mid}\,
			(A,B)\,{\in}\,\Rep{V_a}{n} \,{\times}\, \Rep{V_b}{n}\\ \text{ and } A\cup B\equi{V_x}{n} R\}.\end{align*}
	\end{claim}

	Claim \ref{claim:internal_node} is implied by the two following facts.
	
	\begin{fact}\label{fact:max_cut_oneway}
		Let $R\in\Rep{V_x}{n}$. For every $(A,B)\in\Rep{V_a}{n}\times\Rep{V_b}{n}$ such that $A\cup B\equi{V_x}{n} R$, we have
		\[ \cT_x[R]\geq \cT_a[A] + \cT_b[B]   + |E(A,V_b\setminus B)| + |E(B,V_a\setminus A)| . \]
	\end{fact}
	\begin{proof}
		Let $(A,B)\in\Rep{V_a}{n}\times\Rep{V_b}{n}$ such that $A\cup B\equi{V_x}{n} R$.
		By definition of $\cT_a[A]$, there exists $X_a\subseteq V_a$ such that $X\equi{V_a}{n} A$ and $\cT_a[A]=|E(X_a,V_a\setminus X_a)|$.
		Symmetrically, there exists $X_b\subseteq V_b$ such that $X_b\equi{V_b}{n} B$ and $\cT_b[B]=|E(X_b,V_b\setminus X_b)|$.
		
		From Lemma \ref{lem:internal_node}, the number of edges between $X_a\cup X_b$ and $V_x\setminus(X_a\cup X_b)$ equals $\cT_a[A]+\cT_b[B]  + |E(A,V_b\setminus B)| + |E(B,V_a\setminus A)| $.
		By Fact \ref{fact:equivbiggerset}, we deduce that $X_a\cup X_b\equi{V_x}{n} A\cup B$.
		Thus, $X_a\cup X_b\equi{V_x}{n} R$, and by definition, $\cT_x[R]$ is bigger than the number of edges between $X_a\cup X_b$ and $V_x\setminus(X_a\cup X_b)$. This proves the fact.
	\end{proof}
	\begin{fact}\label{fact:max_cut_secondway}
		For every $R\in\Rep{V_x}{n}$, there exists $(A,B)\in \Rep{V_a}{n}\times \Rep{V_b}{n}$ such that $A\cup B\equi{V_x}{n} R$ and
		\[ \cT_x[R]= \cT_a[A] + \cT_b[B] + |E(A,V_b\setminus B)| + |E(B,V_a\setminus A)|. \]
	\end{fact}
	\begin{proof}
		By definition of $\cT_x[R]$, there exists a set $X\subseteq V_x$ such that $X\equi{V_x}{n} R$ and $\cT_x[R]=|E(X,V_x\setminus V_x)|$.
		For every $i\in\{a,b\}$, let $X_i=X\cap V_i$.
		
		Let $A:=\rep{V_a}{n}(X_a)$ and let $B:=\rep{V_b}{n}(X_b)$.
		By definition, we have $X\equi{V_x}{n} A\cup B$, and by Claim \ref{fact:max_cut_oneway}, we have
		\begin{align}\label{eq7}
			\cT_x[R]\geq& \cT_a[A] + \cT_b[B] + |E(A,V_b\setminus B)| + |E(B,V_a\setminus A) |.
		\end{align}
		Moreover, as $\cT_x[R]=|E(X,V_x\setminus V_x)|$ and from Lemma \ref{lem:internal_node}, we deduce that
		\begin{align}\label{eq7.1}
			\cT_x[R]=& |E(X_a,V_a\setminus X_a)| + |E(X_b, V_b\setminus X_b)| + |E(A,V_b\setminus B)| + |E(B,V_a\setminus A) |.
		\end{align}
		By definition of $\cT_a[A]$, we know that $\cT_a[A]\geq |E(X_a,V_a\setminus X_a)|$.
		Symmetrically, we have $\cT_b[B]\geq |E(X_b,V_b\setminus X_b)|$.
		Hence, we conclude from inequality (\ref{eq7}) and (\ref{eq7.1}) that $\cT_x[R]= \cT_a[A] + \cT_b[B] + |E(V_a,V_b)| - |E(A,B)|$.
	\end{proof}
	
	We deduce that we can compute the entries of $\cT_x$ by doing the following:
	\begin{itemize}
		\item For every $R\in\Rep{V_x}{n}$, initialize some temporary variable  $w_R$ to 0.
		\item For every $(A,B)\in \Rep{V_a}{n}\times \Rep{V_b}{n}$, compute $R=\rep{V_x}{n}(A\cup B)$ and update $w_{R}$ as follows:
		\[ w_{R}:=\max\{w_{R}, \cT_a[A] + \cT_b[B] + |E(A,V_b\setminus B)| + |E(B,V_a\setminus A) |\}. \]
	\end{itemize}
	From Claim \ref{claim:internal_node}, at the end of this subroutine, we have correctly updated $\cT_x[R]$ which is precisely equal to $w_R$.
	Recall that each call to the functions $\rep{V_x}{n}$ takes $\log(\snec_n(T,\delta))\cdot n^2)$ time.
	Thus, the running time to compute the entries of $\cT_x$ is $O(\snec_n(T,\delta)^2\cdot \log(\snec_n(T,\delta))\cdot n^2)$.
	The total running time of our algorithm follows from the fact that $T$ has $2n-1$ nodes.
\end{proof}

\begin{theorem}\label{thm:maxMinCut}
	There exists an algorithm that, given an $n$-vertex graph $G$ and a rooted layout $(T,\delta)$, solves \textsc{Maximum Minimal Cut} in time  $O(\snec_n(T,\delta)^2\cdot \snec_1(T,\delta)^{4(\omega + 1.5)}\cdot n^{3})$.
\end{theorem}
\begin{proof}
	As we have two connectivity constraints, \ien, we ask both the solution and its complement to induce connected graphs, we have to modify our notion of representative set and how we compute one.
	Let $x\in V(T)$. For $\cA\subseteq 2^{V_x}$ and $Y\subseteq \comp{V_x}$, we define
	\[      \best^\star(\cA,Y):= \max\{ |E(X, V_x\setminus X)| \mid X\in \cA \text{ and } X\cup Y \text{ is a minimal cut} \}. \]
	
	For $R_Y,R_{\comp{Y}} \in \Rep{\comp{V_x}}{1}$, and $\cA,\cB\subseteq 2^{V_x}$, we say that $\cB$ \emph{$(x,R_Y,R_{\comp{Y}})$-represents} $\cA$ if for every $Y\subseteq \comp{V_x}$ such that $R_Y\equi{\comp{V_x}}{1} Y$ and $R_{\comp{Y}}\equi{\comp{V_x}}{1} \comp{V_x}\setminus Y$, we have $\best^\star(\cB,Y)=\best^\star(\cA,Y)$.
	
	Observe that, when $\cA=\emptyset$ or there is no $X$ in $\cA$ such that $X\cup Y$ is a minimal cut, then $\best^\star(\cA,Y)=-\infty$.
	
	For all $x\in V(T)$, we define $\bI_x:= \Rep{V_x}{n}\times\Rep{\comp{V_x}}{1}\times\Rep{\comp{V_x}}{1}$.
	Moreover, for each $R\in\Rep{V_x}{n}$, we define $\cA_x[R]:=\{X\subseteq V_x \mid X\equi{V_x}{n} R\}$.
	For every node $x\in V(T)$, our algorithm will compute a table $\cT_x$ satisfying the following invariant.

	\subsection*{Invariant}
	For every $(R,R_Y,R_{\comp{Y}})\in \bI_x$, the entry $\cT_x[R,R_Y,R_{\comp{Y}}]$ is a subset of $\cA_x[R]$ of size at most $\snec_1(T,\delta)^4$ that $(x,R_Y,R_{\comp{Y}})$-represents $\cA_x[R]$.

	The following claim proves that we can compute a small representative set for this new notion of representativity.
	
	\begin{claim}\label{claim:reduceMaxMinCut}
		Let $x\in V(T)$ and $(R,R_Y,R_{\comp{Y}})\in \bI_x$. There exists an algorithm $\reduce^\star$ that, given $\cA\subseteq 2^{V_x}$ such that, for all $X\in\cA$, we have $X\equi{V_x}{n} R$, outputs in time $O(|\cA|\cdot \snec_1(V_x)^{4(\omega-1)}\cdot n^2)$ a subset $\cB\subseteq \cA$ such that $\cB$ $(x,R_Y,R_{\comp{Y}})$-represents $\cA$ and $|\cB|\leq \snec_1(V_x)^4$.
	\end{claim}
	\begin{proof}
		Let $\cD$ be the set that contains all $Y\subseteq \comp{V_x}$ such that $Y\equi{\comp{V_x}}{1} R_Y$ and $\comp{V_x}\setminus Y \equi{\comp{V_x}}{1} R_{\comp{Y}}$.
		
		First assume that $R_Y=\emptyset$ or $R_{\comp{Y}}=\emptyset$.
		We distinguish the following cases:
		\begin{itemize}
			\item Suppose that $R_Y=\emptyset$ and $R_{\comp{Y}}\equi{\comp{V_x}}{1} \comp{V_x}$.  If $Y=\emptyset$, then $\best^\star(\cA,\emptyset)$ is
			the maximum $|E(X,\comp{X})|$ over of the sets $X\in\cA$ which are minimal cuts, and if no $X\in \cA$ is a minimal cut, then
			$\best(\cA,\emptyset)=-\infty$.  On the other hand, when $Y\neq \emptyset$, we have $\best^\star(\cA,Y)=0$ if $\emptyset\in \cA$ and $Y$ is a
			minimal cut or $-\infty$ otherwise.  Hence, it is sufficient to return a set $\cB$ constructed as follows: if $\cA$ contains a minimal cut, we
			add to $\cB$ a minimal cut $X$ in $\cA$ such that $|E(X,\comp{X})|$ is maximum and if $\emptyset\in \cA$, we add the empty set to $\cB$.
			
			\item If $R_Y\equi{\comp{V_x}}{1} \comp{V_x}$ and $R_{\comp{Y}}=\emptyset$, then return $\cB$ that contains---if there exists one---a set $X\in \cA$ such that $X\cup \comp{V_x}$ is a minimal cut and $|E(V_x\setminus X, \comp{V_x\setminus X})|$ is maximum and the empty set if $\emptyset \in \cA$. This case is symmetrical to the previous one.
			
			\item Otherwise, we return $\emptyset$. This is correct because in this case $\cD=\emptyset$. Indeed, for every $W\subseteq \comp{V_x}$, if $W\equi{\comp{V_x}}{1} \emptyset$, then $\comp{V_x}\setminus W \equi{\comp{V_x}}{1} \comp{V_x}$.
		\end{itemize}

		Assume from now on that $R_Y\neq \emptyset$ and $R_{\comp{Y}}\neq \emptyset$.
		We start by removing from $\cA$ the sets $X$ that satisfy at least one of the following properties:
		\begin{itemize}
			\item there exists $C\in\cc(G[X])$ such that $N(C)\cap R_Y= \emptyset$, or
			\item there exists $C\in\cc(G[V_x\setminus X])$ such that $N(C)\cap R_{\comp{Y}}= \emptyset$.
		\end{itemize}
		As explained in the proof of Theorem \ref{thm:reduce1}, we can safely remove all such sets as they never form a minimal cut with a set $Y\in \cD$.
		Moreover, we remove from $\cD$ the sets $Y$ that satisfy at least one of the following properties:
		\begin{itemize}
			\item there exists $C\in \cc(G[Y])$ such that $N(C)\cap R= \emptyset$, or
			
			\item there exists $C\in \cc(G[\comp{V_x}\setminus Y])$ with $N(C)\cap (V_x\setminus R) = \emptyset$.\enlargethispage{-12pt}
		\end{itemize}
		One can check that, for all sets $Y$ removed from $\cD$ and all $X\in \cA$, we have $\best^\star(\cA,Y)\break =-\infty$.
		In particular, if a set $Y$ satisfies the second property, then for every $X\in \cA$, the graph $G[\comp{X\cup Y}]$ is not connected.
		Indeed, by Fact \ref{fact:necncomplement}, we have $V_x\setminus X\equi{V_x}{n} V_x\setminus R$ for every $X\in \cA$, so the second property implies that a connected component of $\comp{V_x}\setminus Y$ has no neighbor in $V_x\setminus X$.
		
		For every $Y\in \cD$, we let $v_Y$ be one fixed vertex of $Y$.
		For every $R_1,R_2\in \Rep{\comp{V_x}}{1}$, $X\subseteq V_x$, and $Y\subseteq \comp{V_x}$, we define the following predicates:
		\begin{align*}
			&\cC(X,(R_1,R_2))\\&\quad:=
			\begin{cases} 1 & \textrm{if $\exists (X_1,X_2)\in \cuts(X)$ such that $N(X_1)\cap R_2=\emptyset$ and $N(X_2)\cap R_1 =\emptyset$,}\\
				0 & \textrm{otherwise}. \end{cases}\\
			&\comp{\cC}((R_1,R_2),Y)\\&\quad:=
			\begin{cases} 1 & \textrm{if $\exists (Y_1,Y_2)\in \cuts(Y)$ such that $v_Y \in Y_1$, $Y_1\equi{\comp{V_x}}{1}R_1$, and $Y_2 \equi{\comp{V_x}}{1}R_2$},\\
				0 & \textrm{otherwise}. \end{cases}
		\end{align*}
		These predicates are similar to the eponymous matrices in the proof of Theorem \ref{thm:reduce1}.
		In the following, we denote by $\mathfrak{R}$ the set of all quadruples of elements in $\Rep{\comp{V_x}}{1}$.
		Let $\cM^\star$, $\cC^\star$, and $\comp{\cC}^\star$ be, respectively, an $(\cA, \cD)$-matrix, an $(\cA,\mathfrak{R} )$-matrix, and an $(\mathfrak{R}, \cD)$-matrix such that
		\begin{align*}
			\cM^\star[X,Y] & :=
			\begin{cases} 1 & \textrm{if $X\cup Y$ is a minimal cut},\\
				0 & \textrm{otherwise}, \end{cases}\\
			\cC^\star[X,(R_Y^1,R_Y^2,R_{\comp{Y}}^1,R_{\comp{Y}}^2)]&:=
			\begin{cases} 1 & \textrm{if $\cC(X,(R_Y^1,R_Y^2))$ and $\cC(V_x\setminus X, (R_{\comp{Y}}^1,R_{\comp{Y}}^2))$ hold,}\\
				0 & \textrm{otherwise}, \end{cases}\\
			\comp{\cC}^\star[(R_Y^1,R_Y^2,R_{\comp{Y}}^1,R_{\comp{Y}}^2),Y]&:=
			\begin{cases} 1 & \textrm{if $\comp{\cC}((R_Y^1,R_Y^2),Y)$ and $\comp{\cC}((R_{\comp{Y}}^1,R_{\comp{Y}}^2),\comp{V_x}\setminus Y)$ hold},\\
				0 & \textrm{otherwise}. \end{cases}
		\end{align*}
		With the same arguments used in the proofs of Theorems \ref{thm:reduce1} and  \ref{thm:maxcut}, one can easily prove the following:
		\begin{itemize}
			\item For every $X\in\cA$ and $Y\in\cD$, we have
			\[ (\cC^\star\cdot \comp{\cC}^\star)[X,Y]= 2^{|\cc(G[X\cup Y])|-1}\cdot 2^{|\cc(G[\comp{X\cup Y}])|-1}. \]
			Consequently, $(\cC^\star\cdot \comp{\cC}^\star)[X,Y]$ is odd if and only if $X\cup Y$ is a minimal cut.
			Hence, $\cM^\star=_2 \cC^\star\cdot \comp{\cC}^\star$, where $=_2$ denotes the equality in $GF(2)$.

			\item The matrix $\cC^\star$ can be computed in time $O(|\cA|\cdot \snec_1(V_x)^4\cdot n^2)$.
			
			\item Let $\cB\subseteq \cA$ be a basis of maximum weight of the row space of $\cC^\star$ \wrtn to the weight function $\w(X):=|E(X,V_x\setminus
			X)|$.  Then, one can prove with the same arguments as in Theorem \ref{thm:reduce1} that $\cB$ $(x,R_Y,R_{\comp{Y}})$-represents $\cA$ and such a
			set can be computed in time $O(|\cA|\cdot \snec_1(V_x)^{4(\omega-1)}\cdot n^2)$.
		\end{itemize}
	\end{proof}
	
	As for the other algorithms of this paper, the computation of $\cT_x[R,R_Y,R_{\comp{Y}}]$ is trivial for the leaves $x$ of $T$.
	Take $x$ an internal node of $T$ with $a$ and $b$ as children such that the invariant holds for $a$ and $b$.
	For $(R,R_Y,R_{\comp{Y}})\in \bI_x$, $(A,A_Y,A_{\comp{Y}})\in\bI_a$, and $(B,B_Y,B_{\comp{Y}})\in\bI_b$, we say that $(A,A_Y,A_{\comp{Y}})$ and $(B,B_Y,B_{\comp{Y}})$ are $(R,R_Y,R_{\comp{Y}})$-\emph{compatible} if the following conditions are satisfied:
	\begin{itemize}
		\item $A\cup B \equi{V_x}{n} R$,
		\item $A\cup R_Y\equi{\comp{V_b}}{1} B_Y$ and $B\cup R_Y\equi{\comp{V_a}}{1} A_Y$,
		\item $(V_a\setminus A)\cup R_{\comp{Y}}\equi{\comp{V_b}}{1} B_{\comp{Y}}$ and $(V_b\setminus B)\cup R_{\comp{Y}} \equi{\comp{V_a}}{1} A_{\comp{Y}}$.
	\end{itemize}
	For each $(R,R_Y,R_{\comp{Y}})\in \bI_x$, the algorithm computes the entry $\cT_x[R,R_Y,R_{\comp{Y}}]:=\break\reduce^\star(\cB_x[R, R_Y,R_{\comp{Y}}])$ where $\cB_x[R,R_Y,R_{\comp{Y}}]$ is the set
	\bgroup
	\makeatletter
	\def\tagform@#1{\maketag@@@{\ignorespaces#1\unskip\@@italiccorr}}
	\makeatother
	{\def\theequation{{\normalsize(\thesection.\arabic{equation})}}\fontsize{8}{10}\selectfont{\begin{align*}
				\cB_x[R,R_Y,R_{\comp{Y}}]:=\!\!\!\!\bigcup_{(A,A_Y,A_{\comp{Y}}) \text{ and } (B,B_Y,B_{\comp{Y}}) \text{ are } (R,R_Y,R_{\comp{Y}})\text{-compatible}} \cT_a[A,A_Y,A_{\comp{Y}}]\otimes \cT_b[B,B_Y,B_{\comp{Y}}],
	\end{align*}}}\egroup
	where the union is taken over all $(A,A_Y,A_{\comp{Y}})\in \bI_a$ and $(B,B_Y,B_{\comp{Y}})\in \bI_b$.
	Observe that  $\cB[R,R_Y,R_{\comp{Y}}]$ is a valid input for  $\reduce^\star$ because, for all $X\in \cB[R,R_Y,R_{\comp{Y}}]$, we have $X\equi{V_x}{n} R$.
	Thus, the conditions of Claim \ref{claim:reduceMaxMinCut} on the input of $\reduce^\star$ are satisfied.
	
	Finally, with the same arguments used to prove Theorem \ref{thm:dom}, one can easily prove that the invariant holds for $x$ and the table $\cT_x$ can be computed in time
	\[ O(\snec_n(V_x)^2\cdot \snec_1(V_x)^{4(\omega + 1.5)}\cdot n^{2}). \]
	The total running time follows from the fact that $T$ has $2n-1$ nodes.
\end{proof}

\section{\boldmath$(A,C)$-$M$-{\sc partition} problem}\label{sec:problem}

In this section, we explain how the methods used here can be adapted to solve any locally checkable partitioning problem with connectivity and/or acyclicity
constraints.  Given a $(q\times q)$-matrix $M$ with entries being finite or co-finite subsets of $\bN$, an $M$-partition of a graph $G$ is a partition $\{V_1,\ldots, V_q\}$ of $V(G)$ such that for each $1\leq i \leq q$ and each $v\in V_i$ we have $|N(v)\cap V_j|\in M[i,j]$ for each $1\leq j\leq q$.

The \textsc{$M$-partition} problem asks, for a given graph $G$, whether $G$ admits an $M$-partition.
\NP-hard problems fitting into this framework include some homomorphism problems  \cite[Table~2]{Bui-XuanTV13} and the question of deciding whether an input graph has a partition into $q$ $(\sigma,\rho)$-dominating sets, which is in most cases NP-complete for small values of $q$ \cite[Table 1]{Bui-XuanTV13}.
It is proved in~\cite{Bui-XuanTV13} that the \textsc{$M$-partition} problem can be solved in time $O(\snec_d(T,\delta)^{3q}\cdot q\cdot n^4)$, with $d:=\max_{1\leq i,j\leq q}(d(M[i,j]))$, for any graph $G$ given with a rooted layout $(T,\delta)$. The algorithm is of the same flavor as for computing a $(\sigma,\rho)$-dominating set, except that, for each $x\in V(T)$, the table $\cT_x$ associated with $x$ is indexed by $(\Rep{V_x}{d})^q \times (\Rep{\comp{V_x}}{d})^q$ where $\cT_x[(R_1,\ldots,R_q),(R'_1,\ldots,R'_q)]$ is set to true if there is a partition $(S_1,\ldots,S_q)$ of $V_x$ such that for each $1\leq i,j \leq q$, $S_i\equi{V_x}{d} R_i$ and $S_i\cup R'_i$ $M[i,j]$-dominates $S_j$.

Let \textsc{$(A,C)$-$M$-partition} be the problem that asks, given a graph $G$ and two collections $A,C$ of subsets of $\{1,\dots,q\}$, whether $G$ admits an $M$-partition $\{V_1,\ldots, V_q\}$ such that, for every $\{i_1,\dots,i_t\}\in A$, the graph $G[V_{i_1}\cup\dots\cup V_{i_t}]$ is a forest and, for every $\{i_1,\dots,i_t\}\in C$,  the graph $G[V_{i_1}\cup\dots\cup V_{i_t}]$ is connected.
It is a routine verification to check that our framework can be extended to deal with \textsc{$(A,C)$-$M$-partition} problems:
\begin{itemize}
	\item The trick used in our algorithm for \textsc{Maximum Minimal Cut} to deal with two connectivity constraints can be generalized to solve any \textsc{$(\emptyset,C)$-$M$-partition}.
	
	\item We can generalize the concepts used in section \ref{sec:maxinducedtree} and in particular the decomposition into $R'$-consistent sets to solve any \textsc{$(A,C)$-$M$-partition} with $A\subseteq C$.
	
	\item We can use the idea of the reduction for the \textsc{Acyclic $(\sigma,\rho)$-Dominating Set} problems to solve any \textsc{$(A,C)$-$M$-partition} problem with the algorithm solving the special case when $A\subseteq C$.
\end{itemize}

We can therefore state the following without proof.

\begin{theorem}\label{thm:lcvpp} For a $(q\times q)$-matrix $M$ with entries being finite or co-finite subsets of $\bN$, let $d(M):=\max_{1\leq i,j\leq q}(d(M[i,j]))$.
	\begin{enumerate}
		\item There exists an algorithm that, given an $n$-vertex graph $G$ with a rooted layout $(T,\delta)$ of $G$, solves any \textsc{$(\emptyset,C)$-$M$-partition} problem in time $O(\snec_{d(M)}\break (T,\delta)^{3q}\cdot\snec_1(T,\delta)^{2(\omega + 1 )\cdot \ell} \cdot q \cdot n^{O(1)})$ with $\ell=|C|$.
		
		\item For each $\f\in\{\mw,\rw,\Qrw,\mim\}$, there exists an algorithm that, given an $n$-vertex graph $G$ and a rooted layout $(T,\delta)$  of $G$, solves any     \textsc{$(A,C)$-$M$-partition} problem, in time
		$O(\snec_{d(M)}(T,\delta)^{3q}\cdot\snec_1(T,\delta)^{O(\ell) }\cdot \cN_\f(T,\delta)^{O(\ell)} \cdot n^{O(1)})$ with $\ell=|A|+|C|$.
	\end{enumerate}
\end{theorem}

Observe that \textsc{$(A,C)$-$M$-partition} subsumes the problem \textsc{$\cF$-partition} of partitioning $V(G)$ into $q$ sets inducing graphs that belong to a
given graph class $\cF$, for several graph classes $\cF$ such as trees, forests, paths, and cycles.
For trees, forests or paths, Theorem~\ref{thm:lcvpp} provides an algorithm that, given a rooted layout $\cL$, solves \textsc{$\cF$-partition} and whose running time is upper bounded by $2^{O(\cw(\cL)q)}\cdot n^{O(q)}$, $2^{O(\Qrw(\cL)\log(\Qrw(\cL)q))}\cdot n^{O(q)}$, $2^{O(\rw(\cL)^2q)}\cdot n^{O(q)}$  and $n^{O(\mim(\cL)q)}$.

Observe that, for clique-width, $\bQ$-rank-width and rank-width, the factor $n^{O(q)}$ comes from the dependence in $n$ of the upper bounds on $\cN_\f(T,\delta)$ (see Table~5). 
Recently, Bergougnoux et al. \cite{BergougnouxDJ23} improve the techniques of this paper for handling acyclicity and remove the dependence in $n$ on these upper bounds for clique-width and rank-width (they ignore $\bQ$-rank-width in their analysis, but it should be doable for this parameter as well).
The authors of \cite{BergougnouxDJ23} obtain a model-checking algorithm for a logic which can express \textsc{$(A,C)$-$M$-partition}.
Their results imply that \textsc{$\cF$-partition}, for trees, paths or forests, can be solved in time $2^{O(\cw(\cL)q)}\cdot n^{O(1)}$ and $2^{O(\rw(\cL)^2q)}\cdot n^{O(1)}$.

For instance, with $A=\{\{1\},\dots,\{q\}\}$, $C=\emptyset$, and $M$ the matrix with all entries being $\bN$, a solution of \textsc{$(A,C)$-$M$-partition} is a
partition of a graph into $q$ induced forests. Such a partition problem is known as the \textit{tree-$q$-coloring} problem \cite{LiZ20}.  Li and Zhang
investigated the problem of finding the minimum $q$ such that the input graph has a tree-$q$-coloring and have proved in \cite{LiZ20} that this problem is FPT,
parameterized by the tree-width of the input graph. 
It would be interesting to study the parameterized complexity of this optimization problem with other
width measures. For instance, finding the minimum $q$ such that the input $n$-vertex graph has a tree-$q$-coloring can be solved in time $n^{f(k)}$ for some
function $f$, $k$ the clique-width of the input \cite[Theorem 9]{Rao07}. The problem seems to be $W[1]$-hard, parameterized by clique-width, what is the optimal
time complexity, under say (S)ETH?

With $A=C=\{\{1\},\dots,\{q\}\}$ and the same matrix $M$, a solution is a partition of a graph into $q$ induced trees.
The minimum number of induced trees covering the vertices of the graph is known as the \textit{tree cover number} from the algebra community~\cite{Barioli11}.
To the best of our knowledge, the computation of the tree cover number has not been studied.

Finally, observe that Theorem~\ref{thm:lcvpp} holds for the following variants of \textsc{$(A,C)$-$M$-partition}:
\begin{itemize}
	\item The optimization variant asking for a solution that maximizes any function $f(X_1,\dots,X_q)$ such that for every $X_1,\dots,X_q\subseteq V(G)$, we have
	\[ f(X_1,\dots,X_q)= \sum_{v\in V(G)} f(X_1\cap \{v\},\dots, X_q\cap \{v\}).  \] This property guarantees that the best partial solutions associated
	with $\mathcal{A}$ are those that can maximize $f$. Consequently, when we compute the basis of the matrix to obtain a representative,\footnote{The
		matrix that represents the join of partial solutions; see, for instance, Theorem \ref{thm:reduce1} and Claim~\ref{claim:reduceMaxMinCut}.} it is enough to take a basis of maximum weight,
	where the weight of a row $(X_1,\dots,X_q)$, representing one partial solution, is $f(X_1,\dots,X_q)$.  In particular, this property is satisfied by every function $f$ such that
	$f(X_1,\dots,X_q)=\w_1(X_1)+\w_2(X_2)+\dots+\w_q(X_q)$ with $\w_1,\dots,\w_q: V(G) \rightarrow \bQ$ being $q$ node-weight functions given in the
	input.
	
	\item The variant where the constraint ``$(X_1,\dots,X_q)$ is a partition of $V(G)$'' is replaced by a collection of constraints $A \bullet B$ with $\bullet\in\{ \subseteq, =\}$ and $A,B$ are two expressions based on the set operators $\cup,\cap$, set variables $X_1,\dots,X_q$, and constant sets $S_1,\dots,S_t$ given in the input. All the constraints in the collection must be satisfied by a solution.
	For any such constraint $\varphi$, one can prove that
	$\varphi(X_1,\dots,X_q,S_1,\dots,S_t)$ is true if and only if $\varphi(X_1\cap \{v\},\dots,X_q\cap\{v\},S_1\cap \{v\},\dots,S_t\cap\{v\})$ is true for every $v\in V(G)$.
	Thus, to handle these constraints, we just need to modify how the table is computed at the leaves of the decomposition, such as what we did for \textsc{Node-Weighted Steiner Tree} in Corollary~\ref{col:steiner}.
	That is, for every leaf of a decomposition associated with a vertex $v\in V(G)$, it is enough to consider only the partial solutions $(X_1,\dots,X_q)\in \{\emptyset,\{ v\}\}^q$ that satisfy all the constraints in the given collection.

	Many natural constraints can be expressed with this variant. For example, with the constraint $X_1\cup \dots \cup X_q= S_i$ and $S_i=V(G)$, we require that $(X_1,\dots,X_q)$ covers every vertex.
	We can also require that some pairs of $X_i$'s are disjoint with the constraint $X_i \cap X_j = S_i$ and given $S_i=\emptyset$ (notice that we can require that $(X_1,\dots,X_q)$ is partition with this variant).
	
\end{itemize}
We can combine the two variants to deal with some red-blue variants of problems such as \textsc{Connected Red-Blue Dominating Set} that asks, given a graph $G$, $R,B\subseteq
V(G)$, for a subset $X\subseteq B$ of minimum size such that $R\subseteq N(X)$ and $G[X]$ is connected. This problem was studied in~\cite{Abu-KhzamML11}
and  corresponds to the variant of \textsc{$(A,C)$-$M$-Partition} with $q=2$, with the constraints $X_1\subseteq B$ and $X_2=R$, $A=\emptyset$, $C=\{\{1\}\}$, $M$ the matrix where $M[2,1]=\bN^+$ (so that $X_2\subseteq N(X_1)$ for every solution $(X_1,X_2)$) and the other entries are $\bN$.
For the optimization, we require that a solution maximizes $f(X_1,X_2)=-|X_1|$.

\section{Conclusion}
\label{sec:conclusion}

This paper highlights the importance of the $d$-neighbor-\break equivalence relation in the design of algorithms parameterized by clique-width,\break ($\bQ$-)rank-width, and mim-width.
We prove that, surprisingly, this notion is helpful for problems with global constraints and also \textsc{Max Cut}: a $\W[1]$-hard problem parameterized by clique-width.
Can we use it for other $\W[1]$-hard problems parameterized by clique-width such as \textsc{Hamiltonian Cycle}, \textsc{Edge Dominating Set}, and \textsc{Graph Coloring}?

Recently, the $d$-neighbor-equivalence relation was used by the first author and collaborators in order to give a meta-algorithm for solving the
\textsc{Subset Feedback Vertex Set} problem and the \textsc{Node Multiway Cut} problem,  both problems in time $2^{O(\Qrw(G)^2\cdot\log(\Qrw(G)))}\cdot n^4$, $2^{O(\rw(G)^3)}\cdot n^4$, and $n^{O(\mim(\cL)^2)}$ with $G$ the input graph and $\cL$ a given decomposition \cite{BergougnouxPT20}.
It would be interesting to see whether the $d$-neighbor-equivalence relation can be helpful to deal with other kinds of constraints such as 2-connectivity and other generalizations of \textsc{Feedback Vertex Set} such as the ones studied in \cite{BonnetBKM19}.

Concerning mim-width, it is known that \textsc{Hamiltonian Cycle} is \NP-complete on graphs of mim-width 1, even when a rooted layout is provided \cite{JaffkeKT20a}.
\textsc{Graph Coloring} is \NP-complete on graphs of mim-width 2 because this latter is known to be \NP-complete on circular arc graphs \cite{GareyJMP80}.
Recently, \textsc{Max Cut} has been shown to be \NP-complete on interval graphs \cite{AdhikaryBMR20}. Consequently, \textsc{Max Cut} in \NP-hard on graphs of mim-width 1, even when a rooted layout is provided.

As explained in the introduction, the $2^{O(\mw(\cL))}\cdot n^{O(1)}$ time algorithms we obtain for clique-width are asymptotically optimal under \ETH.
This is also the case of our algorithms for \textsc{Max Cut} parameterized by clique-width and $\bQ$-rank-width.
Indeed, Fomin et al. \cite{FominGLS14} proved that there is no $n^{o(k)}\cdot f(k)$ time algorithm, $k$ being the clique-width of a given decomposition, for \textsc{Max Cut} unless \ETH fails.
Since the clique-width of a graph is always bigger than its $\bQ$-rank-width \cite{OumSV13}, this lower bound holds also for $\bQ$-rank-width.

However, for the other algorithmic results obtained in this paper, it is not known whether they are optimal under \ETH.
It would be particularly interesting to have tight upper bounds for rank-width since we know how to compute efficiently this parameter.
To the best of our knowledge, there is no algorithm parameterized by rank-width that is known to be optimal under \ETH.
Even for ``basic'' problems such as \textsc{Vertex Cover} or \textsc{Dominating Set}, the best algorithms \cite{Bui-XuanTV13} run in time $2^{O(k^2)}\cdot n^{O(1)}$, $k$ being the rank-width of the graph.
On the other hand, the best lower bounds state that, unless \ETH fails, there is no $2^{o(k)}\cdot n^{O(1)}$ time algorithm parameterized by rank-width for \textsc{Vertex Cover} (or \textsc{Dominating Set}) and no $n^{o(k)}\cdot f(k)$ time algorithm for \textsc{Max Cut} \cite{FominGLS14}.

Fomin et al. \cite{FominLPS16} have shown that we can use fast computation of representative sets in matroids\footnote{A matroid is a structure that abstracts and generalizes the notion of linear independence in vector spaces; see the book~\cite{Oxley06} for more information on matroids.} to obtain deterministic $2^{O(\tw(G))}\cdot n^{O(1)}$ time algorithms parameterized by tree-width for many connectivity problems.
Is this approach also generalizable with $d$-neighbor-width?
Can it be of any help for obtaining $2^{o(\rw(G)^2)}\cdot n^{O(1)}$ time algorithms for problems like \textsc{Vertex Cover} or \textsc{Dominating Set}?

\appendix

\section{Lower bounds on the \boldmath$n$-neighbor equivalence}\label{appendix}
The following theorem proves that the upper bounds of Lemma \ref{lem:comparen} on $\nec_n(A)$ are essentially tight.
In this theorem, we use a graph closely related to the one used in \cite[Theorem 4]{Bui-XuanTV11} to prove that, for any $k\in \bN$, there exists a graph $G$ and $A\subseteq V(G)$ such that $\rw(A)=k$ and $\nec_1(A)\in 2^{\Theta(k^2)}$.

\begin{theorem}\label{thm:cutnecn}
	For every $k\in \bN$, there exist an $n$-vertex graph $G_k$ and $A\subseteq V(G)$ such that $\rw(A)=k+1$, $\Qrw(A)=\mw(A)=2^k$, and $\nec_n(A)=(n/2^{k}-1)^{2^k}$.
\end{theorem}
\begin{proof}
	In the following, we will manipulate multiple graphs that have some vertices in common. To avoid any confusion, for every $\f\in \{\mw,\rw,\Qrw,\nec_n\}$, graph $G$, and $A\subseteq V(G)$, we will denote the value of $\f(A)$ in $G$  by $\f^G(A)$.
	For every $k\in \bN$, we define the sets $A_k:=\{a_S \mid S\subseteq [k]\}$, $B_k:=\{b_S\mid S \subseteq [k]\}$ and the graph $H_k:=(A_k\cup B_k, \{ \{ a_S,b_T\} \mid |S\cap T|\text{ is even}\})$.
	For instance, Figure \ref{fig:H_k} shows the graph $H_2$.
	
	For a matrix $M$, let us denote by $\rw(M)$ and $\Qrw(M)$ the rank of $M$ over $GF(2)$ and $\bQ$, respectively.
	For every $k\in \bN$, we denote by $M_k$ the adjacency matrix $M_{A_k,B_k}$.
	Notice that $\rw^{H_k}(A_k)=\rw(M_k)$ and $\Qrw^{H_k}(A_k)=\Qrw(M_k)$.
	We also consider the matrix $\comp{M}_k$ where for every $S,T\subseteq  [k]$, we have $\comp{M}[a_S,b_{T}]=1$ if $|S\cap T|$ is odd and 0 otherwise.
	Figure \ref{fig:H_k} shows the matrices $M_1$ and $\comp{M}_1$.
	
	\begin{figure}[!pb]
		\centering
		\includegraphics[width=0.8\columnwidth]{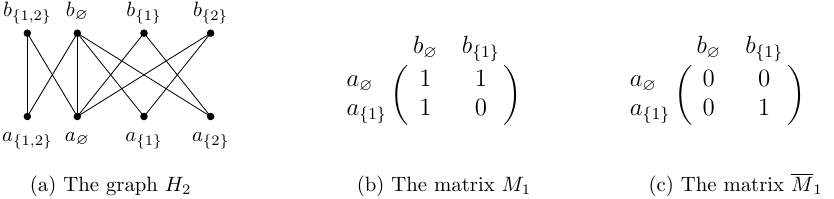}
		\caption{Graph $H_k$ with $k=2$ and the matrices $M_1$ and $\comp{M}_1$.}
		\label{fig:H_k}
	\end{figure}

	We want to prove by recurrence on $k$ that $\rw(M_k)=k+1$ and $\Qrw(M_k)=2^{k}$ for every $k\in\bN$.
	This is true for $k=0$ as we have $\rw(M_0)=\Qrw(M_0)=1$.
	
	Let $<$ be the total order on the subsets of $\bN^+$ where $S< T$ if $\max(S\triangle T) \in T$.
	For example, we have $\emptyset < \{1\}< \{2\} < \{1,2\} < \{3\} < \{1,3\} < \{2,3\} < \{1,2,3\}$.
	For every $k\in \bN$, we assume that the $i$th row (resp., column) of $M_k$ and $\comp{M}_k$ are associated, respectively, with $a_S$  (resp., $b_S$) where $S$ is the $i$th smallest subset of $[k]$ \wrtn $<$.
	This way, for every $k\in\bN$, the first $|A_{k+1}|/2=2^{k}$ rows (resp., columns) of $M_{k+1}$ and $\comp{M}_{k+1}$ are associated with the vertices $a_S$ (resp., $b_S$) that belong to $A_{k}$ (resp., $B_{k}$), \ien, we have $S\subseteq [k]$.
	On the other hand, the last $2^{k}$ rows (resp., columns) of  $M_k$ and $\comp{M}_k$  are associated with the vertices $a_S$ (resp., $b_S$) with $S=T\cup \{k+1\}$ with $T\subseteq [k]$.
	By definition, we deduce that, for every $k\in \bN$, we have
	\begin{equation}\label{eqA.1}
		M_{k+1}=
		\left(
		\begin{array}{c|c}
			M_{k} & M_{k} \\
			\hline
			M_{k} & \comp{M}_{k}
		\end{array}
		\right) \text{ and }
		\comp{M}_{k+1}=
		\left(
		\begin{array}{c|c}
			\comp{M}_{k} & \comp{M}_{k} \\
			\hline
			\comp{M}_{k} & M_{k}
		\end{array}
		\right).
	\end{equation}
	
	Since row and column additions do not change the rank of a matrix and from the definition of $M_k$ and $\comp{M}_k$, we deduce that, for every rank function $\f\in \{\rw,\Qrw\}$, we have
	\[
	\f(M_{k+1})=\f\left(
	\begin{array}{c|c}
		M_{k} & M_{k} \\
		\hline
		0 & \comp{M}_{k} - M_k
	\end{array}
	\right)=
	\f\left(
	\begin{array}{c|c}
		M_{k} & 0 \\
		\hline
		0 & \comp{M}_{k} - M_k
	\end{array}
	\right)= \f(M_{k}) + \f(\comp{M}_{k} - M_k).
	\]
	In $GF(2)$, the matrix $\comp{M}_{k} - M_k$ is an all-$1$'s matrix, we deduce that $\rw(M_{k+1})=\rw(M_k)+1$ for every $k\in \bN$.
	As $\rw(M_0)=1$, by recurrence, this implies that $\rw(M_k)=\rw^{H_k}(A_k)=k+1$.
	
	Now, let us prove that $\Qrw(M_k)=2^{k}$.
	From (\ref{eqA.1}), for every $k\in \bN$, we deduce that
	\begin{equation*} \comp{M}_{k+1}- M_{k+1}=
		\left(
		\begin{array}{c|c}
			\comp{M}_{k} - M_k & \comp{M}_{k} - M_k \\
			\hline
			\comp{M}_{k} - M_k & M_k - \comp{M}_{k}
		\end{array}
		\right). \end{equation*}
	Hence, the rank over $\bQ$ of the matrix $\comp{M}_{k+1}-M_{k+1}$ equals
	\begin{align}
		\Qrw\left(
		\begin{array}{c|c}
			\comp{M}_{k} - M_k & \comp{M}_{k} - M_k \\
			\hline
			0 & 2(M_k - \comp{M}_{k})
		\end{array}
		\right)&=
		\Qrw\left(
		\begin{array}{c|c}
			\comp{M}_{k} - M_k & 0 \\
			\hline
			0 & \comp{M}_{k} - M_k
		\end{array}
		\right)\\&=2\Qrw(\comp{M}_{k} - M_{k}).\notag
	\end{align}
	
	As $\Qrw(M_0)=\Qrw(\comp{M}_{0} - M_{0})=1$, by recurrence, we deduce that $\Qrw(\comp{M}_{k+1} - M_{k+1})= 2^{k}$ and thus $\Qrw(M_{k+1})= 2^k$.
	We conclude that $\Qrw^{H_k}(A_k)=2^{k}$.
	As $\Qrw^{H_k}(A_k)\leq \mw^{H_k}(A_k) \leq |A_k|=2^{k}$, we deduce that $\mw^{H_k}(A_k)=2^{k}$.

	Let $k,t\in \bN$ and $G_k$ be the graph obtained from $H_k$ by cloning $t$ times each vertex in $A_k$.
	That is, $V(G_k)=V(H_k)\cup \{a^2_S,\dots,a^t_S \mid S\subseteq [k]\}$, and for every $i\in \{2,\dots,t\}$ and $S\subseteq [k]$, we have $N(a_S^i)=N(a_S)$.

	Let $A_k^\star:=A_k\cup \{a^2_S,\dots,a^t_S \mid S\subseteq [k]\}$.
	Obviously, we have $\rw^{H_k}(A_k^\star)=k+1$, $\Qrw^{H_k}(A_k^\star)=\mw^{H_k}(A_k^\star)=2^{k}$ because the adjacency matrix between $A_k^\star$ and $B_k$ can be obtained from $M_k$ by adding $t-1$ copies of each row.
	
	Let $n$ be the number of vertices of $G$. We claim that $\nec_n^{H_k}(A_k^\star)=(n/2^{k}-1)^{2^k}$.
	For every $X\subseteq A_k^\star$, we define the column vector $v_X=(x_{\emptyset},x_{\{1\}},x_{\{2\}},x_{\{1,2\}},\dots,x_{[k]})$  such that, for every $S\subseteq [k]$, we have $x_S=|X\cap \{a_S,a_S^2,a_S^3,\dots,a_S^t\}|$.
	Observe that $M_k\cdot v_X=(y_{\emptyset},y_{\{1\}},\dots,y_{[k]})$ where, for every $S\subseteq [k]$, we have $y_S=|N(b_S)\cap X|$.
	Thus, for every $X,W\subseteq A^{\star}_{H_k}$, we have $X\equi{A^{\star}_{H_k}}{n} W$ if and only if $M_k\cdot v_X=M_k\cdot v_W$.
	As $\Qrw(M_k)=2^{k}$, by the rank-nullity problem, we deduce that the linear application associated with $M_k$ (over $\bQ$) is a bijection.
	We conclude that, for every $X,W\subseteq A^{\star}_{H_k}$, we have $X\equi{A^{\star}_{H_k}}{n} W$ if and only if $v_X=v_W$.
	Thus, $\nec_n^{H_k}(A_k^\star)=t^{2^k}$.
	By construction, we have $n=2^{k}(t+1)$.
	Consequently, $t=n/2^{k}-1$ and $\nec_n^{H_k}(A_k^\star)=(n/2^{k}-1)^{2^k}$.
\end{proof}

\end{document}